\documentclass{article}

\usepackage{microtype}
\usepackage{graphicx}
\usepackage{subfigure}
\usepackage{booktabs} 
\usepackage{hyperref}

\usepackage[accepted]{icml2024} 

\usepackage{amsmath}
\usepackage{amssymb}
\usepackage{mathtools}
\usepackage{amsthm}
\usepackage{dsfont} 
\usepackage{physics} 

\newcommand{\ind}{\mathds{1}}
\renewcommand{\P}{\mathbb{P}}

\renewcommand{\var}{\text{Var}}
\newcommand{\cov}{\,\text{Cov}}
\newcommand{\B}{\mathcal{B}}
\newcommand{\nprop}{n,\text{prop}}

\newcommand{\TV}{\text{TV}}
\newcommand{\Exp}{\text{Exp}}

\newcommand{\X}{\mathcal{X}}
\newcommand{\Y}{\mathcal{Y}}

\newcommand{\N}{\mathbb{N}}
\newcommand{\R}{\mathbb{R}}
\newcommand{\sph}{\mathbb{S}}
\renewcommand{\d}{\text{d}}
\newcommand{\epsi}{\varepsilon}
\newcommand{\diag}{\text{diag}}
\renewcommand\det{\,\text{det}}
\newcommand{\GL}{\,\text{GL}}
\newcommand{\Pc}{\mathcal{P}}

\newcommand{\A}{\mathcal{A}}

\newcommand{\F}{\mathcal{F}}

\newcommand{\M}{\mathcal{M}}

\newcommand{\Nc}{\mathcal{N}}
\newcommand{\U}{\mathcal{U}}

\renewcommand{\O}{\mathcal{O}}
\newcommand{\fatzero}{\mathbf{0}}

\newcommand{\ra}{\rightarrow}

\newcommand{\lora}{\longrightarrow}

\newcommand{\asconv}{\stackrel{\normalfont\text{a.s.}}{\lora}}

\newcommand\ccint[2]{\,[#1,#2]}
\newcommand\ooint[2]{\,]#1,#2[}

\newcommand\coint[2]{\,[#1,#2[}

\theoremstyle{plain}
\newtheorem{theorem}{Theorem}[section]

\newtheorem{lemma}[theorem]{Lemma}

\theoremstyle{definition}
\newtheorem{definition}[theorem]{Definition}

\newtheorem{example}[theorem]{Example}
\newtheorem{remark}[theorem]{Remark}

\icmltitlerunning{Parallel Affine Transformation Tuning of MCMC}

\begin{document}

\twocolumn[
\icmltitle{Parallel Affine Transformation Tuning of Markov Chain Monte Carlo}

\begin{icmlauthorlist}
	\icmlauthor{Philip Schär}{jena}
	\icmlauthor{Michael Habeck}{jena}
	\icmlauthor{Daniel Rudolf}{passau}
\end{icmlauthorlist}

\icmlaffiliation{jena}{Microscopic Image Analysis Group, Friedrich Schiller University Jena, Jena, Germany}
\icmlaffiliation{passau}{Faculty of Computer Science and Mathematics, University of Passau, Passau, Germany}

\icmlcorrespondingauthor{Daniel Rudolf}{daniel.rudolf@uni-passau.de}

\icmlkeywords{MCMC, adaptive MCMC, latent space, slice sampling}

\vskip 0.3in
]

\printAffiliationsAndNotice{}

\begin{abstract}
	The performance of Markov chain Monte Carlo samplers strongly depends on the properties of the target distribution such as its covariance structure, the location of its probability mass and its tail behavior. We explore the use of bijective affine transformations of the sample space to improve the properties of the target distribution and thereby the performance of samplers running in the transformed space. In particular, we propose a flexible and user-friendly scheme for adaptively learning the affine transformation during sampling. Moreover, the combination of our scheme with Gibbsian polar slice sampling is shown to produce samples of high quality at comparatively low computational cost in several settings based on real-world data.
\end{abstract}

\section{Introduction} \label{Sec:intro}

A variety of methods in probabilistic inference and machine learning relies on the ability to generate (approximate) samples from a high-dimensional probability distribution. But methods generating samples of high quality tend to be slow, so that the sampling can pose a severe computational bottleneck. Here we develop a new approach to address the general black-box sampling problem on $\R^d$ for arbitrary dimensions $d \in \N$.

Let us start with the concrete formulation of the aforementioned problem. Whenever random variables are introduced in the following, we implicitly require them to be defined on a sufficiently rich common probability space $(\Omega,\F,\P)$. The target distribution $\nu$ on $(\R^d,\B(\R^d))$ is given through a potentially unnormalized density ${\varrho: \R^d \ra \coint{0}{\infty}}$, meaning
\begin{equation}
	\nu(A)
	= \frac{\int_A \varrho(x) \d x}{\int_{\R^d} \varrho(x) \d x} ,
	\qquad A \in \B(\R^d) .
	\label{Eq:def_nu}
\end{equation}
We may evaluate $\varrho$ at any given point, but we generally cannot evaluate $\nu$ and do not even know the normalization constant of $\varrho$ (the denominator in \eqref{Eq:def_nu}). The task is to sample -- approximately -- from $\nu$, i.e.~to generate realizations $x_0,x_1,\dots \in \R^d$ of random variables $X_0,X_1,\dots$ on $(\R^d, \B(\R^d))$ such that the distribution $\P^{X_n} := \P \circ X_n^{-1}$ of $X_n$ is close to $\nu$ in some sense, at least for large enough $n$.

A powerful approach for solving the sampling problem is \textit{Markov chain Monte Carlo} (MCMC). An MCMC sampler generates a truncated realization $x_0, \dots, x_m$ of a Markov chain $(X_n)_{n \in \N_0}$ with some initial distribution $\xi$ and a transition kernel\footnote{For basic knowledge regarding the definitions and properties of transition kernels see e.g.~\citet{RudolfDiss}.} $P$ on $\R^d \times \B(\R^d)$ that leaves $\nu$ invariant. For commonly used MCMC methods, there are theoretical results ensuring that $\P^{X_n}$ converges to $\nu$, meaning
\begin{equation*}
	\TV(\P^{X_n}, \nu)
	= \TV(\xi P^n, \nu)
	\stackrel{n \ra \infty}{\lora} 0 ,
\end{equation*}
under weak assumptions on $\varrho$ and $\xi$, where $\TV$ denotes the total variation distance.

In real-world applications of MCMC methods, information about essential properties of the target distribution, such as the covariance structure or the number, location and relative importance of modes, is often not available. Therefore, MCMC users typically resort to generic off-the-shelf samplers like \textit{random walk Metropolis} (RWM) \cite{Metropolis}, the \textit{Metropolis-adjusted Langevin algorithm} (MALA) \cite{MALA}, \textit{Hamiltonian Monte Carlo} (HMC) \cite{BayesianNeuralNets} or \textit{hybrid uniform slice sampling} \cite{SSNeal}. Although these methods are in principle able to address sampling problems with very little a priori information about the target, their performance is often suboptimal, especially in high dimensions, resulting in slow convergence and long-lasting autocorrelations. If some information about the target is available, ``informed'' MCMC methods that can explicitly incorporate this knowledge into their transition mechanisms often vastly outperform the generic uninformed methods. Instances of informed MCMC samplers include \textit{independent Metropolis-Hastings} (IMH) \cite{Tierney}, the \textit{preconditioned Crank-Nicolson Metropolis-Hastings algorithm} (pCN-MH) \cite{MCMC_fct}, \textit{elliptical slice sampling} \cite{EllipticalSS} (ESS) and \textit{Gibbsian polar slice sampling} (GPSS)\footnote{As GPSS does not have a user-specified proposal distribution, it is not immediately obvious how it can make use of information about the target. We will clarify this later on.} \cite{GPSS}.

Related to this, \textit{adaptive MCMC} \cite{ExAda} provides a general framework to iteratively acquire knowledge about the target distribution during sampling, and incorporate this information into the transition mechanism. At any iteration the transition mechanism may change depending on the whole ``past'', for instance transitioning from $X_{n-1}$ to $X_n$ can be performed by a kernel $P_n$, where $P_n$ may depend on $X_0,\dots,X_{n-1}$. Numerical experiments indicate that this may improve the sample quality.

In this paper, we consider \textit{affine transformation tuning} (ATT), a general principle of performing MCMC transitions via a latent space that is connected to the sample space by a bijective affine transformation. In particular, we discuss how to construct adaptive MCMC methods based on ATT, where the adaptivity is used to iteratively learn a ``good'' affine transformation. When learning the affine transformation, say $\alpha(y) = W y + c$, the goal is to bring the pushforward of the target distribution $\nu$ under the inverse transformation $\alpha^{-1}(x) = W^{-1}(x - c)$ as close as possible to \textit{isotropic position} (i.e.~mean zero and identity covariance, or a suitable analogue in very heavy-tailed settings). We can then exploit the approximate isotropy of the transformed target distribution to enhance the performance of informed MCMC samplers.

Our framework is very general and can be applied to a large variety of informed MCMC methods, such as IMH, pCN-MH, ESS or GPSS. Nonetheless, we focus on applying ATT to GPSS. This is motivated by the fact that GPSS is able to effectively acquire information about the target distribution at early stages of a run (unlike IMH, which typically performs poorly until it has acquired a substantial amount of information about the target), and known to work well for various types of targets in isotropic position, including both light-tailed and heavy-tailed ones (unlike pCN-MH and ESS, which are designed for settings with Gaussian priors).
As the application of an affine transformation $\alpha$ to the target distribution $\nu$ does not change the tail behavior (e.g.~Gaussian tails will stay Gaussian, heavy tails will stay heavy), this makes GPSS uniquely well-suited for ATT. We refer to Appendix \ref{SubApp:exp_methods} for an explanation of GPSS.

The empirical evidence for the good performance of GPSS in isotropic settings is complemented by theoretical results \cite{PSS_paper,kPSS_paper} regarding the underlying ``ideal'' method \textit{polar slice sampling} (PSS) \cite{PolarSS}, in conjunction with a theoretical result \cite{GSSS_contraction} regarding \textit{geodesic slice sampling on the sphere} (GSSS) \cite{SphericalSS}. Essentially, \citet{GSSS_contraction} demonstrated that the use of GSSS within the GPSS transition does not prevent GPSS from achieving the dimension-independent performance of PSS established in the two former works.

The remainder of our paper is structured as follows.
Section \ref{Sec:ATT} introduces ATT as a general principle for MCMC transitions. In Section \ref{Sec:par_sched}, we propose parallel ATT (PATT), a flexible framework for setting up samplers that run multiple parallel ATT chains and let them share information to more quickly learn a suitable transformation. In Section \ref{Sec:theo_just}, we present a theoretical result that justifies the use of PATT. The results of several numerical experiments with PATT are presented in Section \ref{Sec:experiments}. We conclude the paper's main body with some final remarks in Section \ref{Sec:conclusion}.

In addition, we offer complementary information on ATT and related topics in the supplementary material:
Appendices \ref{App:adj_types}, \ref{App:param_choices} and \ref{App:schedule} provide detailed considerations and guidelines regarding the choices of the adjustment types, transformation parameters and update schedules defined in Sections \ref{Sec:ATT} and \ref{Sec:par_sched}. Theses appendices may serve as a ``cookbook'' for implementing and/or applying ATT or PATT.
In place of a related work section, we give a detailed overview of connections between our method and various others in Appendix \ref{App:connections}.
In Appendix \ref{App:equiv_trad_ada}, we prove that in certain cases a simple adaptive MCMC implementation of ATT is equivalent to other, more traditional adaptive MCMC methods, in that the respective transition kernels coincide.
The proof of our theoretical result from Section \ref{Sec:theo_just} is provided in Appendix \ref{App:theo_just_proof}.
In Appendix \ref{App:exp_details} we elaborate on the models and hyperparameter choices for the experiments behind the results presented in Section \ref{Sec:experiments}, and provide some further results.
Appendix \ref{App:ablation} presents a series of ablation studies demonstrating that each non-essential component of PATT can, in principle, substantially improve its performance. Appendix \ref{App:plots} offers more plots illustrating the main experiments as well as the ablation studies.

\section{Affine Transformation Tuning} \label{Sec:ATT}

\subsection{Formal Description} \label{SubSec:formal_desc}

For the sake of a systematic formalization, let us name two copies of $\R^d$, the \textit{sample space} $\X := \R^d$ and the \textit{latent space} $\Y := \R^d$. Suppose we have an affine transformation
\begin{equation*}
	\alpha: \Y \ra \X, \; y \mapsto W y + c ,
\end{equation*}
with a shift $c \in \R^d$ and an element $W \in \GL_d(\R)$ of the \textit{general linear group} (over $\R$ in dimension $d$),
\begin{equation*}
	\GL_d(\R)
	:= \{W \in \R^{d\times d} \mid \det(W) \neq 0\} .
\end{equation*}
By the invertibility of $W$, this transformation has the inverse
\begin{equation}
	\alpha^{-1}: \X \ra \Y, \; x \mapsto W^{-1} (x - c) .
	\label{Eq:alpha_inv}
\end{equation}
The transformation and its inverse therefore establish a one-to-one correspondence between points in $\X$ and $\Y$.

Suppose now that we are given a target distribution $\nu$ on $(\X, \B(\X))$ via an unnormalized density ${\varrho: \X \ra \coint{0}{\infty}}$ (as in \eqref{Eq:def_nu}), then the corresponding transformed distribution on $(\Y,\B(\Y))$ is the pushforward measure $\nu_{\alpha} := (\alpha^{-1})_{\#} \nu$. This results in the following intuitive relation between $\nu$ and $\nu_{\alpha}$: Let $X \sim \nu$, $Y \sim \nu_{\alpha}$, then
\begin{align*}
	&\P(\alpha(Y) \in A)
	= \P(Y \in \alpha^{-1}(A)) \\
	&= \nu_{\alpha}(\alpha^{-1}(A))
	= \nu(A)
	= \P(X \in A) ,
	\quad A \in \B(\X) ,
\end{align*}
meaning $\alpha(Y)$ and $X$ have the same distribution.

This simple construction is the basis of what we call an \textit{affine transformation tuning} (ATT) \textit{transition}. An ATT transition works by moving to the latent space via the inverse transformation, taking a step in the latent space and then moving back to the sample space. To express this formally, denote by $P_{\alpha}: \Y \times \B(\Y) \ra \ccint{0}{1}$ the transition kernel of an MCMC method on the latent space that leaves the transformed target distribution $\nu_{\alpha}$ invariant. In the following, the MCMC method taking this role will be referred to as the \textit{base sampler} (of the resulting ATT sampler). Now an ATT transition for target $\nu$, based on transformation $\alpha$ and auxiliary kernel $P_{\alpha}$, from the current state ${x_{n-1} \in \X}$ to a new state ${x_n \in \X}$, is implemented by Algorithm \ref{Alg:ATT_transition}.

\begin{algorithm}[tb]
	\caption{ATT transition}
	\label{Alg:ATT_transition}
	\textbf{Input:} transformation $\alpha: \Y \ra \X$, inverse transformation $\alpha^{-1}: \X \ra \Y$, transition kernel $P_{\alpha}$ on $(\Y,\B(\Y))$ targeting transformed target $\nu_{\alpha}$, current state $x_{n-1} \in \X$ \\
	\textbf{Output:} new state $x_n \in \X$
	\begin{algorithmic}[1]
		\STATE Map current state to latent space: $y_{n-1} := \alpha^{-1}(x_{n-1})$.
		\STATE Take a step in latent space according to $P_{\alpha}$, i.e.~draw a new state $Y_n \sim P_{\alpha}(y_{n-1}, \cdot)$, call the result $y_n \in \Y$.
		\STATE Map new state to sample space: $x_n := \alpha(y_n)$.
	\end{algorithmic}
\end{algorithm}

The transition kernel $P$ corresponding to the transition $x_{n-1} \ra x_n$ performed by Algorithm \ref{Alg:ATT_transition} is given by
\begin{align}
	\begin{split}
		P(x, A)
		&= \int_{\Y} \ind_A(\alpha(y)) P_{\alpha}(\alpha^{-1}(x), \d y) \\
		&= P_{\alpha}(\alpha^{-1}(x), \alpha^{-1}(A))
	\end{split}
	\label{Eq:transition_kernel_formula}
\end{align}
for any $x \in \X$, $A \in \B(\X)$.

A key observation regarding ATT is that the map $\alpha$ has constant Jacobian $J_{\alpha}(y) = W$, hence
the distribution $\nu_{\alpha}$ has unnormalized density
\begin{equation*}
	y
	\mapsto \; \abs{\!\det(J_{\alpha}(y))} \, \varrho(\alpha(y))
	= \abs{\!\det(W)} \, \varrho(\alpha(y)) ,
\end{equation*}
so that
\begin{equation*}
	\varrho_{\alpha}: \Y \ra \coint{0}{\infty}, \; y \mapsto \varrho(\alpha(y))
\end{equation*}
is also an unnormalized density of $\nu_{\alpha}$. In particular, an evaluation of the target density $\varrho_{\alpha}$ on the latent space is only as computationally costly as one of the untransformed target density $\varrho$ plus one of the transformation $\alpha$.

What are the roles of the parameters $c$ and $W$ of the transformation $\alpha(y) = W y + c$? Clearly, the parameter $c$ controls the center of the distribution. Specifically, if $\nu$ is centered around $c$, then the point to which this center corresponds in the latent space is ${\alpha^{-1}(c) = W^{-1}(c - c) = 0}$ (for arbitrary $W$), so that $\nu_{\alpha}$ is centered around the origin. The parameter $W$ on the other hand affects the covariance structure of the target distribution. Specifically, if $X \sim \nu$ has covariance matrix $\cov(X)$ with Cholesky decomposition ${\cov(X) = L L^{\top}}$, then $Y := \alpha^{-1}(X) \sim \nu_{\alpha}$ has covariance matrix
\begin{equation*}
	\cov(Y) = W^{-1} \cov(X) (W^{-1})^{\top} = (W^{-1} L) (W^{-1} L)^{\top} .
\end{equation*}
Thus, $\cov(Y)$ is the identity matrix whenever $W^{-1} L$ is orthogonal, the simplest case being $W = L$.

\subsection{Connections to Traditional Adaptive MCMC}

A straightforward way to deploy ATT in practice is to implement it as an adaptive MCMC method that uses the adaptivity to learn a suitable transformation $\alpha$ during the run, and performs each iteration by calling Algorithm \ref{Alg:ATT_transition} with the latest version of $\alpha$. Intuitively, the resulting method adaptively learns how best to transform the target distribution in order to simplify it in the eyes of its base sampler.
This is in contrast to ``traditional'' adaptive MCMC approaches, where the adaptivity is typically used to adjust the parameters of an underlying sampler's proposal distribution, thereby adjusting this sampler to better suit the target distribution.

Nevertheless, there are cases in which the transition kernels of an adaptive ATT chain  coincide with those of a corresponding traditional adaptive MCMC chain, where the latter uses a parametrized version of the former's base sampler as its underlying sampler. In Appendix \ref{App:equiv_trad_ada}, we introduce \textit{ATT-friendliness} as a property that precisely encapsulates which methods produce such an equivalence when used as the underlying/base sampler. We then consider several samplers in detail, namely RWM, IMH (with Gaussian proposal) and ESS, and verify for each of them that it is ATT-friendly.

\subsection{Adjustment Types} \label{SubSec:adj_types}

Motivated by considerations regarding the computational overhead introduced by ATT, we examine ways to restrict our previous specification of $\alpha$ as an a priori arbitrary bijective affine transformation. We frame these restrictions by the type of adjustments to the target distribution that they permit. To this end, we introduce the following terms.

\textit{Centering} is performed by transformations that shift the origin, which -- in absence of other adjustments -- are of the form $\alpha(y) = y + c$ for some $c \in \R^d$. \textit{Variance adjustments} are made by transformations that alter the target's coordinate variances without affecting the correlations between variables, which -- in absence of centering -- are transformations of the form $\alpha(y) = \diag(v) y$ for some $v \in \R^d$ with $v_i \neq 0 \; \forall i$. \textit{Covariance adjustments} are the unrestricted analogue to variance adjustments, meaning that -- again in absence of centering -- they encompass all transformations $\alpha(y) = W y$ with $W \in \GL_d(\R)$. Of course both variance and covariance adjustments can also be combined with centering, leading to transformations of the form $\alpha(y) = W y + c$, with $c \in \R^d$ and either $W = \diag(v)$ or $W \in \GL_d(\R)$. Note that the latter combination recovers the original generality, in that it allows $\alpha$ to be an arbitrary bijective affine transformation.

For some considerations on how to decide on adjustment types for a given problem, we refer to Appendix \ref{App:adj_types}. The question of how to choose the transformation's parameters for a given adjustment type will also be considered later.

\section{Parallelization and Update Schedules} \label{Sec:par_sched}

\subsection{Parallelized ATT} \label{SubSec:PATT}

Since ATT uses samples to approximate global features of the target distribution, we can utilize non-trivial parallelization to improve its performance by running $p \in \N$ parallel chains $(X_i^{(j)})_{i \in \N_0}$, $j=1,\dots,p$ that all rely on the same transformation $\alpha$. Whenever the parameters $c$ and $W$ are to be updated, say in iteration $n$, their new values can be computed based on all $n \cdot p$ samples $(X_i^{(j)})_{0 \leq i < n, 1 \leq j \leq p}$ generated up to this point. For $p \gg 1$, this should lead to a substantially faster (in terms of iterations per chain) convergence of the transformation parameters to their asymptotic values than with $p$ independent chains. 

However, if not used with caution, this parallelization approach can substantially reduce the method's iteration speed\footnote{By iteration speed we mean the number of iterations a method completes per time unit, e.g.~per second of CPU time.}: For example, an iteration of our default base sampler GPSS involves a stepping-out procedure and two shrinkage procedures (see Appendix \ref{SubApp:exp_methods} or \citet{GPSS} for details), each involving a random number of target density evaluations. Moreover, this number depends not only on the target distribution and the choice of GPSS's hyperparameter, but also (and often predominantly) on the random threshold drawn in each iteration to determine a slice. The computing time required for an iteration of GPSS is therefore also random and may vary substantially. Running the parallelized ATT approach with parameter updates in every iteration would require a synchronization of the $p$ chains for each update, meaning that all chains have to first terminate the current iteration before the parameters can be updated. Because each chain requires a random and largely varying amount of computing time to complete the iteration, the slowest chain will generally take far longer than an average chain to complete the iteration. In other words, the parallel scheme would spend a substantial amount of computing time waiting for the slowest chain to complete the iteration.

\subsection{Update Schedules}

A simple remedy is to update the transformation parameters not after every single iteration, but only after a larger number of iterations has been completed. The $p$ chains still need to be synchronized before every parameter update. But, if the number of iterations between updates is large enough (compared to both $p$ and the variance of the runtime of a base sampler iteration in the given setting), the variance of the accumulated runtimes of individual chains will be quite small compared to the total runtime, so that comparatively little time is spent idling. To formalize the approach of only occasionally updating the parameters, we introduce the following notions.

\begin{definition}
	An \textit{update schedule} is a strictly increasing sequence of positive integers $S := (s_k)_{k \in I}$, where $s_k \in \N$, $s_1 \geq 2$ and either $I = \N$ or $I = \{1,2,\dots,n_I\} \subset \N$ for some $n_I \in \N$. For a given update schedule $S = (s_k)_{k \in I}$, we call $n \in \N$ an \textit{update time} if and only if there exists $k \in I$ such that $s_k = n$. We write $n \in S$ if $n$ is an update time and $n \not\in S$ otherwise.
\end{definition}

The idea is now to define a suitable update schedule and then run the parallelized ATT approach, while updating the transformation parameters only in those iterations that are update times. This general approach will in the following be referred to as PATT (for \textit{parallel ATT}). In order to enable a concrete description of PATT, we require the following auxiliary definitions.

\begin{definition} \label{Def:alpha_family}
	We define $\A$ to be the family of bijective affine transformations on $\R^d$, i.e.~of all ${\alpha: \R^d \ra \R^d}$, ${y \mapsto W y + c}$ with $W \in \GL_d(\R), c \in \R^d$.
\end{definition}

\begin{definition}
	We define an \textit{ergodic base sampler} for a given target distribution $\nu$ to be a family of transition kernels $(P_{\alpha})_{\alpha \in \A}$ on $(\R^d, \B(\R^d))$, where for each $\alpha \in \A$, the kernel $P_{\alpha}$ is ergodic towards $\nu_{\alpha}$, meaning
	\begin{equation}
		\lim_{n \ra \infty} \TV(P_{\alpha}^n(y,\cdot), \nu_{\alpha}) = 0 \qquad \forall y \in \R^d .
		\label{Eq:tv_ergodicity}
	\end{equation}
\end{definition}

Note that \eqref{Eq:tv_ergodicity} is a relatively weak assumption on $P_{\alpha}$ and that there are a number of easily verified constraints on $\nu$ and $P_{\alpha}$ that imply it, see for example \citet{Tierney}, Section 3.

\begin{algorithm}[tb]
	\caption{PATT}
	\label{Alg:PATT}
	\textbf{Input:} ergodic base sampler $(P_{\alpha})_{\alpha \in \A}$ for the target distribution $\nu$, number of parallel chains $p \in \N$, initial states $x_0^{(1)},\dots,x_0^{(p)} \in \X$, update schedule $(s_k)_{k \in I}$ \\
	\textbf{Output:} samples $(x_i^{(j)})_{i \geq 0, 1 \leq j \leq p}$
	\begin{algorithmic}[1]
		\STATE Set $s_0 := 1$ and $\alpha_1: \Y \ra \X, \; y \mapsto y$.
		\FOR{$k = 1,2,\dots \in I$}
		\FOR{all $p$ chains in parallel, indexed by $j$}
		\FOR{$i = s_{k-1},\dots,s_k-1$}
		\STATE $x_i^{(j)} := \hyperref[Alg:ATT_transition]{\text{ATT\_transition}}(\alpha_k, \alpha_k^{-1}, P_{\alpha_k}, x_{i-1}^{(j)})$ \\
		\ENDFOR
		\ENDFOR
		\STATE Choose a new transformation $\alpha_{k+1} \in \A$ based on the available samples $(x_i^{(j)})_{0 \leq i < s_k, 1 \leq j \leq p}$.
		\ENDFOR
		\IF{$\abs{I} < \infty$}
		\STATE Set $n := \abs{I}$ and $\alpha := \alpha_{n+1}$.
		\FOR{all $p$ chains in parallel, indexed by $j$}
		\FOR{$i = s_n,s_n+1,\dots$}
		\STATE $x_i^{(j)} := \hyperref[Alg:ATT_transition]{\text{ATT\_transition}}(\alpha, \alpha^{-1}, P_{\alpha}, x_{i-1}^{(j)})$ \\
		\ENDFOR
		\ENDFOR
		\ENDIF
	\end{algorithmic}
\end{algorithm}

Based on these definitions, we can formulate PATT as Algorithm \ref{Alg:PATT}. It is not immediately clear how exactly one should choose the parameters of the transformations $\alpha_k$ (line 8 of Algorithm \ref{Alg:PATT}), even if one has already decided on which adjustment types to use. We therefore provide detailed suggestions on how to specify and efficiently update these parameters in Appendix \ref{App:param_choices}. It is also not obvious how to choose a good update schedule. On the one hand, a large delay between updates will reduce the overall waiting time. On the other hand, more frequent updates will reduce the number of iterations to reach a specific sample quality, so there is a trade-off between these two choices. How the transformation parameters $c$, $W$ themselves are chosen (e.g.~as sample mean and sample covariance, we refer to this as \textit{parameter choice} in the following) should also be considered when designing an update schedule. Specifically, returning to the notation of Section \ref{Sec:ATT} for the moment, if the parameter choice does not permit incorporating a new sample $x_n \in \R^d$ in a complexity independent of the number $n$ of previously incorporated samples $x_0,\dots,x_{n-1}$, then update schedules not designed to take this into account may result in computational costs that continue to grow in an unlimited fashion with increasing number of iterations. An example is the coordinate-wise sample median as a choice for updating the center parameter $c$, see Appendix \ref{SubApp:sample_median} for details.
In Appendix \ref{App:schedule} we provide some guidelines for designing setting-specific update schedules and propose default update schedules depending on the types of adjustments used, the sample space dimension $d$ and the number of parallel chains $p$. For an exemplary analysis of how the use of interacting parallel ATT chains (rather than independent parallel ones) and update schedules (rather than updates after every iteration) affects sample quality and run time of the resulting sampler, we refer to the ablation study in Appendix \ref{SubApp:par_US}.

\subsection{Initialization Burn-In} \label{SubSec:init_burn_in}

The main purpose of PATT is to allow samplers to work better in the long run. In particular, it is not very helpful (and can even be detrimental) in the early stages of sampling, where each chain needs to find a region of high probability mass, starting from a potentially poorly chosen initial state. Therefore, unless the user is very confident in the initialization of the parallel chains, they should first let the base sampler carry them forward for a reasonable number $n_{\text{burn}}$ of iterations. After this \textit{initialization burn-in} phase, the user can apply PATT by using the final state of each burn-in chain as the initial state of the corresponding PATT chain. To avoid further complicating our description of PATT, we consider this early stage a process entirely separate from the remaining sampling.

Although asymptotically irrelevant, the use of such an initialization burn-in period can have a substantial positive impact in the short- and mid-term of the sampling procedure, see also the ablation study on this effect in Appendix \ref{SubApp:IBP}.

\section{Theoretical Justification} \label{Sec:theo_just}

In this section, we examine the convergence of PATT with finite update schedules $S = (s_k)_{k \in I}$ where $\abs{I} < \infty$. Our motivation for this is twofold. First, we will demonstrate below that a finite update schedule necessitates only a weak assumption on the underlying base sampler to guarantee convergence of the corresponding PATT sampler. Second, after a sufficiently large number of iterations, the transformation parameters will be close to their optimal values. It may then be more economical to stop updating the parameters to rid oneself of the computational cost incurred from the update itself (which can be quite substantial, particularly in the case of covariance adjustments in high dimensions).

In order to properly express this section's result, we require some new notation encoding a PATT sampler's choice of the transformation parameters. Suppose in the following that the dimension $d$ of the sample space and the number of parallel chains $p$ are both clear from context.

\begin{definition} \label{Def:adj_schemes}
	We define a \textit{centering scheme} to be a family of functions $(\mathbf{c}_n)_{n \in \N}$ where $\mathbf{c}_n$ maps $\R^{(n \cdot p) \times d} \ra \R^d$, and a \textit{(co)variance adjustment scheme} to be a family of functions $(\mathbf{W}_n)_{n \in \N}$ where $\mathbf{W}_n$ maps $\R^{(n \cdot p) \times d} \ra \GL_d(\R)$.
\end{definition}

Intuitively, $\mathbf{c}_n$ and $\mathbf{W}_n$ take as input all the samples the PATT sampler has generated up to iteration $n$ and output the new transformation parameters $c \in \R^d$ and $W \in \GL_d(\R)$. Note that the option not to use centering is encoded in the above definition by $\mathbf{c}_n \equiv 0$ for all $n \in \N$ and the option not to use (co)variance adjustments is encoded by $\mathbf{W}_n \equiv I_d$ for all $n \in \N$. 

We now state a general result on the ergodicity of PATT for finite update schedules. As the use of an initilization burn-in period is clearly irrelevant to this property, we ignore it in our analysis.

\begin{theorem} \label{Thm:finite_US_ergodicity}
	Let $S = (s_k)_{k \in I}$ with $I = \{1,\dots,n_I\}$ for some $n_I \in \N$ be an arbitrary finite update schedule. Then for any collection of
	\begin{itemize}
		\item a target distribution $\nu$ on $(\R^d, \B(\R^d))$ given by an unnormalized density $\varrho: \R^d \ra \coint{0}{\infty}$ as in \eqref{Eq:def_nu},
		\item a number $p \in \N$ of parallel chains,
		\item a centering scheme $(\mathbf{c}_n)_{n \in \N}$,
		\item a (co)variance adjustment scheme $(\mathbf{W}_n)_{n \in \N}$, and
		\item an ergodic base sampler $(P_{\alpha})_{\alpha \in \A}$ for $\nu$,
	\end{itemize}
	the resulting PATT sampler is ergodic in the sense that the samples $(X_i^{(j)})_{i \in \N_0, 1 \leq j \leq p}$ generated by it satisfy
	\begin{equation*}
		\frac{1}{n p} \sum_{i=0}^{n-1} \sum_{j=1}^p f(X_i^{(j)})
		\asconv \int_{\R^d} f(x) \nu({\normalfont \d} x)
	\end{equation*}
	as $n \ra \infty$, for any $\nu$-integrable function $f: \R^d \ra \R$ and any choice of the initial states $X_0^{(1)}, \dots, X_0^{(p)}$.
\end{theorem}
\begin{proof}
	See Appendix \ref{App:theo_just_proof}.
\end{proof}

Of course also theoretical results for PATT schemes with infinite update schedules $(s_k)_{k \in \N}$ are desirable. One way to establish such results would be to utilize the \textit{AirMCMC} framework \cite{AirMCMC}: It applies to those adaptive MCMC methods that are \textbf{a}dapted \textbf{i}ncreasingly \textbf{r}arely (rather than changing the transition kernel in every iteration), which for PATT corresponds to using an update schedule $(s_k)_{k \in \N}$ for which the sequence $(s_{k+1} - s_k)_{k \in \N}$ is strictly increasing. The advantage of the AirMCMC framework is that it imposes signficantly weaker conditions than the general adaptive MCMC framework to arrive at the same theoretical guarantees. However, the conditions imposed in theorems on AirMCMC are still stronger than the corresponding ones for homogeneous MCMC. We therefore expect any theoretical results on PATT based on AirMCMC to rely on much more restrictive assumptions than the above theorem.

\section{Numerical Experiments} \label{Sec:experiments}

In this section we give a brief overview of the results of a number of numerical experiments in which we let PATT samplers compete with several other methods. Our experiments are in large part inspired by those of \citet{GenEllSS} and \citet{NUTS}.

\subsection{Methodology}

The main purpose of these experiments is to showcase the potential of PATT-GPSS (by which we denote any variant of PATT with base sampler GPSS) as a well-performing, user-friendly black-box sampler\footnote{By the term \textit{black-box sampler} we refer to any sampling method that does not require its target distribution to have a particular structure (e.g.~being a posterior resulting from a Gaussian prior, as in \citet{EllipticalSS}) and in particular does not explicitly rely on the probabilistic model underlying the target.}. To this end, we always deployed PATT-GPSS with the default update schedules defined in Appendix \ref{App:schedule}, thereby eliminating the need to devise setting-specific schedules. Moreover, all of our experiments used PATT-GPSS with centering (via sample means) and covariance adjustments (via sample covariances), showing these adjustment types to be a versatile default choice.

Although PATT-GPSS appears to work well for all targets that are not too misshapen, it should be emphasized that, under certain conditions, PATT with other base samplers can achieve the same sample quality at even lower computational cost. To demonstrate this, we also ran PATT-ESS, meaning PATT with general-purpose ESS (GP-ESS, cf.~Example \ref{Ex:GP-ESS}) as its base sampler, in our experiments. Like PATT-GPSS, we always ran PATT-ESS with centering, covariance adjustments and the default update schedules defined in Appendix \ref{App:schedule}.

As competitors for the PATT samplers, we found it natural to choose methods that were themselves proposed as user-friendly black-box samplers. From the class of non-adaptive, gradient-free MCMC methods, we chose \textit{hit-and-run uniform slice sampling} (HRUSS) (\citet{MacKayBook}, Section 29.7). Among the traditional adaptive MCMC methods, we chose the \textit{adaptive random walk Metropolis} algorithm (AdaRWM) \cite{HaarioAdaRWM}, in the formulation proposed by \citet{ExAda}. We implemented both HRUSS and AdaRWM in a \textit{naively parallelized} fashion, allowing us to run the same number of parallel chains for each of them as for PATT, while not letting these chains interact with one another. To represent sophisticated parallelized schemes of similar complexity as PATT, we chose \textit{(two-group) generalized elliptical slice sampling} (GESS) \cite{GenEllSS}. Finally, among gradient-based MCMC methods, we chose the \textit{No-U-Turn Sampler} (NUTS) \cite{NUTS}, as implemented by the probabilistic modeling and inference software \textit{Stan}\footnote{\url{https://mc-stan.org/}}.

We emphasize that these four methods and the two PATT samplers are all essentially tuning-free, and that, accordingly, the amount of effort we spent on hand-tuning their parameters to each experiment was negligible. For a brief explanation of the inner workings of GPSS and each of the four competitor methods, we refer to Appendix \ref{SubApp:exp_methods}, and for an explanation of GP-ESS to Example \ref{Ex:GP-ESS}.


When numerically analyzing the sampling performance of PATT and its competitors, we were more interested in their respective long-term efficiency than in their behavior in the early stages. We therefore used a generous burn-in period, in that we considered only those samples generated in the latter half of iterations for this analysis. To assess the performance of each method, we computed two cost metrics, two sample quality metrics and two aggregates of the former and the latter.

To measure sampling costs, we on the one hand let each sampler count the number of \textit{target density evaluations} (TDE) it required in each iteration of each chain and used these figures to compute the average number of TDE it required per (single-chain) iteration (TDE/it). Taking TDE/it to represent a sampler's cost is common practice (see e.g.~\citet{GenEllSS,GPSS}) and well-motivated by the observation that, in real-world applications, the sampler's remaining overhead is typically negligible compared to the amount of resources it expends on TDE. On the other hand, since several of the methods we considered (including the PATT samplers) can have considerable overhead, we also measured the samplers' physical runtimes (for details see Appendix \ref{SubApp:runtimes}).

To assess sample quality, we relied on two different quantities. Firstly, we considered the \textit{mean step size} (MSS), that is, the Euclidean distance between two consecutive samples, averaged over all pairs of consecutive samples under consideration. Supposing that a sampler has already reached (empirically) stable behavior and is no longer moving through the sample space erratically, the MSS gives an indication as to how quickly it traverses the target's regions of high probability mass (with a higher MSS meaning larger steps and therefore suggesting a quicker exploration of these regions). It has previously been considered by \citet{GPSS} and is closely related to the more frequently used \textit{mean squared jump distance}, but more intuitive in its concrete values (since Euclidean distances are much closer to our real world experiences than squared Euclidean distances).

Secondly, we considered the \textit{mean integrated autocorrelation time} (mean IAT), which, for a given set of samples from a multi-chain sampler, is obtained by computing the IATs of each univariate marginal of the samples in each individual chain and then averaging these values over both the marginals and the chains. The IAT is commonly considered in the analysis of sampling methods because it is mathematically well-motivated to view the quotient of nominal sample size and IAT as the number of \textit{effective samples} (ES) (see e.g.~\citet{Gelman}, Section 11.5).

This means that the IAT can be viewed as the number of (single-chain) iterations required to produce one ES, which gives rise to natural performance metrics which weigh cost and sample quality against one another: The product of TDE/it and IAT represents the number of TDE required for one ES (TDE/ES). Similarly, the quotient of the number of samples generated per second (Samples/s) and the IAT gives the number of ES per second (ES/s). Although we also provide the non-aggregate metrics individually, we use TDE/ES and ES/s as the primary criteria to judge the overall performance of the samplers.

The source code for our numerical experiments is provided as a github repository\footnote{\url{https://github.com/microscopic-image-analysis/patt_mcmc/}}. In the interest of reproducibility, all our experiments are designed to be executable on a regular workstation (rather than requiring a cluster).

\subsection{Results}

Here we briefly summarize the key parameters and results of our experiments. For details on the models, the resulting target densities, the data and the samplers' settings (number of iterations etc.), we refer to Appendix \ref{App:exp_details}.

In our first experiment, we performed Bayesian inference on a model in which both prior and likelihood were given by multivariate exponential distributions, that is, distributions whose densities have elliptical level sets and tails like ${x \mapsto \exp(-\norm{x}/\sigma)}$ for some $\sigma > 0$. We set the sample space dimension to $d = 50$ and worked with synthetic data. The resulting sampling statistics are shown in Table \ref{Tab:multiv_exp_dists}. Moreover, trace plots of the first univariate marginal and histograms of the step sizes are presented in Appendix \ref{App:plots}, Figure \ref{Fig:multiv_exp_dists}.

\begin{table*}[t]
	\caption{Sampling statistics for the experiment on Bayesian inference with multivariate exponential distributions.}
	\label{Tab:multiv_exp_dists}
	\vskip 0.1in
	\begin{center}
		\begin{small}
			\begin{sc}
				\begin{tabular}{lrrrrrr}
					\toprule
					Sampler		& TDE/it & Samples/s	& Mean IAT	& MSS	& TDE/ES	& ES/s \\
					\midrule
					PATT-ESS	&   2.80 &  6466		&    7.19	& 4.90 	&    20.10	&  899.77 \\
					PATT-GPSS	&   6.28 &  5903		&    1.06	& 8.49 	&     6.68	& 5554.64 \\
					HRUSS		&   5.21 &  8465		& 3200.26	& 0.53 	& 16657.77	&    2.65 \\
					AdaRWM		&   1.00 & 27521		&  151.13	& 0.56 	&   151.13	&  182.10 \\
					GESS		&   5.02 &  3500		&   33.36	& 2.42 	&   167.62	&  104.92 \\
					Stan's NUTS	&   7.18 &   341		&   93.40	& 5.41 	&   670.33	&    3.66 \\
					\bottomrule
				\end{tabular}
			\end{sc}
		\end{small}
	\end{center}
	\vskip -0.2in
\end{table*}

Next we conducted a series of experiments on \textit{Bayesian logistic regression} (BLR) with mean-zero Gaussian prior for different data sets varying in the number of data points and features. In each of these experiments we added a constant feature to the data to enable an intercept. In some cases, we also performed \textit{feature engineering} (FE) by augmenting the data with the two-way interactions between the given features. The resulting sample space dimensions were ${d=25}$ (\textsc{Credit}), $d=31$ (\textsc{Breast}), $d=45$ (\textsc{Pima}, with FE) and $d=78$ (\textsc{Wine}, with FE).
For the sake of brevity, we only state the TDE/ES values for each experiment here, see Table \ref{Tab:BLR_summary}. The complete sampling statistics are presented in Appendix \ref{SubApp:BLR}, Tables \ref{Tab:BLR_German_credit}, \ref{Tab:BLR_breast_cancer}, \ref{Tab:BLR-FE_Pima_diabetes} and \ref{Tab:BLR-FE_wine_quality}. Moreover, Figures \ref{Fig:BLR_German_credit_covs}, \ref{Fig:BLR_breast_cancer_covs}, \ref{Fig:BLR-FE_Pima_diabetes_covs} and \ref{Fig:BLR-FE_wine_quality_covs} in Appendix \ref{App:plots} show the final covariance/scale matrices used by the adaptive samplers, thereby giving some insights into the intricate covariance structure the samplers had to learn to perform well.

\begin{table}[t]
	\caption{TDE/ES statistics for the experiments on BLR.}
	\label{Tab:BLR_summary}
	\begin{center}
		\begin{small}
			\begin{sc}
				\begin{tabular}{lrrrr}
					\toprule
					Sampler 	& Credit & Breast  & Pima    & Wine \\
					\midrule
					PATT-ESS	&    1.8 &    43.4 &     5.9 &    12.6 \\
					PATT-GPSS	&    7.6 &    71.7 &    19.1 &    28.3 \\
					HRUSS 		& 2220.9 & 16893.9 &  3331.4 & 25926.8 \\
					AdaRWM		&  191.1 &   151.0 &   420.2 &  3207.8 \\
					GESS		&  394.7 &   572.3 & 12170.4 &      -- \\
					Stan's NUTS &   50.7 &   175.3 &    29.8 &   227.5 \\
					\bottomrule
				\end{tabular}
			\end{sc}
		\end{small}
	\end{center}
	\vskip -0.1in
\end{table}

Finally, we conducted an experiment on Bayesian hyperparameter inference for Gaussian process regression of US census data in dimension $d=30$. Here the results are shown in Table \ref{Tab:hyperparam_inf}.

\begin{table*}[t]
	\caption{Sampling statistics for the experiment on Bayesian hyperparameter inference for Gaussian process regression of US census data.}
	\label{Tab:hyperparam_inf}
	\begin{center}
		\begin{small}
			\begin{sc}
				\begin{tabular}{lrrrrrr}
					\toprule
					Sampler		& TDE/it	& Samples/s & Mean IAT	& MSS	& TDE/ES & ES/s \\
					\midrule
					PATT-ESS	&  2.84		&  848.57	&   7.74	& 36.75	&  21.95 & 109.68 \\
					PATT-GPSS	&  7.95		&  636.44	&   6.37	& 39.61	&  50.61 &  99.91 \\
					HRUSS		&  5.25		& 1040.55	&  95.63	&  7.95 & 501.96 &  10.88 \\
					AdaRWM		&  1.00		& 5130.24   & 163.91	&  3.78 & 163.91 &  31.30 \\
					Stan's NUTS	& 13.21		&  117.77   &   1.14	& 70.54 &  15.30 & 101.49 \\
					\bottomrule
				\end{tabular}
			\end{sc}
		\end{small}
	\end{center}
	\vskip -0.1in
\end{table*}

Across all of these experiments, the two PATT samplers produced high quality samples at a very reasonable computational cost. NUTS typically managed to generate samples of slightly higher quality than PATT, but at a disproportionately higher cost per sample. Conversely, AdaRWM generally requires very little computational effort per sample, but its sample quality in our experiments lagged far behind those of PATT and NUTS. Neither HRUSS nor GESS were able to keep up with the better performing samplers in any of the experiments.

Consequently, in each of the first five experiments, the better performing PATT sampler beat all its competitors in terms of TDE/ES (often by an order of magnitude or more), and in all but one of these experiments (namely BLR for breast cancer data) also in terms of ES/s. In the experiment on Bayesian hyperparameter inference, NUTS outperformed the PATT samplers in the TDE/ES sense (PATT-ESS only slightly, PATT-GPSS more significantly), but the two approaches were evenly matched according to the alternative performance summary criterion ES/s.

\section{Concluding Remarks} \label{Sec:conclusion}

We propose \textit{parallel affine transformation tuning} (PATT), a general framework for setting up a computationally efficient, adaptively self-improving multi-chain sampling method based on an arbitrary \textit{base sampler}, i.e.~a single-chain, non-adaptive MCMC method. PATT provides the user with plenty of freedom in choosing adjustment types, update schedules and the base sampler, while at the same time offering reasonable default choices for all of these parameters. Moreover, with the default update schedules we propose, PATT is expected to scale well with the hardware it runs on, adapting suitably to every architecture from low-end personal workstations to large clusters.

Particularly noteworthy is the synergy of PATT with \textit{Gibbsian polar slice sampling} (GPSS) \cite{GPSS} as the base sampler, in short denoted PATT-GPSS. As we demonstrated through a series of numerical experiments, PATT-GPSS is often able to produce samples of very high quality at a reasonable computational cost.

We wish to emphasize that this is not only the result of GPSS being well-suited to adaptively learning affine transformations like we do in PATT. Rather, PATT-GPSS manages to benefit from the extremely high performance ceiling GPSS has under optimal conditions, namely for target distributions that are rotationally invariant around the origin and unimodal along rays emanating from it. In fact, there is strong theoretical (cf.~our remarks on this in Section \ref{Sec:intro}) and numerical (cf.~the experiments of \citet{GPSS}) evidence that GPSS generally performs dimension-independently well in such rotationally-invariant settings. Importantly, through PATT, various challenging targets can be brought so close to rotational invariance that GPSS, in its role as the base sampler, starts to approach the remarkably good performance it is known to exhibit under optimal conditions (cf.~our numerical results in Section \ref{Sec:experiments} and Appendices \ref{App:exp_details} and \ref{App:ablation}).

The synergy of PATT with \textit{elliptical slice sampling} (ESS) \cite{EllipticalSS} should not be understated either. Although PATT-ESS would most likely work poorly for any target with (heavy) polynomial tails (according to the numerical results of \citet{GPSS}), our numerical experiments showed that it consistently performs very well across a wide range of exponentially-tailed targets. For targets whose tails are Gaussian or lighter, it even exhibited excellent performances, usually unmatched by any of its competitors (including PATT-GPSS, which in all those settings produced samples of higher quality than those of PATT-ESS, but at a disproportionately larger cost per sample).

Accordingly, if one wishes to apply PATT to a target distribution which one knows to have at most Gaussian tails, we would actually recommend using PATT-ESS. For all other cases, i.e.~the tails being heavier than Gaussian or their heavyness not being clear from the model underlying the target, we stand by our recommendation of PATT-GPSS as the default version of PATT.

We would also like to stress one important limitation of PATT (with either GPSS or ESS as its base sampler). Although its performance ceiling under optimal conditions -- say, for target densities whose level sets are elliptical -- does, by all indications, not deteriorate with increasing dimension, the same cannot be said of the speed at which it learns the affine transformation necessary to transform such a target into one that is rotationally invariant around the origin. While this deterioration in speed is quite moderate, it nevertheless means that in very high dimensions, PATT may require a large number of iterations to begin approaching its asymptotic performance. This issue is not unique to PATT. For example, slowly converging empirical covariances in high dimensions had already been observed for AdaRWM by \citet{ExAda}.


\section*{Acknowledgements}
We thank the anonymous referees for their suggestions. PS and MH gratefully acknowledge funding by the Carl Zeiss Foundation within the program ``CZS Stiftungsprofessuren'' and the project ``Interactive Inference''. Moreover, the authors are grateful for the support of the DFG within project 432680300 -- CRC 1456 subprojects A05 and B02.

\section*{Impact Statement}

This paper presents work whose goal is to advance the field of probabilistic inference. There are many potential societal consequences of our work, none which we feel must be specifically highlighted here.


\appendix
\section{Choosing the Adjustment Type(s)} \label{App:adj_types}

In this section we provide some general considerations on the potential advantages and downsides of the three adjustment types (cf.~Section \ref{SubSec:adj_types}) in an ATT or PATT sampler and try to give recommendations for when to use each of them.

Generally speaking, base samplers can be divided into those whose transition kernel is affected by centering with a fixed parameter $c \in \R^d \setminus \{\fatzero\}$, which we call \textit{location-sensitive}, and those that are unaffected by it, which we call \textit{location-invariant}. Some examples of location-sensitive base samplers are IMH, pCN-MH, ESS and GPSS, and some examples of location-invariant ones are RWM, MALA, HMC and all variants of uniform slice sampling.

Naturally, it only makes sense to use centering when one is working with a location-sensitive base sampler. On the other hand, as we will see in the next subsection, the computational overhead associated with centering is minimal, so that small gains in performance already justify its use. Furthermore, the improvement in sample quality due to centering is often large, particularly if the sample space is high-dimensional.
We therefore recommend to use centering per default when working with a location-sensitive base sampler. Of course even with such base samplers there are some scenarios where centering is ill-advised, for example those where the target distribution is known to be centered around the origin.

Variance adjustments can improve the performance, in terms of computational cost per iteration and/or sample quality, of virtually all base samplers. Such performance gains can be expected as soon as the target distribution has a substantial variation among its coordinate variances (i.e.~the diagonal entries of its covariance matrix) and should grow alongside any increases in this variation. Moreover, like with centering, the computational overhead associated with the adjustments is minimal. We therefore recommend using variance adjustments pretty much whenever covariance adjustments are ill-advised or infeasible and the variation among the target distribution's coordinate variances is not known to be small.

Covariance adjustments are a sort of double-edged sword. On the one hand, they constitute a powerful way of reducing the complexity of a given target distribution with non-trivial correlation structure, thereby improving the sample quality of virtually all base samplers. This even includes some which are unaffected by both other adjustments, such as random-scan uniform slice sampling (RSUSS) \cite{SSNeal}. On the other hand, covariance adjustments are far more computationally costly than centering and variance adjustments, particularly if the sample space is high-dimensional. We therefore find it difficult to give far-reaching recommendations regarding the use of these adjustments, though a good rule of thumb may be to always consider their use and only refrain from it if it proves computationally infeasible or mixing-wise unhelpful.

Some exemplary numerical evidence on the advantages of using the different adjustment types, both individually and in combination, can be found in the ablation study on adjustment types in Appendix \ref{SubApp:ablation_adj_types}.

We conclude this section with some general considerations on potential demerits to using the different adjustments. Specifically, we wish to emphasize that the use of a given adjustment type is likely to be detrimental to a PATT sampler's performance if the feature of the target density $\varrho$ that it modifies (e.g.~the extent to which it is centered around the origin) is already in its optimal configuration in the untransformed target (although we imagine such knowledge is seldom available in practice). This is intuitively clear: If $\varrho$ is already in its optimal configuration, an affine transformation $\varrho_{\alpha}$ of it cannot possibly be in a better one and is, for virtually all choices of the transformation $\alpha$, actually in a worse one.

Moreover, even if the feature modified by an adjustment type is entirely irrelevant to the base sampler, allowing PATT to modify that feature may still be somewhat detrimental to the performance of the PATT sampler, as it introduces instability into the sampling procedure (since, from the base sampler's point of view, the target changes at every update time).
On the other hand, we note that this detrimental effect is only relevant in the short- to mid-term. Asymptotically (assuming an infinite update schedule), the affine transformation learned by PATT will cease to alter the feature in question and thus no longer be detrimental to the sampler's performance.

Nonetheless, if one seeks to apply PATT in real world applications, one should take care to avoid adjusting both those features for which one knows the target distribution to already be in the optimal configuration and those that are irrelevant to the base sampler one wishes to use. Of course knowledge of whether a feature is in optimal configuration in the target density $\varrho$ presupposes a certain degree of insight into the probabilistic model underlying it, so the above considerations are not helpful in a true black-box setting.

\section{Choosing and Updating the Parameters} \label{App:param_choices}

Here we examine, for each of the three adjustment types, one or more ways to choose the transformation parameter associated with it, based on a given set of samples. Moreover, as promised earlier, we elaborate on efficient ways of updating the chosen parameters and the associated computational complexities.

As ATT can be retrieved from PATT as a special case, we work within the PATT framework in the following. That is, we examine the parameter choices for a sampler that runs $p \geq 1$ parallel ATT chains according to some update schedule $S = (s_k)_{k \in I}$, meaning that at each update time $s_k$ all the samples up to that time are pooled and new transformation parameters are computed based on them. For $i \geq 0$ and $1 \leq j \leq p$, we denote by $x_i^{(j)} \in \X = \R^d$ the state (in the sample space) of the $j$-th chain in the $i$-th iteration.

\subsection{Sample Mean}

For centering, one needs to choose the parameter ${c \in \R^d}$. As outlined in Section \ref{SubSec:formal_desc}, for the transformed distribution $\nu_{\alpha}$ to be centered around the origin, $c$ would need to be the center of the untransformed target $\nu$. Supposing that $\nu$ is not too misshapen, and light-tailed enough to have a well-defined mean (a very weak assumption), it is reasonable to view this mean as its center. Since the mean of an intractable target distribution is generally not known to the user, it needs to be approximated based on the available samples. The natural way to do this is by using the \textit{sample mean} $c := m_n$, which in iteration $n$ of our multi-chain scheme is given by
\begin{equation}
	m_n := \frac{1}{n p} \sum_{i=0}^{n-1} \sum_{j=1}^p x_i^{(j)} .
	\label{Eq:sam_mean}
\end{equation}
In order to update the sample mean in a computationally efficient manner at each update time $s_k$, one may rely on the numerically stable recursion
\begin{equation}
	m_{s_k}
	= m_{s_{k-1}} + \frac{1}{s_k p} \sum_{i={s_{k-1}}}^{s_k-1} \sum_{j=1}^p (x_i^{(j)} - m_{s_{k-1}}) ,
	\label{Eq:sam_mean_rec}
\end{equation}
based on Welford's algorithm \cite{Welford}, to utilize the old sample mean $m_{s_{k-1}}$ in computing the new one.

Using \eqref{Eq:sam_mean_rec}, we can incorporate the $(s_k - s_{k-1}) \cdot p$ new samples into the sample mean in just $\O((s_k - s_{k-1}) \cdot p \cdot d)$ operations (recall that each sample is a $d$-dimensional vector). If we \textit{amortize} this cost across the $s_k - s_{k-1}$ (multi-chain) iterations these samples were generated in, we obtain an amortized cost of $\O(p \cdot d)$ per iteration for the parameter updates. Since this complexity is already reached by just storing the corresponding $p$ samples, the computational overhead associated with the use of these updates is negligible.

\subsection{Sample Variance}

The goal of variance adjustments is to transform the target distribution $\nu$ in such a way that the covariance matrix of the transformed target distribution $\nu_{\alpha}$ has only ones on its diagonal. Analogously to centering, actually reaching this goal would require knowledge of the target distribution that is generally unavailable. Specifically, one would need to know the coordinate variances
\begin{equation*}
	\sigma_l^2
	:= \var_{X \sim \nu}(X[l]) ,
\end{equation*}
where we write $X = (X[1],\dots,X[d])^{\top}$, use them to generate the vector of standard deviations $v := (\sigma_1, \dots, \sigma_d)^{\top}$ and use $W := \diag(v)$ as the matrix parameter of $\alpha$. To deal with the values $\sigma_l^2$ being unknown, we may approximate them by the \textit{coordinate-wise sample variances} of the available samples. This corresponds to estimating the vector $v$ by
\begin{equation}
	v_n
	:= \sqrt{ \frac{1}{n p - 1} \sum_{i=0}^{n-1} \sum_{j=1}^p \big( x_i^{(j)} - m_n\big)^2 } ,
	\label{Eq:sam_std_dev}
\end{equation}
where for any vector $a = (a[1],\dots,a[d]) \in \R^d$ we denote
\begin{align*}
	&a^2 := (a[1]^2,\dots,a[d]^2)^{\top} \in \R^d , \\
	&\sqrt{a} := (\sqrt{a[1]},\dots,\sqrt{a[d]})^{\top} \in \R^d .
\end{align*}
Again we may greatly improve the computational efficiency of performing these adjustments by relying on recursive computations. Let us set
\begin{equation*}
	q_n := \sum_{i=0}^{n-1} \sum_{j=1}^p (x_i^{(j)} - m_n)^2
\end{equation*}
and denote by $\odot$ the operator for element-wise multiplication of $d$-dimensional vectors, i.e.
\begin{equation*}
	a \odot b
	:= (a[1] \cdot b[1], \dots, a[d] \cdot b[d])^{\top} ,
	\quad a, b \in \R^d .
\end{equation*}
A straightforward calculation then reveals that Welford's algorithm for recursively computing $v_n$ in a numerically stable fashion \cite{Welford} can be adapted to the update schedule framework, which results in the formulas
\begin{align}
	q_{s_k} &= q_{s_{k-1}} \!+\! \sum_{i = s_{k-1}}^{s_k - 1} \sum_{j=1}^p (x_i^{(j)} - m_{s_{k-1}}) \!\odot\! (x_i^{(j)} - m_{s_k}) , \notag \\
	v_{s_k} &= \sqrt{ \frac{1}{s_k p - 1} \, q_{s_k} } . 
	\label{Eq:sam_std_dev_rec}
\end{align}
The above recursions, together with \eqref{Eq:sam_mean_rec}, allow us to incorporate the $(s_k - s_{k-1}) \cdot p$ new samples into $v_n$ in the same (amortized) complexity as for the sample mean, so that this use of ATT can also be implemented with negligible computational overhead.

We emphasize that in practice one would not actually compute and maintain the full matrices $W_n := \diag(v_n)$, but rather make use of the fact that the matrix-vector product $\diag(v_n) y$ coincides with $v_n \odot y$. In other words, the evaluation of $\alpha(y) = \diag(v_n) y$ (and analogously also that of $\alpha^{-1}$) can also be implemented to only use $\O(d)$ operations per evaluation (note that this does not change when combining these adjustments with centering). In particular, the remaining overhead of variance adjustments (i.e.~aside from maintaining the transformation's parameters) is also of the same negligible computational complexity as the one for centering.

\subsection{Sample Covariance}

As we saw in Section \ref{SubSec:formal_desc}, the canonical way to transform the target distribution into having the identity matrix as its covariance is to use the Cholesky factor $L$ of $\Sigma := \cov(X), \; X \sim \nu$, meaning $\Sigma = L L^{\top}$, as the matrix parameter $W$ of $\alpha$, that is, to set $W := L$. Analogously to the variance adjustments, a natural way of approximating this optimal choice is to first approximate the true covariance $\Sigma$ by the \textit{sample covariance}
\begin{equation*}
	\Sigma_n := \frac{1}{n p - 1} \sum_{i=0}^{n-1} \sum_{j=1}^p (x_i - m_n) (x_i - m_n)^{\top}
\end{equation*}
then compute its Cholesky decomposition $\Sigma_n = L_n L_n^{\top}$ and use $W_n := L_n$ as the matrix parameter of $\alpha$.

Like for centering and variance adjustments, we may use recursive computations based on Welford's algorithm \cite{Welford}. With
\begin{equation*}
	Q_n := \sum_{i=0}^{n-1} \sum_{j=1}^p (x_i - m_n) (x_i - m_n)^{\top}
\end{equation*}
we obtain the recursions
\begin{align*}
	Q_{s_k} &= Q_{s_{k-1}} + \sum_{i=s_{k-1}}^{s_k-1} \sum_{j=1}^p (x_i - m_{s_{k-1}}) (x_i - m_{s_k})^{\top} , \\
	\Sigma_{s_k} &= \frac{1}{s_k p - 1} Q_{s_k} .
\end{align*}
Again together with \eqref{Eq:sam_mean_rec}, these recursions allow for the incorporation of the $(s_k - s_{k-1}) \cdot p$ new samples into $\Sigma_n$ in $\O((s_k - s_{k-1}) \cdot p \cdot d^2)$ operations. If we amortize this over the $s_k - s_{k-1}$ (multi-chain) iterations, we obtain a cost of $\O(p \cdot d^2)$ per iteration, which is optimal for PATT schemes with $p$ chains and non-sparse covariance matrix parameter. Recall, however, that updating an ATT scheme based on the sample covariance $\Sigma_n$ also necessitates computing the sample covariance's Cholesky factor $L_n$ and that factor's matrix inverse $L_n^{-1}$ (to evaluate $\alpha^{-1}$, cf.~\eqref{Eq:alpha_inv}, which is necessary when translating a state from sample space to latent space), both of which has complexity $\O(d^3)$. We therefore recommend suitably scaling covariance adjustments with the dimension $d$, e.g.~by using $p = \O(d)$ parallel chains or by choosing an update schedule $S = (s_k)_{k \in I}$ such that $s_k - s_{k-1}$ is at least of order $\O(d)$, either of which makes it so that the complexity of the aforementioned computations no longer dominates that of the overall method.

We also note that by using a non-diagonal (more precisely non-sparse) matrix $W$ in the transformation $\alpha$, one necessarily incurs an \textit{evaluation overhead} of at least $\O(d^2)$ per iteration, because evaluating $\varrho_{\alpha}$ in a given point $y \in \Y = \R^d$ involves evaluating $\alpha$ at $y$, which in turn necessitates computing the matrix-vector-product $W y$, which has complexity $\O(d^2)$.

\subsection{Sample Median} \label{SubApp:sample_median}

If the target distribution is heavy-tailed, any MCMC-style sampler for it will occasionally make trips far into the tails. The samples from these trips will (at least temporarily) have an enormous impact on the sample mean, perturbing it far from the true mean and consequently hindering efficient sampling. Furthermore, if the tails are heavy enough, $\nu$ will not even have a well-defined mean, so that the sequence $(m_n)_{n \in \N_0}$ simply would not converge. Consequently, whenever one is faced with a substantially heavy-tailed target distribution, it is advisable to estimate the target's center by something other than the sample mean. Possible choices include trimmed means, Winsorized means and the coordinate-wise sample median. In the following, we focus on the latter option and elaborate on how it could be implemented in practice.

As the name suggests, we define the coordinate-wise sample median of a set of samples $(x_i^{(j)})_{0 \leq i < n, 1 \leq j \leq p}$ to be the vector $z_n \in \R^d$ whose $l$-th entry is given by the median of the $l$-th entries of the vectors $x_i^{(j)}$,
\begin{equation*}
	z_n[l]
	= \text{median}(x_0^{(1)}[l],...,x_{n-1}^{(p)}[l]) .
\end{equation*}
The naive approach to computing $z_n$ is to simply sort these $l$-th entries (separately for each $l$) and find the median as the middle element (or the arithmetic mean of the two middle elements, if $n \cdot p$ is even) of the sorted sequence. Because the distribution of the $l$-th vector entries will, in the settings we consider, usually be far from uniform, efficient sorting schemes (such as bucket sort) are out of question, and so the complexity of the aforementioned sorting-based median computation is $\O(n p \log(n p) d)$.

Similarly, one can maintain the median of a growing sequence of values by explicitly maintaining the two sorted half-sequences of elements that are smaller, respectively larger, than the median in two instances of a suitable data structure (e.g.~a binary search tree). In the present case, supposing that the method has already been running for $n$ iterations, this approach would require $\O(\log(n p) p d)$ operations to update all $d$ medians based on a batch of $p$ new samples $(x_{n-1}^{(j)})_{1 \leq j \leq p}$. In other words, because the cost of inserting an element into a sorted data structure increases logarithmically in the size of the data structure, maintaining the sample median in this way would lead updates to become arbitrarily expensive if the method is run for sufficiently many iterations. Naturally, we would like to avoid this sort of degeneration of our method.

Though this appears to be infeasible if one requires the median to be updated in each iteration, it can still be accomplished if one is satisfied with occasional updates: It turns out that by using an update schedule that offsets the steadily increasing update cost by performing updates less and less frequently as the sampling goes on, we can achieve an amortized cost per iteration that does not depend on the number of iterations up to that point. To this end, we note that the first approach we mentioned (simply sorting a given set of samples) does not have optimal complexity: It was shown by \citet{MedianLinear} that the median of a given set of values can actually be computed in (deterministic) linear complexity. Thus an update of the sample median after $n$ iterations with $p$ parallel chains can be performed in complexity $\O(n \cdot p \cdot d)$. Now consider the update schedule $S = (s_k)_{k \in \N}$ with $s_k = \lfloor a^{k + b} \rfloor$ for some $a \in \ooint{1}{\infty}$ and $b \in \ooint{0}{\infty}$. Since we may amortize the cost $\O(s_k \cdot p \cdot d)$ of updating at iteration $n = s_k$ over the $s_k - s_{k-1}$ iterations since the last update, we obtain an amortized cost of
\begin{align*}
	\frac{\O(s_k \cdot p \cdot d)}{s_k - s_{k-1}}
	&= \O\left( \frac{a^{k + b}}{a^{k + b} - a^{k - 1 + b} } \cdot p \cdot d \right) \\ 
	&= \O\left( \frac{a}{a - 1} \cdot p \cdot d \right) 
	= \O(p \cdot d) 
\end{align*}
per iteration, which no longer depends on $n$. Note that with integer choices of $a$ and $b$ the rounding becomes unnecessary, but that a choice $a \in \ooint{1}{2}$ may be preferable performance-wise because it can allow significantly more updates within a finite number of iterations than even the minimal integer choice $a = 2$.

\subsection{Well-Definedness Issues}

In this section we discuss possible discrepancies between the transformation parameter choices outlined in Section \ref{App:param_choices} and the implicit requirement (made somewhat more explicit in Definition \ref{Def:adj_schemes}) that the parameter choices one makes should always lead to a well-defined, bijective transformation $\alpha$. For centering this is generally not an issue, regardless of which parameter choice one uses.

For variance adjustments, one already needs to be a bit more careful: Because the sample variance of a set of identical values is zero, if the available samples all coincide in some coordinate, the resulting vector $v_n$ will have a zero entry, so that the transformation $\alpha(y) := v_n \odot y$ is not actually bijective. Of course when using PATT with $p > 1$ parallel chains this pathological case is easily avoided by initializing the different chains with coordinate-wise distinct states $x_0^{(j)}$, $1 \leq j \leq p$ (e.g.~independently drawn from some Gaussian). However, if one only wants to use a single chain or if one insists on initializing all parallel chains with the same initial state (for whatever reason), some care may need to be taken to ensure well-definedness.

Whether such a setup is problematic actually comes down to the base sampler one wishes to use: For various slice sampling methods, e.g.~ESS, GPSS and hit-and-run uniform slice sampling (HRUSS), it is clear that two consecutive samples $x_{n-1}$, $x_n$ produced by any of these methods differ almost surely in all $d$ components, so that computing the sample means based on any number $n \geq 2$ of iterations suffices to ensure their well-definedness. On the other hand, any method based on random-scan Gibbs sampling, i.e.~on updating the state only in one randomly chosen coordinate per iteration, is at risk of violating the well-definedness, because for any $n < \infty$ there is a positive (albeit usually very small) probability that one of the coordinates was never chosen for updating throughout the first $n$ iterations, so that all of the available samples coincide in that coordinate. Accordingly, this issue can even affect methods that are rejection-free overall, most notably RSUSS. Of course Metropolis-Hastings (MH) methods are, in regards to this well-definedness, even more problematic than random-scan Gibbs sampling. Because in each iteration they may reject the proposal with some positive probability, it is entirely plausible for an MH method to produce a decently sized sequence of states that all coincide.

In summary, in order to ensure well-definedness of ATT transformations based on the sample variance, one must either use multiple parallel chains, initialized in coordinate-wise distinct states, or use a method that is \textit{coordinate-wise rejection-free} (in particular no random-scan Gibbs sampling, no MH methods), or slightly modify the variance estimator one uses, e.g.~by incorporating a ``dummy state'' into the computations.

For covariance adjustments based on the sample covariance, well-definedness is an even larger issue: For these adjustments to be well-defined, the sample covariance must be positive definite (since it must have a Cholesky decomposition), for which the samples it is computed from must span the entire sample space $\R^d$. In other words, while the cases that break the well-definedness of variance adjustments do the same for covariance adjustments, there are some additional cases that break only the latter. Moreover, these additional cases, while extremely rare in practice, cannot simply be precluded by choosing a suitable base sampler.

Therefore, if one must be absolutely certain that the covariance adjustments are well-defined, one should replace the sample covariance $\Sigma_n$ by a regularized version of itself, i.e.~by $\Sigma_n + \epsi I_d$ for some $\epsi > 0$, where $I_d \in \R^{d \times d}$ is the identity matrix. By choosing $\epsi$ large enough (depending on the given $\Sigma_n$), the resulting matrix can always be made positive definite, which suffices to ensure well-definedness overall.

\section{Choosing an Update Schedule} \label{App:schedule}

Here we summarize our findings regarding update schedules from Section \ref{Sec:par_sched} and Appendix \ref{App:param_choices} and add some more general considerations.

When using the PATT framework with $p > 1$ parallel chains, it is not advisable to update the transformation parameters in every iteration, even if the computational overhead from the parameter updates themselves is negligible. Rather, one should always run the chains for a reasonable number of iterations in between updates, in order to reduce the total time spent waiting for the slowest chain before each parameter update.

Moreover, it is usually not advisable to begin the tuning right away, i.e.~to compute the first set of transformation parameters based on very few samples, as at least the canonical parameter choices we discussed in Appendix \ref{App:param_choices} are not particularly robust towards such situations, so that transformation parameters computed from an overly small set of initial samples may well lead to such bad transformations that the method's medium-term sampling efficiency is significantly hampered.

For the simplest parameter choices, in particular the sample mean $m_n$ from \eqref{Eq:sam_mean} and the sample standard deviation $v_n$ from \eqref{Eq:sam_std_dev}, we deem it reasonable to use a linear update schedule, i.e.~$S = (s_k)_{k \in \N}$ with $s_k = a k + b$ for some ${a, b \in \N}$, chosen according to the above considerations ($a$ being the delay between consecutive updates and $b$ being the initial tuning burn-in, i.e.~the number of iterations before the first transformation is chosen). Nevertheless, we think it may be slightly more efficient overall to use an update schedule that slowly increases the delay between consecutive updates, particularly because the effect of a fixed number, say $a \cdot p$, of new samples on statistics like the sample mean and sample standard deviation diminishes as the total number of samples increases.

For slightly more extravagant parameter choices, for which efficient updates have amortized costs per iteration that can depend on the dimension $d$, but not on the number of iterations $n$, a good general guideline appears to be that a linear update schedule for these choices may be used, but that its parameters should be scaled with (powers of) $d$ in order to yield an amortized update cost per iteration that is dominated by other tasks. Specifically, with sample covariances, the computational overhead associated with the sampling itself is $\O(p \cdot d^2)$ per iteration, whereas updating the parameters costs $\O(d^3)$. Thus, to achieve good amortization, the guideline suggests using an update schedule $S = (s_k)_{k \in \N}$ for which $s_k - s_{k-1}$ is of order $\O(d)$. A natural way to achieve this would be $s_k = a d k + b$ for some $a, b \in \N$, where consecutive update times are a fixed multiple of $d$ apart.

Finally, for even more costly parameter choices like the sample median, for which an update after $n$ iterations costs $\O(n \cdot p \cdot d)$, it seems advisable to use an exponential schedule, as this leads to an amortized cost per iteration independent of $n$ (see our analysis in Appendix \ref{SubApp:sample_median}).

Though the parameter choices should play an important role in deciding for an update schedule, some characteristics of the target distribution $\nu$ may also be considered in this process. For instance, in cases where $\nu$ is known (or strongly suspected) to be heavy-tailed or very misshapen, update times for ATT should tend to be further apart than in analogous settings with light-tailed and well-shaped targets. This is because in the former settings the base sampler is much more likely to get into situations that require many target density evaluations to get out of (e.g.~walking far into the tails in a heavy-tailed setting), which, at least for base samplers using a random number of target density evaluations per iteration, increases the variance of the runtime of each chain for a given number of iterations, and therefore the time spent waiting for the slowest chain when synchronizing them for an update.

Although we expect that such nuanced considerations can enable posing even better-tuned samplers, in our experiments with PATT we have found it easier to always rely on default schedules that are parametrized solely by the selected adjustment types, the dimension $d \in \N$ of the problem at hand and the number $p \in \N$ of parallel chains to be used. These default schedules are defined as follows: Each of them is a priori infinite and only truncated by the finite number of iterations to be performed. If the PATT-sampler is to use centering with sample medians, the default schedule is
\begin{equation*}
	s_k = \lfloor 1.5^{k+16} \rfloor , 
	\qquad k \in \N
\end{equation*}
(note that $1.5^{17} \approx 10^3$). Otherwise, if the sampler is to use covariance adjustments with sample covariances, the default schedule is
\begin{equation*}
	s_k = \max(d, 25) \cdot p \cdot k ,
	\qquad k \in \N .
\end{equation*}
Finally, for all choices of adjustment types and parameter choices that fall into neither of these two categories, the default schedule is
\begin{equation*}
	s_k = 25 \cdot p \cdot k ,
	\qquad k \in \N .
\end{equation*}

\section{Connections to Other Works} \label{App:connections}

Since ATT and PATT are simple concepts, expectedly they are related, in one way or another, to various other approaches. Our aim with this section is to provide a comprehensive overview of all such connections. Let us begin with some of the stronger ones.

On the one hand, adaptive MCMC as a general principle was first popularized with its application to methods that use the adaptivity to (among other things) find a representation of the target distribution's covariance structure, see \citet{ExAda} and references therein. In contrast to ATT, these methods used the covariance information to adjust their proposal distribution, rather than transform the target. Though the two approaches can in principle lead to equivalent transition mechanisms (cf.~our analysis in Appendix \ref{App:equiv_trad_ada}), they usually differ and for some underlying samplers (e.g.~GPSS) there is no proposal distribution to be tuned and so no traditional adaptive MCMC methods equivalent to their ATT versions exist.

On the other hand, there are a number of sampling approaches that, like ATT, proceed by moving back and forth between the sample space and a latent space. Most notably, \citet{Mueller} proposed adaptively transforming the target distribution to adjust its covariance structure (by the same principle as the covariance adjustments of ATT), for the purpose of improving the performance of an underlying Metropolis-within-Gibbs sampler. 
Similarly, \citet{LovaszVempala} suggested affinely transforming a log-concave target distribution with the aim, just like in ATT, of bringing it into isotropic position, for the purpose of improving the performance of an underlying MCMC method (for which they had two particular choices in mind). Despite this idea appearing very similar to that of ATT on the surface, the two approaches' relation is not actually that strong because Lovász \& Vempala's assumption on the target to be log-concave allowed them to formulate a method which is inapplicable without the assumption. As we do not want to restrict ourselves to the log-concave case, we cannot learn much from their approach and must proceed quite differently than they propose.

Recently, some effort has been made \cite{HirdLivingston} to quantify how linearly \textit{preconditioning} the target distribution (as a one-time action, i.e.~non-adaptively) affects the theoretical properties (e.g.~mixing time, spectral gap) of an MCMC method when applying it to the transformed target instead of the untransformed one. Their considerations cover both variance and covariance adjustments (cf.~Section \ref{SubSec:adj_types}), but no centering (i.e.~their transformations are actually linear, not just affine linear). Although \citet{HirdLivingston} interpret transforming the target as an alternative viewpoint to adjusting some proposal distribution (since the MCMC methods they are interested in all have such a proposal component), many of the conclusions they draw about what types of target distributions are simplified by linear transformations, as well as some of their analysis on which transformation parameter choices are optimal for certain types of models, should apply to ATT and PATT as well.

In a slightly different direction, \citet{StereographicMCMC} proposed to use the $d$-sphere $\sph^d$ as a latent space for sampling from target distributions on $\R^d$, specifically by transforming back and forth between the spaces with a generalized stereographic projection and using, for example, a simple RWM sampler to make steps in the latent space.
Conversely, \citet{Lie} suggested a latent space based sampling scheme for sampling from target distributions on (possibly infinite-dimensional) spheres, by relying on samplers that are well-defined on Hilbert spaces (such as pCN-MH and ESS).

Most contemporary latent space based sampling schemes, however, make use of \textit{approximate transport maps}. That is, they learn (sometimes adaptively) a highly non-linear bijective transformation that aims to transform the given target distribution into a simple reference distribution, typically the standard Gaussian. Whereas the seminal work that first proposed such a scheme \cite{ParnoMarzouk} built its transformation with \textit{triangular maps}, nowadays most such schemes instead rely on \textit{normalizing flows} (NFs), see \citet{Grenioux} for a recent overview. Most notable among such approaches (in terms of their similarity to this work), \citet{TransportESS} proposed to employ NFs with ESS as the underlying base sampler and even to rely on parallelization and a type of update schedule in learning the transformation. Clearly, transforming the target distribution into a specific reference distribution is a much more challenging task than just bringing it into isotropic position (the latter being the goal of ATT). Accordingly, the transformations the approximate transport map approach relies on are costly to learn and to evaluate, leading us to speculate that, in terms of computational cost, PATT samplers should often have substantial advantages over samplers based on NFs. Unfortunately, we have found it impractical to meaningfully confirm this suspicion through numerical experiments: Because NFs are implemented by means of deep neural networks, the practical efficiency of NF-based samplers (as measured by, say, iterations per second of wall-clock time) significantly depends not only on the specific implementation of the NFs, but even on the available hardware (e.g.~on whether one is able to utilize a graphics card). Moreover, since NF-based samplers generally rely on fairly complicated transformations (as this is necessary to accomplish the approximate transport), their computational overhead (i.e.~the amount of computation spent on things other than evaluating the target) can be expected to make up a large fraction of their overall computational cost. Consequently it does not make much sense to compare them to PATT based on their respective TDE counts either. It is therefore unclear to us how a fair comparison between the two paradigms could be achieved. In any case, it is clear that there are certain types of degeneracies which NF-based samplers are far better equipped to overcome than PATT samplers. For instance, a target density which is strongly ``banana-shaped'' in some directions would always remain banana-shaped under affine transformations (which no reasonable base sampler we are aware of would be able to handle very well), whereas NFs could remove this degeneracy without issue.

To the best of our knowledge, in the present adaptive MCMC literature the focus is on samplers that rely on Metropolis-Hastings algorithms. Very few efforts have been made to harness the power of adaptive MCMC for slice sampling. Since we focus on applying PATT to slice samplers (specifically GPSS and ESS), we find it appropriate to briefly cover these related efforts.
Firstly, \citet{Tibbits14} devised an adaptive MCMC scheme to tune deterministic scan uniform slice sampling (DSUSS) \cite{SSNeal}, by letting it perform its one-dimensional updates along lines spanned by the elements of an adaptively learned basis of $\R^d$  (derived from an estimate of the target's covariance structure) rather than by those of the standard basis. We surmise that this should be equivalent (in the sense discussed for other cases in Appendix \ref{App:equiv_trad_ada}) to applying ATT with covariance adjustments to DSUSS.
Secondly, \citet{RegionalESS} proposed a scheme in which the target distribution is approximated by an adaptively learned mixture model and approximate samples from the target are generated by applying different versions of ESS to it, depending on the mixture model and the current state of the chain.

Furthermore, we want to comment on PATT's relation to \textit{generalized elliptical slice sampling} (GESS) \cite{GenEllSS}. In contrast to the two previous references, the authors of GESS did not opt for an adaptive MCMC method, but rather for a sophisticated non-adaptive one. Within their approach they choose those parameters of their sampler that would usually be chosen adaptively based just on the chain's current state, which they made feasible by introducing both parallelization and a ``two-group'' component into the method (while still viewing it as a single Markov chain overall). Though founded on entirely different principles, the resulting method has some significant parallels to PATT, in that both methods maintain a number of parallel sub-samplers on the sample space and occasionally pool the information they gather to update their parameters, for the purpose of better adapting them to the target and thereby improving their performance. We mention that PATT may have a conceptual advantage over GESS, since it is easier to scale to small numbers of parallel chains (e.g.~to accomodate users that wish to run a sampler on hardware with relatively few CPU cores). Regarding numerical comparisons between GESS and PATT, we refer to Section \ref{Sec:experiments} and Appendix \ref{App:exp_details}.

We also wish to point out two earlier works that PATT shares certain mechanims with. 
Firstly, \citet{ParAdaMCMC} suggested several adaptive MCMC methods that share PATT's basic parallelization approach, i.e.~each of these methods runs a number of parallel chains on the sample space and uses the samples generated by all chains to update the adaptively learned parameters. In contrast to PATT, \citet{ParAdaMCMC} did not consider any kind of update schedule, instead synchronizing the chains and updating the parameters in every iteration. This was feasible for them because they only considered choosing their underlying sampler as a Metropolis-Hastings method, and these methods are easy to synchronize since they only require a single target density evaluation per iteration.

Secondly, \citet{AirMCMC} proposed a new framework for adaptive MCMC, called \textit{AirMCMC}, which differs from the commonly used one by only allowing the adaptively learned parameters to be updated at a predetermined sequence $(s_k)_{k \in \N}$ of iterations. Furthermore, the sequence $(d_k)_{k \in \N}$ of differences $d_k = s_{k+1} - s_k$ is itself required to be strictly increasing, so that over time an AirMCMC method updates its parameters increasingly rarely. Obviously, the idea to only update the parameters at predetermined times has a large overlap with our concept of an update schedule. Notably, for \citet{AirMCMC} the main motivation for using such update times is that it eases theoretical analysis (compared with the general adaptive MCMC framework). The main results of \citet{AirMCMC} might therefore be of use in proving theoretical guarantees, such as ergodicity, of certain types of PATT samplers with infinite update schedules.

We should also elucidate the relation between PATT and the \textit{affine invariant ensemble samplers} of \citet{GoodmanWeare}. Generally speaking, these methods maintain a large collection of parallel chains, which they call \textit{walkers}, and update each walker by a Metropolis-Hastings step based on a small, randomly chosen subset of the other walkers. More specifically, the simplest such sampler updates a given walker by randomly selecting just one other walker and moving along the straight line through both walkers' current states. At first glance, the methods of \citet{GoodmanWeare} have some close parallels to PATT, as both utilize parallel chains and make them interact in non-trivial ways (although the ensemble samplers' chain interactions are far more direct and consequently much stronger than those of PATT). Moreover, the ensemble samplers are \textit{affine invariant}, meaning that their performance is unaffected by affine transformations of the target density (at least if the initial states are transformed accordingly). Given that PATT itself does nothing more than learning and applying a suitable affine transformation, one may therefore be inclined to question its value altogether, since it seems one could simply use an affine invariant ensemble sampler instead. However, PATT is designed to synergize with base samplers like GPSS and ESS, which (empirically) perform dimension-independently well under optimal conditions and thereby allow PATT with these base samplers to also perform dimension-independently well (at least asymptotically), even under far less restrictive conditions. The ensemble samplers of \citet{GoodmanWeare} on the other hand are known to deterioriate with increasing dimension, because their transition mechanisms are not well aligned with high-dimensional geometry \cite{emceeDegenBlog,emceeDegenPaper}.

Finally, we wish to mention that there are a number of works proposing MCMC-based sampling approaches specifically designed to perform well for heavy-tailed target distributions \cite{JohnsonGeyer,Kamatani_ht,FourierMCMC,ItoDiff_ht}. Although PATT does not have particularly strong methodological connections to any of these works, its specialization to PATT-GPSS is at least related to them in the sense that it is also designed to work well for heavy-tailed targets. Moreover, in our view, a significant advantage of PATT-GPSS over the aforementioned approaches lies in its generality, since it can be expected to work well in both light-tailed and heavy-tailed settings, while requiring only minimal setting-specific interventions by the user (namely to suitably choose the scalar hyperparameter of GPSS).

\section{ATT-friendly adaptive MCMC schemes}
\label{App:equiv_trad_ada}

Some classical adaptive MCMC schemes can be interpreted in terms of ATT. More generally, consider an adapation scenario w.r.t.~a family of transition kernels with covariance and shift parameters. For this let $\M$ be the set of all positive definite matrices in $\mathbb{R}^{d\times d}$, recall that $\A$ is the set of bijective affine transformations on $\R^d$, c.f.~Definition~\ref{Def:alpha_family}, and $\nu_{\alpha} := (\alpha^{-1})_{\#} \nu$ is the pushforward measure of $\nu$ under $\alpha^{-1}$. We introduce the following property.

\begin{definition} \label{Def:ATT-friendly}
	Let $\mathcal{P}$ be a family of probability measures on $(\R^d, \B(\R^d))$ that is closed under affine transformations in the sense that
	\begin{equation}
		\forall \nu \in \Pc, \alpha \in \A: \quad \nu_{\alpha} \in \Pc .
		\label{Eq:closed_wrt_affine}
	\end{equation}
	A family $(K_{\nu,c,\Sigma})_{(\nu,c,\Sigma)\in \Pc \times \R^d \times \M}$ of transition kernels on $\R^d \times \B(\R^d)$ that satisfies
	\begin{equation}
		\forall (\nu,c,\Sigma)\in \Pc \times \R^d \times \M: \quad \nu K_{\nu,c,\Sigma} = \nu
		\label{Eq:class_inv_prop}
	\end{equation}
	is called \textit{ATT-friendly}, if for any $(\nu,c,\Sigma) \in \Pc \times \R^d \times \M$ there is an $\alpha \in \A$ such that
	\begin{equation}
		K_{\nu,c,\Sigma}(x,A) = K_{\nu_{\alpha},\fatzero,I_d}(\alpha^{-1}(x),\alpha^{-1}(A))
		\label{Eq:ATT-friendly}
	\end{equation}
	for all $x \in \R^d, A \in \mathcal{B}(\R^d)$, where $\fatzero \in \R^d$ is the zero-vector and $I_d \in \R^{d\times d}$ the identity matrix.
\end{definition}

Observe that, by \eqref{Eq:transition_kernel_formula}, the identity \eqref{Eq:ATT-friendly} can be interpreted as saying that $K_{\nu,c,\Sigma}$ is the transition kernel of ATT for target distribution $\nu$ with fixed transformation $\alpha$ and the base sampler with transition kernel $K_{\nu_{\alpha},\fatzero,I_d}$ on the latent space. Moreover, an adaptive MCMC scheme based on an ATT-friendly family of transition kernels can, by \eqref{Eq:ATT-friendly}, be interpreted as updating a linear transformation whenever the pair of parameters $(c,\Sigma) \in \R^d \times \M$ is updated.

In the following, we establish the ATT-friendliness of a number of classical MCMC schemes in an exemplary, case-by-case fashion. This allows us to formally establish equivalences between adaptive implementations of ATT and certain adaptive MCMC versions of these classical schemes.

\begin{example}[Random walk Metropolis] \label{Ex:RWM}
	Let $\Pc$ be the family of distributions on $(\R^d,\B(\R^d))$ that admit a strictly positive Lebesgue density. Note that this $\Pc$ trivially satisfies \eqref{Eq:closed_wrt_affine}. Let $\nu \in \Pc$ and denote by ${\varrho: \R^d \ra \ooint{0}{\infty}}$ its density. The transition kernel of the \textit{random walk Metropolis} (RWM) algorithm for $\nu$ with covariance parameter $\Sigma \in \M$ and fixed step size $\sigma > 0$ is given by
	\begin{align*}
		&M_{\nu,\Sigma}(x,A)
		= \int_{A} \min\!\left\{1,\frac{\varrho(\tilde{x})}{\varrho(x)}\right\}\Nc_d(x, \sigma^2 \Sigma)(\d \tilde{x}) \\
		&\quad+ \ind_A(x)\left(1-\int_{\R^d} \min\!\left\{1,\frac{\varrho(\tilde{x})}{\varrho(x)}\right\}\Nc_d(x, \sigma^2 \Sigma)(\d \tilde{x})\right).
	\end{align*}
	The family $(M_{\nu,\Sigma})_{(\nu,\Sigma) \in \Pc \times \M}$ is one of the standard families of transition kernels that serve as building blocks for adaptive MCMC, c.f.~\citet{ExAda}. For example, multiple classical adaptive MCMC methods (e.g.~\citet{HaarioAdaRWM,ExAda}) follow the basic idea to perform the $n$-th iteration by $M_{\nu,\Sigma_n}$, where $\Sigma_n$ is the sample covariance of the first $n$ samples $x_0, \dots, x_{n-1}$ (usually $\Sigma_n$ is slightly regularized to enable better theoretical guarantees).

	By introducing a dummy index $c \in \R^d$ that the kernels do not actually depend on, we obtain the family $M := (M_{\nu,\Sigma})_{(\nu,c,\Sigma)\in\mathcal{P} \times \R^d \times \M}$ that fits the format of the class in Definition \ref{Def:ATT-friendly}. As it is well-known that $\nu M_{\nu,\Sigma} = \nu$ for all $(\nu,\Sigma) \in \Pc \times \M$, this $M$ also satisfies the invariance requirement \eqref{Eq:class_inv_prop}. We may therefore examine the ATT-friendliness of $M$.

	Let ${\nu \in \Pc}$, ${\Sigma \in \M}$ and write $\Sigma$ in Cholesky decomposition as ${\Sigma = L L^{\top}}$ for some $L \in \GL_d(\R)$. Then for $\alpha \in \A$ with $\alpha(y) := L y$, recalling that $\nu_{\alpha}$ has density $\varrho_{\alpha}(y) = \varrho(\alpha(y))$ and using the substitution $z := L \tilde{x}$, we get for any $A\in \mathcal{B}(\R^d)$ and $x\not\in A$ that
	\begin{align*}
		& M_{\nu_{\alpha},I_d}(\alpha^{-1}(x),\alpha^{-1}(A)) 
		=  M_{\nu_{\alpha},I_d}(L^{-1}x,L^{-1}(A)) \\ 
		& = \int_{L^{-1}(A)} \min\!\left\{1,\frac{\varrho(L \tilde{x})}{\varrho(L L^{-1}x)}\right\} \Nc_d(L^{-1}x, \sigma^2 I_d)(\d \tilde{x}) \\ 
		& = \int_{A} \min\!\left\{1,\frac{\varrho(z)}{\varrho(x)}\right\} \Nc_d(x, \sigma^2 LL^{\top})(\d z) 
		= M_{\nu,\Sigma}(x,A), 
	\end{align*}
	where the 2nd last equality follows by well-known properties of Gaussian distributions. Now for arbitrary $x \in \R^d$, $A \in \mathcal{B}(\R^d)$, this yields
	\begin{align*}
		& M_{\nu,\Sigma}(x,A) 
		= M_{\nu,\Sigma}(x,A\setminus\{x\}) + \ind_A(x) M_{\nu,\Sigma}(x,\{x\}) \\ 
		& =  M_{\nu,\Sigma}(x,A\setminus\{x\}) + \ind_A(x) \big(1-M_{\nu,\Sigma}(x,\R^d\setminus\{x\})\big) \\ 
		&= M_{\nu_{\alpha},I_d}(\alpha^{-1}(x), \alpha^{-1}(A \setminus \{x\})) \\
		&\quad + \ind_A(x) \big(1 - M_{\nu_{\alpha},I_d}(\alpha^{-1}(x), \alpha^{-1}(\R^d \setminus \{x\}))\big) \\ 
		&= M_{\nu_{\alpha},I_d}(\alpha^{-1}(x), \alpha^{-1}(A) \setminus \{\alpha^{-1}(x)\}) \\
		&\quad + \ind_A(x) M_{\nu_{\alpha},I_d}(\alpha^{-1}(x), \{\alpha^{-1}(x)\}) \\ 
		& = M_{\nu_{\alpha},I_d}(\alpha^{-1}(x),\alpha^{-1}(A)). 
	\end{align*}
	That is, $M$ also satisfies the ATT-friendliness property \eqref{Eq:ATT-friendly}, so that RWM is ATT-friendly. In particular, the aforementioned adaptive MCMC schemes for RWM based on the sample covariance $\Sigma_n$ are equivalent to the corresponding adaptive ATT schemes, in the sense that their respective transition kernels coincide.
\end{example}

\begin{example}[Independent Metropolis-Hastings] \label{Ex:IMH}
	As in the previous example, let $\Pc$ be the family of distributions on $(\R^d,\B(\R^d))$ that admit a strictly positive Lebesgue density.
	For $\tau \in \R^d$, $\Pi \in \M$, denote by $x \mapsto \Nc_d(x; \tau, \Pi)$ the p.d.f.~of $\Nc_d(\tau,\Pi)$. Let $\sigma > 0$ be a fixed step size parameter, let $(\nu,c,\Sigma) \in \Pc \times \R^d \times \M$ and denote by $\varrho: \R^d \ra \ooint{0}{\infty}$ the density of $\nu$. Define acceptance probabilities by
	\begin{equation*}
		\beta_{\nu,c,\Sigma}(x, \tilde{x})
		:= \min\!\Big\{1,\frac{\varrho(\tilde{x}) \Nc_d(x; c, \sigma^2 \Sigma)}{\varrho(x) \Nc_d(\tilde{x}; c, \sigma^2 \Sigma)}\Big\}
	\end{equation*}
	for $x, \tilde{x} \in \R^d$. With that the transition kernel of \textit{independent Metropolis-Hastings} (IMH) \textit{with Gaussian proposal} can be written as
	\begin{align*}
		& M_{\nu,c,\Sigma}(x,A) \\
		&= \int_A \beta_{\nu,c,\Sigma}(x, \tilde{x}) \Nc_d(c, \sigma^2 \Sigma)(\d \tilde{x}) \\
		&\quad +\ind_A(x) \Big(1-\int_{\R^d} \beta_{\nu,c,\Sigma}(x, \tilde{x}) \Nc_d(c, \sigma^2 \Sigma)(\d \tilde{x})\Big) .
	\end{align*}
	Of course there are rather obvious ways to construct adaptive MCMC methods based on the above transition kernel, for example by performing the $n$-th transition by $M_{\nu,c_n,\Sigma_n}$, where $c_n$ and $\Sigma_n$ are sample mean and sample covariance of the first $n$ samples $x_0,\dots x_{n-1}$. However, the resulting method does not appear to be in use in any real world applications, nor does it seem to have been specifically studied from the theoretical side. We surmise the reason for this to be a combination of two main factors. On the one hand, adaptive MCMC methods generally work poorly if they are overly slow at exploring the target distribution in the early stages of adaptation, because they tend to severely overfit their adaptation to what they saw near the method's initial state. It is easily seen that IMH has a natural tendency to exhibit such behavior. On the other hand, a Gaussian can hardly be a particularly precise approximation to any truly challenging target distribution, and so even a perfectly adapted IMH with Gaussian proposal will likely not work all that well in practice.

	We note that the theoretical properties of adaptive IMH with Gaussian proposal may nevertheless be of interest, for example to better understand properties of other methods. This may apply to the closely related and only slightly more complicated adaptive IMH sampler proposed by \citet{AdaIMH_GMM}, which draws its proposals from a Gaussian mixture rather than just a simple Gaussian.

	In any case, let us now examine the ATT-friendliness of IMH with Gaussian proposal. As in the previous example, it is well known that the family of its transition kernels satisfies \eqref{Eq:class_inv_prop}.

	Again write $\Sigma = L L^{\top}$ and recall $\varrho_{\alpha}(y) = \varrho(\alpha(y))$.
	For $\alpha\in\A$ with $\alpha(y) = L y + c$, so that $\alpha^{-1}(x) = L^{-1}(x - c)$, we have for any $A\in \mathcal{B}(\R^d)$ and $x\not\in A$ by the substitution $z := \alpha(\tilde{x})$ that
	\begin{align*}
		& M_{\nu_{\alpha},\fatzero,I_d}(\alpha^{-1}(x),\alpha^{-1}(A)) \\
		=& \int_{\alpha^{-1}(A)} \min\!\Big\{1,\frac{\varrho(\alpha(\tilde{x})) \Nc_d(\alpha^{-1}(x); \fatzero, \sigma^2 I_d)}{\varrho(x) \Nc_d(\tilde{x}; \fatzero, \sigma^2 I_d)}\Big\} \\
		& \qquad\qquad\qquad\qquad\qquad\qquad\qquad \cdot \Nc_d(\fatzero, \sigma^2 I_d)(\d \tilde{x}) \\ 
		=& \int_A \min\Big\{1,\frac{\varrho(z) \Nc_d(\alpha^{-1}(x); \fatzero, \sigma^2 I_d)}{\varrho(x) \Nc_d(\alpha^{-1}(z); \fatzero, \sigma^2 I_d)} \Big\} \Nc_d(c, \sigma^2 \Sigma)(\d z) \\ 
		=& \int_A \min\Big\{1,\frac{\varrho(z) \Nc_d(x; c, \sigma^2 \Sigma)}{\varrho(x) \Nc_d(z; c, \sigma^2 \Sigma)} \Big\} \Nc_d(c, \sigma^2 \Sigma)(\d z) \\
		=& \; M_{\nu,c,\Sigma}(x,A),
	\end{align*}
	where in the 2nd last equality we used the fact that
	\begin{align}
		\begin{split}
		&\Nc_d(L^{-1}(x^{\prime} - c); \fatzero, \sigma^2 I_d) \\
		&= \det(\Sigma)^{1/2} \Nc_d(x^{\prime}; c, \sigma^2 \Sigma) ,
		\quad x^{\prime} \in \R^d ,
		\end{split}
		\label{Eq:Gaussian_transform}
	\end{align}
	which is straightforward to verify using the densities' definitions.


	From this we obtain, by the same arguments as in Example~\ref{Ex:RWM}, that the family $(M_{\nu,c,\Sigma})_{(\nu,c,\Sigma)\in \mathcal{P}\times\R^d\times\M}$ of IMH kernels with Gaussian proposal is ATT-friendly.
\end{example}

\begin{remark}
	We wish to emphasize that the transformation properties of Gaussians that we made use of in Examples \ref{Ex:RWM} and \ref{Ex:IMH} are also satisfied by multivariate $t$-distributions, that is, by all distributions $\xi$ on $(\R^d, \B(\R^d))$ that have (unnormalized) densities $\eta: \R^d \ra \ooint{0}{\infty}$ of the form
	\begin{equation*}
		\eta(x) 
		= \left(1 + \frac{1}{\gamma} (x - \tau)^{\top} \Pi^{-1} (x - \tau)\right)^{-(d + \gamma)/2}
	\end{equation*}
	for any degrees-of-freedom parameter $\gamma \in \ooint{0}{\infty}$, center $\tau \in \R^d$ and scale matrix $\Pi \in \M$.

	Consequently, if we were to modify the RWM and IMH kernels of Examples \ref{Ex:RWM} and \ref{Ex:IMH} to use multivariate $t$-distributions (with a fixed degrees-of-freedom parameter) in place of Gaussians, some small adjustments to the examples' proofs would show that the resulting families of kernels are again ATT-friendly. Obviously the same holds for any other proposal distribution that has the necessary transformation properties.
\end{remark}

\begin{example}[General purpose elliptical slice sampling] \label{Ex:GP-ESS}
	Let $\Pc$ now denote the set of distributions on $(\R^d, \B(\R^d))$ that admit a strictly positive and lower semi-continuous (lsc) Lebesgue density. Note that this class satisfies \eqref{Eq:closed_wrt_affine} because $\nu_{\alpha}$ admits the density $\varrho_{\alpha} = \varrho \circ \alpha$, which is the composition of the lsc function $\varrho$ and the continuous function $\alpha$, and it is well-known that any such composition is again lsc.
	
	In its original formulation \cite{EllipticalSS}, elliptical slice sampling (ESS) is only capable of targeting distributions whose densities have a mean-zero Gaussian factor. However, as pointed out by \citet{GenEllSS}, by artificially introducing such a factor it is easily generalized to enable targeting any $\nu \in \Pc$. For lack of an established term in the literature, we call the resulting method \textit{general purpose elliptical slice sampling} (GP-ESS).
	
	Consistently with \citet{GenEllSS}, we define GP-ESS as follows. Let $\nu \in \Pc$ and denote by $\varrho: \R^d \ra \ooint{0}{\infty}$ its density. We factorize $\varrho$ as
	\begin{align*}
		&\varrho(x) = \varphi^{(0)}_{c,\Sigma}(x) \varphi^{(1)}_{c,\Sigma}(x; \varrho) , \\
		& \varphi^{(0)}_{c,\Sigma}(x) = \Nc_d(x;c,\Sigma) , \\
		& \varphi^{(1)}_{c,\Sigma}(x; \varrho) = \Nc_d(x;c,\Sigma)^{-1} \varrho(x) ,
	\end{align*}
	where $c \in \R^d$ and $\Sigma = L L^{\top} \in \M$ is a positive definite covariance matrix and its Cholesky decomposition, as before.
	For $x, v \in \R^d$ define $\widetilde{p}_{x,v,c} \colon \R \to \R^d$ by
	\begin{equation*}
		\widetilde{p}_{x,v,c}(\omega)
		:= \cos(\omega) (x - c) + \sin(\omega) (v - c) + c
	\end{equation*}
	and let $p_{x,v,c} := \widetilde{p}_{x,v,c}|_{\!\coint{0}{2 \pi}}$ be its restriction to the interval $\coint{0}{2 \pi}$.
	Using the previously introduced functions, a transition of GP-ESS is given by Algorithm~\ref{Alg:GP-ESS}. For $c=0$ it corresponds to ESS targeting $\nu$ by viewing it as having density $\varphi_{\fatzero,\Sigma}^{(1)}(\,\cdot\,; \varrho)$ w.r.t.~the Gaussian reference measure $\Nc_d(
	\fatzero,\Sigma)$.
	
	\begin{algorithm}[tb]
		\caption{GP-ESS transition}
		\label{Alg:GP-ESS}
		\textbf{Input:} target density $\varrho: \R^d \ra \ooint{0}{\infty}$, mean $c \in \R^d$, covariance $\Sigma \in \M$, current state $x_{n-1} \in \R^d$ \\
		\textbf{Output:} new state $x_n \in \R^d$
		\begin{algorithmic}[1]
			\STATE Draw $T_n \sim \U(\ooint{0}{\varphi^{(1)}_{c,\Sigma}(x_{n-1}; \varrho)})$, call the result $t_n$.
			\STATE Draw $V_n \sim \Nc(c,\Sigma)$, call the result $v_n$.
			\STATE Draw $\Omega \sim \U(\!\coint{0}{2 \pi})$, call the result $\omega$.
			\STATE Set $\omega_{\min} := \omega - 2 \pi$ and $\omega_{\max} := \omega$.
			\WHILE{$\varphi^{(1)}_{c,\Sigma}(\widetilde{p}_{x_{n-1},v_n,c}(\gamma); \varrho) \leq t$}
			\STATE {\bfseries if} $\omega < 0$ {\bfseries then} set $\omega_{\min} := \omega$ {\bfseries else} set $\omega_{\max} := \omega$.
			\STATE Draw $\Omega \sim \U(\!\coint{\omega_{\min}}{\omega_{\max}})$, call the result $\omega$.
			\ENDWHILE
			\STATE  \textbf{return} $x_n := \widetilde{p}_{x_{n-1},v_n,c}(\omega)$.
		\end{algorithmic}
	\end{algorithm}
	
	An explicit expression of the corresponding transition kernel is not readily available. Therefore we rely on the reformulation of ESS in Algorithm~{2.2} and Algorithm~{2.3} of \citet{RevESS} and their notation. For this let
	\begin{equation*}
		L^{\nu}_{c,\Sigma}(t)
		:= \{ x \in \R^d \colon \varphi^{(1)}_{c,\Sigma}(x; \varrho) > t \}
	\end{equation*}
	be the level set of $\varphi^{(1)}_{c,\Sigma}(\cdot; \varrho)$ at level $t\in\ooint{0}{\infty}$.
	This allows us to state the transition kernel of GP-ESS for $\nu$ with parameters $c$, $\Sigma$ as
	\begin{align*}
		&E_{\nu,c,\Sigma}(x,A) 
		= \frac{1}{\varphi^{(1)}_{c,\Sigma}(x; \varrho)} \int_0^{\varphi^{(1)}_{c,\Sigma}(x; \varrho)} \int_{\R^d} \\
		&\; Q_{p_{x,v,c}^{-1}(L^{\nu}_{c,\Sigma}(t))}(0,p^{-1}_{x,v,c}(A\cap L^{\nu}_{c,\Sigma}(t)))\, \Nc_d(c,\Sigma)(\d v) \, \d t,
	\end{align*}
	for $x \in \R^d$, $A \in \B(\R^d)$. Here $Q_S$, for $S \in \B(\!\ccint{0}{2 \pi})$, is a transition kernel on $(S,\B(S))$ that corresponds to the shrinkage procedure given by Algorithm~{2.3} of \citet{RevESS}, which is equivalent to performing lines $3$--$8$ of Algorithm~\ref{Alg:GP-ESS}. \citet{RevESS} proved that $Q_S$ is reversible w.r.t.~the uniform distribution on $S$ for any non-empty, open $S$.
	Building on this, \citet{RevESS} proved in their Theorem~3.2 that the classical ESS kernel $E_{\nu_{\alpha},\fatzero,I_d}$ is reversible w.r.t.~$\nu_{\alpha}$ (note that this is where the lower semi-continuity of distributions in $\Pc$ is needed). Therefore by Proposition~19 of \citet{Robust} the $\alpha$ push-forward kernel
	\begin{equation*}
		E_{\nu,c,\Sigma}(x,A)
		= E_{\nu_{\alpha},\fatzero,I_d}(\alpha^{-1}(x),\alpha^{-1}(A))
	\end{equation*}
	is reversible w.r.t.~$\nu$. In particular, this implies $\nu E_{\nu,c,\Sigma} = \nu$, so that the family $E := (E_{\nu,c,\Sigma})_{(\nu,c,\Sigma) \in \Pc \times \R^d \times \M}$ of GP-ESS kernels satisfies the invariance requirement \eqref{Eq:class_inv_prop} of the ATT-friendliness definition.
	
	We now establish the ATT-friendliness of GP-ESS. For any $(\nu,c,\Sigma) \in \Pc \times \R^d \times \M$, we write again $\Sigma$ in Cholesky decomposition as $\Sigma=L L^{\top}$, set $\alpha(y) := L y + c$, so that $\alpha^{-1}(x) = L^{-1}(x - c)$, and note that
	\begin{equation*}
		\varrho_{\alpha}(y)
		= \Nc_d(y; \fatzero,I_d) \varphi_{\fatzero,I_d}^{(1)}(y; \varrho_{\alpha}) ,
		\quad y \in \R^d .
	\end{equation*}
	Moreover, by \eqref{Eq:Gaussian_transform} (with $\sigma := 1$) we have
	\begin{equation}
		\frac{\varphi^{(1)}_{\fatzero,I_d}(\alpha^{-1}(\tilde{x}); \varrho_{\alpha})}{\varphi^{(1)}_{\fatzero,I_d}(\alpha^{-1}(x); \varrho_{\alpha})} 
		= \frac{\varphi^{(1)}_{c,\Sigma}(\tilde{x}; \varrho)}{\varphi^{(1)}_{c,\Sigma}(x; \varrho)} ,
		\quad x, \tilde{x} \in \R^d .
		\label{Eq:Gaussian_transform_ratio}
	\end{equation}
	With this, we obtain
	\begin{align*}
		& \alpha(L^{\nu_{\alpha}}_{\fatzero,I_d}(u\cdot \varphi^{(1)}_{\fatzero,I_d}(\alpha^{-1}(x); \varrho_{\alpha}))) \\
		=& \; \alpha( \{ y \in \R^d \colon \varphi^{(1)}_{\fatzero,I_d}(y); \varrho_{\alpha}) > u \cdot \varphi^{(1)}_{\fatzero,I_d}(\alpha^{-1}(x); \varrho_{\alpha}) \} ) \\ 
		=& \; \{ \tilde{x} \in \R^d \colon \varphi^{(1)}_{\fatzero,I_d}(\alpha^{-1}(\tilde{x}); \varrho_{\alpha}) > u \cdot \varphi^{(1)}_{\fatzero,I_d}(\alpha^{-1}(x); \varrho_{\alpha}) \} \\ 
		\stackrel{\eqref{Eq:Gaussian_transform_ratio}}{=} & \;
		\{ \tilde{x} \in \R^d \colon \varphi^{(1)}_{c,\Sigma}(\tilde{x}; \varrho) > u \cdot \varphi^{(1)}_{c,\Sigma}(x; \varrho) \} \\ 
		=& \; L^{\nu}_{c,\Sigma}(u\cdot \varphi^{(1)}_{c,\Sigma}(x; \varrho)) 
	\end{align*}
	for any $x \in \R^d, A \in \B(\R^d)$, such that
	\begin{align}
		\begin{split}
			&L^{\nu_{\alpha}}_{\fatzero,I_d}(u\cdot \varphi^{(1)}_{\fatzero,I_d}(\alpha^{-1}(x); \varrho_{\alpha})) \\
			&= \alpha^{-1}(L^{\nu}_{c,\Sigma}(u\cdot \varphi^{(1)}_{c,\Sigma}(x; \varrho)) ,
			\quad x \in \R^d, u \in \ooint{0}{1} .	
		\end{split}
		\label{Eq:level_set_id}
	\end{align}
	Temporarily fix $x, v \in \R^d$, $\omega \in \coint{0}{2 \pi}$, then it is easy to see that $\alpha(p_{\alpha^{-1}(x),\alpha^{-1}(v),\fatzero}(\omega)) = p_{x,v,c}(\omega)$
	and consequently we get for any $A\in\mathcal{B}(\R^d)$ that
	\begin{align}
		&p^{-1}_{\alpha^{-1}(x),\alpha^{-1}(v),\fatzero}(\alpha^{-1}(A)\cap L^{\nu_{\alpha}}_{\fatzero,I_d}(u\cdot \varphi^{(1)}_{\fatzero,I_d}(\alpha^{-1}(x); \varrho_{\alpha}))) \notag \\
		&\stackrel{\eqref{Eq:level_set_id}}{=}
		p^{-1}_{\alpha^{-1}(x),\alpha^{-1}(v),\fatzero}(\alpha^{-1}(A\cap L^{\nu}_{c,\Sigma}(u\cdot \varphi^{(1)}_{c,\Sigma}(x; \varrho)))) \notag \\ 
		&= p_{x,v,c}^{-1}(A\cap L^{\nu}_{c,\Sigma}(u\cdot \varphi^{(1)}_{c,\Sigma}(x; \varrho))) \label{Eq:p_id} .
	\end{align}
	Moreover, for any $x \in \R^d$, $t \in \ooint{0}{\varphi^{(1)}_{\fatzero,I_d}(\alpha^{-1}(x); \varrho_{\alpha})}$ and ${A \in \B(\R^d)}$, we get by the substitution $v := \alpha(\widetilde{v})$ that
	\begin{align*}
		&\int_{\R^d} Q_{p_{\alpha^{-1}(x),\widetilde{v},\fatzero}^{-1} (L^{\nu_{\alpha}}_{\fatzero,I_d}(t))}(0, \\
		&\qquad p^{-1}_{\alpha^{-1}(x),\widetilde{v},\fatzero}(\alpha^{-1}(A) \cap L^{\nu_{\alpha}}_{\fatzero,I_d}(t))) \, \Nc_d(\fatzero,I_d)(\d \widetilde{v}) \\
		&= \int_{\R^d} Q_{p_{\alpha^{-1}(x),\alpha^{-1}(v),\fatzero}^{-1}(L^{\nu_{\alpha}}_{\fatzero,I_d}(t))}(0, \\
		&\qquad p^{-1}_{\alpha^{-1}(x),\alpha^{-1}(v),\fatzero}(\alpha^{-1}(A) \cap L^{\nu_{\alpha}}_{\fatzero,I_d}(t))) \, \Nc_d(c,\Sigma)(\d v) .
	\end{align*}
	Using the substitution $t(u) := u\cdot \varphi^{(1)}_{\fatzero,I_d}(\alpha^{-1}(x); \varrho_{\alpha})$ and the above identity, we finally obtain
	\begin{align*}
		&E_{\nu_{\alpha},\fatzero,I_d}(\alpha^{-1}(x),\alpha^{-1}(A)) \\
		&= \int_0^1 \int_{\R^d} Q_{p^{-1}_{\alpha^{-1}(x),\alpha^{-1}(v),\fatzero} (L^{\nu_{\alpha}}_{\fatzero,I_d}(t(u)))}(0, \\
		&\qquad\qquad\qquad p^{-1}_{\alpha^{-1}(x),\alpha^{-1}(v),\fatzero}(\alpha^{-1}(A) \cap L^{\nu_{\alpha}}_{\fatzero,I_d}(t(u)))) \\
		&\qquad \Nc_d(c,\Sigma)(\d v) \, \d u \\
		&\stackrel{\eqref{Eq:p_id}}{=} \int_0^{1} \int_{\R^d} Q_{p_{x,v,c}^{-1} (L^{\nu}_{c,\Sigma}(u \cdot \varphi^{(1)}_{c,\Sigma}(x; \varrho)))}(0, \\
		&\qquad\qquad\qquad\quad p_{x,v,c}^{-1}(A \cap L^{\nu}_{c,\Sigma}(u \cdot \varphi^{(1)}_{c,\Sigma}(x; \varrho))))\,\\
		& \qquad \Nc_d(c,\Sigma)(\d v) \, \d u\\
		& = E_{\nu,c,\Sigma}(x,A) .
	\end{align*}
	Hence the family $E$ of GP-ESS kernels has the ATT-friendliness property \eqref{Eq:ATT-friendly}. Since, as explained earlier, $E$ also satisfies \eqref{Eq:class_inv_prop} we have thus proven that GP-ESS is ATT-friendly. As in the other examples, this in particular shows an equivalence between the obvious adaptive MCMC schemes for GP-ESS and the corresponding adaptive ATT schemes in the sense of coinciding transition kernels.
	
	We note, however, that these obvious adaptive MCMC schemes have not actually been proposed and/or examined in any of the existing literature: \citet{GenEllSS}, who first wrote about GP-ESS, did not consider how to adaptively choose its artificial Gaussian factor. Instead, they modified GP-ESS, essentially by introducing a new auxiliary variable, and then focused on how to automatically choose the new methods parameters. In the, to the best of our knowledge, only other published work regarding both GP-ESS and adaptive approaches, \citet{RegionalESS} again further generalized GP-ESS instead of focusing on the method itself.

	However, this does not mean that choosing the parameters of GP-ESS by standard adaptive MCMC principles, e.g.~setting $c$ as the sample mean and $\Sigma$ as the sample covariance of available samples, results in a poor sampling method. In fact, in our experiments with PATT-ESS (which, by the ATT-friendliness of GP-ESS are largely equivalent to ones with standard adaptive GP-ESS), it performed remarkably well for sufficiently well-behaved targets.

\end{example}

\section{Proof for Theoretical Justification} \label{App:theo_just_proof}

Here we prove the theoretical result of Section \ref{Sec:theo_just}. To simplify the proof,
we first provide a number of small auxiliary results and some new notation.

For the moment, let us arbitrarily fix the target distribution $\nu$ with unnormalized density $\varrho$ and a fixed transformation ${\alpha \in \A}$ (cf.~Def.~\ref{Def:alpha_family}), say $\alpha(y) = W y + c$ for some ${W \in \GL_d(\R)}$ and $c \in \R^d$.

Since $\alpha$ has constant Jacobian $J_{\alpha}(z) = W$, the multivariate change of variables formula implies
\begin{equation}
	\int_{\R^d} f(\alpha(y)) \d y
	= \frac{1}{\abs{\det(W)}} \int_{\R^d} f(x) \d x
	\label{Eq:affine_change_of_vars}
\end{equation}
for any integrable $f: \R^d \ra \R$.
%
Set $\kappa := \int_{\R^d} \varrho(x) \d x$, so that we may write $\nu(\d x) = \kappa^{-1} \varrho(x) \d x$. Observing that
\begin{align*}
	\kappa_{\alpha}
	&:= \int_{\R^d} \varrho_{\alpha}(y) \d y
	= \int_{\R^d} \varrho(\alpha(y)) \d y \\
	&\stackrel{\eqref{Eq:affine_change_of_vars}}{=} 
	\frac{1}{\abs{\det(W)}} \int_{\R^d} \varrho(x) \d x
	= \frac{\kappa}{\abs{\det(W)}} ,
\end{align*}
we obtain
\begin{equation*}
	\nu_{\alpha}(\d y) 
	= \kappa_{\alpha}^{-1} \varrho_{\alpha}(y) \d y
	= \frac{\abs{\det(W)}}{\kappa} \varrho_{\alpha}(y) \d y .
\end{equation*}

For any probability measure $\pi$ and any $\pi$-integrable function $f$ we denote by $\pi(f)$ the expectation of $f$ under $\pi$, i.e.
\begin{equation*}
	\pi(f) := \int_{\R^d} f(x) \pi(\d x) .
\end{equation*}
The following simple identity is useful.
\begin{lemma} \label{Lem:int_id}
	For any integrable function $h: \R^d \ra \R$ one has $\nu(h) = \nu_{\alpha}(h \circ \alpha)$.
\end{lemma}
\begin{proof}
	We have
	\begin{align*}
		\nu(h)
		&= \int_{\R^d} h(x) \nu(\d x) \\
		&= \frac{1}{\kappa} \int_{\R^d} h(x) \varrho(x) \d x \\
		&\stackrel{\eqref{Eq:affine_change_of_vars}}{=} \frac{\abs{\det(W)}}{\kappa} \int_{\R^d} h(\alpha(y)) \varrho(\alpha(y)) \d y \\
		&= \int_{\R^d} (h \circ \alpha)(y) \nu_{\alpha}(\d y) \\
		&= \nu_{\alpha}(h \circ \alpha) . \qedhere 
	\end{align*}
\end{proof}

We now move on to the proof of our main theoretical result.

\begin{proof}[Proof of Theorem \ref{Thm:finite_US_ergodicity}]
	Denote by $\varsigma := s_{n_I}$ the index of the iteration in which the PATT sampler performs its final transformation parameter update. Let $(X_i^{(j)})_{0 \leq i < \varsigma, 1 \leq j \leq p}$ be the samples generated by the sampler up to iteration $\varsigma$. The main idea of our proof is to condition on these samples and analyze the sampler's structure after that point.
	
	Given $X_i^{(j)} = x_i^{(j)}$ for all $0 \leq i < \varsigma$ and $1 \leq j \leq p$, the remaining PATT sampling scheme after iteration $\varsigma$ is quite simple: Because we are past the final parameter update, the parallel chains cease to interact with one another and each ``remainder chain'' (i.e.~each chain viewed as starting at iteration $\varsigma$ of the PATT sampler) becomes a homogeneous Markov chain. Furthermore, the transition mechanism of each remainder chain fits into the ATT framework of Section \ref{Sec:ATT} with fixed transformation $\alpha_{\varsigma} \in \A$ given by
	\begin{align*}
		\alpha_{\varsigma}(y) &:= W_{\varsigma} y + c_{\varsigma} , \\
		W_{\varsigma} &:= \mathbf{W}_{\varsigma}(x_0^{(1)},\dots,x_{\varsigma-1}^{(p)}) ,\\
		c_{\varsigma} &:= \mathbf{c}_{\varsigma}(x_0^{(1)},\dots,x_{\varsigma-1}^{(p)}) .
	\end{align*}
	In particular, the chains rely on a fixed latent space on which they use the kernel $P_{\alpha_{\varsigma}}$ to approximately sample from the transformed target distribution $\nu_{\alpha_{\varsigma}}$. Note that, by assumption, $P_{\alpha_{\varsigma}}$ is ergodic towards $\nu_{\alpha_{\varsigma}}$ in the total variation sense \eqref{Eq:tv_ergodicity}. This is known (see e.g.~\citet{Tierney}, Section 3.1) to imply that $P_{\alpha_{\varsigma}}$ leaves $\nu_{\alpha_{\varsigma}}$ invariant, is $\nu_{\alpha_{\varsigma}}$-irreducible, aperiodic and positive Harris recurrent (we refer to \citet{Tierney} for definitions of these properties). It is therefore ergodic in the sense of \citet{Tierney}, Section 3.2, which by \citet{Tierney}, Theorem 3 implies a strong law of large numbers. Specifically, it ensures that a Markov chain $(Y_n)_{n \in \N_0}$ with transition kernel $P_{\alpha_{\varsigma}}$ satisfies
	\begin{equation}
		\frac{1}{n} \sum_{i=0}^{n-1} g(Y_i) 
		\asconv \nu_{\alpha_{\varsigma}}(g)
		\label{Eq:SLLN_latent_general}
	\end{equation}
	for any $\nu_{\alpha_{\varsigma}}$-integrable function $g: \R^d \ra \R$ and any initial value $Y_0 = y_0$.
	
	Let $f: \R^d \ra \R$ be an arbitrary $\nu$-integrable function and observe that then $g := f \circ \alpha_{\varsigma}$ is $\nu_{\alpha_{\varsigma}}$-integrable, since
	\begin{equation*}
		\nu_{\alpha_{\varsigma}}(\abs{g})
		= \nu_{\alpha_{\varsigma}}(\abs{f} \circ \alpha_{\varsigma})
		= \nu(\abs{f})
		< \infty 
	\end{equation*}
	by Lemma \ref{Lem:int_id}, applied with $h := \abs{f}$.
	Recall that, by construction of ATT, the latent space samples $Y_i^{(j)}$ for $i \geq \varsigma$, $1 \leq j \leq p$ generated by the parallel chains of the PATT sampler after iteration $\varsigma$ are mapped to their sample space counterparts by $X_i^{(j)} := \alpha_{\varsigma}(Y_i^{(j)})$, so that in particular $f(X_i^{(j)}) = f(\alpha_{\varsigma}(Y_i^{(j)})) = g(Y_i^{(j)})$. Thus, by \eqref{Eq:SLLN_latent_general}, the continuous mapping theorem for almost sure convergence, and our earlier observations about the structure of the PATT sampler after its final update, we obtain
	\begin{align*}
		&\frac{1}{(n - \varsigma) p} \sum_{i=\varsigma}^{n-1} \sum_{j=1}^p f(X_i^{(j)}) \\
		&= \frac{1}{(n - \varsigma) p} \sum_{i=\varsigma}^{n-1} \sum_{j=1}^p g(Y_i^{(j)})
		\asconv \nu_{\alpha_{\varsigma}}(g) 
		\stackrel{\ref{Lem:int_id}}{=} \nu(f)
	\end{align*}
	as $n \ra \infty$, no matter which values $X_{\varsigma}^{(1)}, \dots, X_{\varsigma}^{(p)}$ the remainder chains were initialized with. The theorem's claim follows from this by the facts that
	\begin{equation*}
		\frac{1}{p} \sum_{i=0}^{\varsigma-1} \sum_{j=1}^p f(X_i^{(j)}) < \infty
		\qquad \text{a.s.} ,
	\end{equation*}
	so that
	\begin{equation*}
		\frac{1}{n \, p} \sum_{i=0}^{\varsigma-1} \sum_{j=1}^p f(X_i^{(j)}) \asconv 0
		\qquad \text{as $n \ra \infty$} ,
	\end{equation*}
	as well as
	\begin{equation*}
		\lim_{n \ra \infty} \frac{n - \varsigma}{n} = 1 .
	\end{equation*}
	That is, we get
	\begin{align*}
		&\frac{1}{n \, p} \sum_{i=0}^{n-1} \sum_{j=1}^p f(X_i^{(j)})
		= \underbrace{\frac{1}{n \, p} \sum_{i=0}^{\varsigma-1} \sum_{j=1}^p f(X_i^{(j)})}_{\ra \; 0 \text{ a.s.}} \\
		&+ \underbrace{\frac{n - \varsigma}{n}}_{\ra \; 1} \cdot \underbrace{\frac{1}{(n - \varsigma) \, p} \sum_{i=\varsigma}^{n-1} \sum_{j=1}^p f(X_i^{(j)})}_{\ra \; \nu(f) \text{ a.s.}}
		\;\asconv\; \nu(f) . 
	\end{align*}
\end{proof}

\section{Experimental Details} \label{App:exp_details}

This section serves to provide detailed explanations of the experiments whose results we presented in Section \ref{Sec:experiments}.

\subsection{Methods} \label{SubApp:exp_methods}

We begin by giving a high-level overview of the inner workings of GPSS (since we used it as a base sampler for PATT without properly explaining it) and the competitor methods used in the experiments.

When given a target density ${\varrho: \R^d \ra \coint{0}{\infty}}$, GPSS \cite{GPSS} first transforms it into ${\varrho^{(1)}(x) := \norm{x}^{d-1} \varrho(x)}$ (where $\norm{\cdot}$ denotes the Euclidean norm). Then, in each iteration, it transitions from the current state $x_{n-1} \in \R^d$ to a new state $x_n \in \R^d$ by performing three main steps. Firstly, it draws a threshold $t_n$ uniformly from the interval $\ooint{0}{\varrho^{(1)}(x_{n-1})}$. Secondly, it writes the current state in polar coordinates as $x_{n-1} = r_{n-1} \theta_{n-1}$, where $r_{n-1} \in \ooint{0}{\infty}$, $\theta_{n-1} \in \sph^{d-1}$ (the Euclidean $(d-1)$-sphere), and updates the direction component $\theta_{n-1}$ while leaving the radius component $r_{n-1}$ fixed. Concretely, the new direction $\theta_n \in \sph^{d-1}$ is determined by randomly choosing a geodesic of $\sph^{d-1}$ (i.e.~a great circle) that runs through the current direction $\theta_{n-1}$ and then performing a \textit{shrinkage procedure} \cite{SSNeal} on it until drawing a proposal $\theta_{\nprop}$ that satisfies ${\varrho^{(1)}(r_{n-1} \theta_{\nprop}) > t_n}$, upon which the method sets $\theta_n := \theta_{\nprop}$. Note that this corresponds to an iteration of \textit{geodesic slice sampling on the sphere} \cite{SphericalSS} targeting the uniform distribution over the directions $\theta_{\nprop}$ satisfying the aforementioned inequality. Once $\theta_n$ has been determined, GPSS similarly updates the radius component $r_{n-1}$ of the intermediate state $\tilde{x}_n := r_{n-1} \theta_n$, while now leaving the direction component fixed. This is done by performing first a \textit{stepping-out procedure} (also \citet{SSNeal}) and then another shrinkage procedure on the set of possible radii $\ooint{0}{\infty}$, until settling for a proposal $r_{\nprop}$ that satisifes $\varrho^{(1)}(r_{\nprop} \theta_n) > t_n$. The product of updated radius $r_n := r_{\nprop}$ and updated direction $\theta_n$, i.e.~$x_n := r_n \theta_n \in \R^d$, is then used as the new state of the method.

Hit-and-run uniform slice sampling (HRUSS) (\citet{MacKayBook}, Section 29.7) works similarly to the radius update of GPSS. To transition from the current state $x_{n-1} \in \R^d$ to a new state $x_n \in \R^d$, it first draws a threshold $t_n$, now from the uniform distribution on the interval $\ooint{0}{\varrho(x_{n-1})}$. It then chooses a uniformly random direction $v_n \in \sph^{d-1}$ and updates the state by a stepping-out and then a shrinkage procedure on $\R$ (around $0 \in \R$) until drawing a proposal $\gamma_{\nprop} \in \R$ for which the state ${x_{\nprop} := x_{n-1} + \gamma_{\nprop} v_n}$ satisfies ${\varrho(x_{\nprop}) > t_n}$, upon which said proposal is accepted as the new state, $x_n := x_{\nprop}$. Note that the possible proposals $x_{\nprop}$ constitute a straight line through the old state $x_{n-1}$ in direction $v_n$.

For the adaptive random walk Metropolis algorithm (AdaRWM) \cite{HaarioAdaRWM}, there are several different formulations that differ slightly in the proposal distribution they use. We relied on the formulation of \citet{ExAda}, where the $n$-th iteration (for large enough $n$), again transitioning from the current state $x_{n-1} \in \R^d$ to a new state $x_n \in \R^d$, is performed as follows. A proposed new state $x_{\nprop}$ is drawn from the scale mixture
\begin{equation*}
	\beta \Nc_d(x_{n-1}, \tfrac{(0.1)^2}{d} I_d) + (1 - \beta) \Nc_d(x_{n-1}, \tfrac{(2.38)^2}{d} \Sigma_n) ,
\end{equation*}
with $\Sigma_n$ denoting the sample covariance of the previously generated samples $x_0,\dots,x_{n-1}$. As suggested by \citet{ExAda}, we set $\beta := 0.05$. With probability $\min\{1, \varrho(x_{\nprop}) / \varrho(x_{n-1})\}$, the proposal is accepted and used as the new state, $x_n := x_{\nprop}$. Otherwise, the proposal is rejected and the current state also becomes the new state, i.e.~$x_n := x_{n-1}$.

Unlike the three methods we described so far, \textit{two-group generalized elliptical slice sampling} (GESS) \cite{GenEllSS} is an ensemble method, meaning that it maintains a large number of parallel chains on the sample space and that these chains interact in certain ways. As the name suggests, the chains are organized into two distinct groups. Generally speaking, the method proceeds by alternatingly updating the groups, where the states of the chains in one group are always updated according to a transition mechanism (the same one for all chains in that group) that is determined based on the current states of the chains in the other group. The key ideas behind this transition mechanism are firstly to approximate the target distribution by a multivariate $t$-distribution in a maximum likelihood fashion (using the passive group's states to discretely represent the target) and secondly to utilize this approximation in the transition by making use of elliptical slice sampling (ESS) \cite{EllipticalSS} mechanics. \citet{GenEllSS} first considered updating the groups one iteration at a time. However, they also pointed out that it is more efficient to run each group for hundreds of iterations at a time before making it interchange roles with the other one. In our experiments, we always used the formula
\begin{equation*}
	n_{\text{ibu}} := \lfloor \max\{d, 25\} \cdot p / 10 \rfloor ,
\end{equation*}
where $d$ is the sample space dimension and $p$ the number of chains used by GESS in each of its groups, to determine the number of \textit{iterations between updates} of the approximating multivariate $t$-distribution (and with that the role-swap of the two groups). The maximum likelihood algorithm suggested by \citet{GenEllSS} to be used within the parameter update is only guaranteed to work if the number of points it receives as input is at least twice as large as their dimension. Accordingly, we always used $p := 2 d$ chains per group for GESS in our experiments. Below we elaborate on how we adjusted its total number of iterations to account for this.

Finally, the \textit{No-U-Turn sampler} (NUTS) \cite{NUTS} is another non-adaptive MCMC method (like HRUSS), which was developed as a tuning-free alternative to \textit{Hamiltonian Monte Carlo} (HMC) \cite{BayesianNeuralNets,Neal_HMC}. Specifically, NUTS performs its transitions based on a simulation of Hamiltonian dynamics, which often allows it to avoid random-walk-like behavior and achieve a remarkably high sample quality, but at the cost of requiring access to the target density's gradient (unlike any other sampler we consider here) and frequently using a large number of evaluations of this gradient to perform a single iteration. For simplicity, we counted NUTS' \textit{target gradient evaluations} (TGE) the same as TDE in our analysis of the samplers' computational cost\footnote{To be even more precise, we only counted NUTS' number of leapfrog steps (which with efficient implementation require one TGE each), thereby disregarding the additionally required evaluations of the target density itself.}, as was previously done by \cite{GenEllSS}. Since the fully self-tuning version of NUTS is fairly complex, we were not intent on implementing it ourselves. Instead we relied on the Python interface \textit{PyStan}\footnote{\url{https://pystan.readthedocs.io/en/latest/}} of the software package \textit{Stan}. Like we did with HRUSS and AdaRWM, we ran Stan's NUTS in a naively parallelized fashion, meaning that we ran several of its chains in parallel, but did not let them interact in any way.

It should be noted that Stan's NUTS precedes any requested sampling by a short tuning period, in which it adaptively fits certain parameters to the given target. Within this process, the sampler also learns a diagonal \textit{mass matrix}, which serves as the covariance of a Gaussian used to set up the proposal structure. We believe the learning of this diagonal mass matrix could be equivalent to adaptive ATT with variance adjustments, but we refrained from conducting an in-depth investigation of this (possible) connection.

\subsection{General Principles}

Here we state some general principles that we relied on in choosing the parameters of the samplers across the different experiments.

For PATT and the naively parallelized versions of HRUSS, AdaRWM and NUTS that we used to run these three methods, we could freely choose the number $p$ of parallel chains maintained by each method. In order to ensure that our experiments could be executed unaltered on a regular workstation (for the sake of good reproducibility), we ran them on such a machine ourselves. This led us to choose $p := 10$ (slightly less than the number of available processor cores on our machine) for each of the aforementioned methods throughout all of the experiments.

Generally speaking, in each experiment, we aimed to provide the PATT samplers and each of their competitors with roughly the same amount of computational resources, specifically by controlling the total number of samples each sampler was permitted to generate. Moreover, we decided to err on the side of caution by giving the competitors of PATT greater advantages than would perhaps been appropriate in regards to this aim.

More concretely, in each experiment we set a parameter $n_{\text{its}} \in \N$. As HRUSS tends to use a similar amount of TDE/it as GPSS, we ran both PATT-GPSS and HRUSS, and for simplicity also PATT-ESS, for precisely $n_{\text{its}}$ (multi-chain) iterations each, thereby generating a total amount of $p \cdot n_{\text{its}}$ samples per method.
For AdaRWM, we decided to mostly ignore that its computational overhead becomes quite substantial in high dimensions (since it must compute the Cholesky decomposition of a new $d \times d$ matrix in each iteration of each chain in order to draw the proposal). Instead, we chose to account for the fact that it only requires a single TDE/it by running it for more iterations overall. Specifically, we ran it for ${n_{\text{its\_rwm}} := 5 \cdot n_{\text{its}}}$ (multi-chain) iterations in each experiment.

For GESS, its desire to run $p_{\text{GESS}} := 4 d$ chains in total meant that strictly enforcing our goal of a parity of computational resources across methods would often require us to restrict it to very few iterations per chain. However, partly to give GESS some opportunity to show its long-term performance, and partly in order to enable us to properly measure its performance in the first place\footnote{The way we compute IATs, while producing fairly stable results, does require a relatively large number of samples per chain in order to work properly, particularly if the chain is strongly autocorrelated.}, we ignored this and always ran it for $n_{\text{its\_gess}} := 2 \cdot 10^4$ iterations per chain (except in the experiment on BLR for German credit data, where we determined $n_{\text{its\_gess}} := 10^4$ would suffice).

Since the computational cost per iteration of Stan's NUTS varies substantially from setting to setting but tends to be at least as high as those of the PATT samplers, we ran it for $n_{\text{its}}$ iterations whenever feasible, but scaled this number down to $n_{\text{its\_stan}} := n_{\text{its}} / 10$ in those experiments where it was particularly costly to run compared with the other methods (namely the one on Bayesian inference with multivariate exponential distributions and the one on BLR-FE for wine quality data) and to $n_{\text{its\_stan}} := n_{\text{its}} / 2$ in the experiment where the target was the costliest to evaluate (namely the one on Bayesian hyperparameter inference).

To initialize the different samplers, we always drew independent samples from an initial distribution, which for convenience we always chose to be some Gaussian, usually (unless specified otherwise) the standard Gaussian $\Nc_d(\fatzero, I_d)$. We used the exact same initial states for PATT-ESS, PATT-GPSS, HRUSS, AdaRWM and Stan's NUTS, and generated separate ones from the same distribution for GESS to accomodate its larger number of chains.

As the initialization was usually uninformed, for our PATT samplers it made sense to use an initialization burn-in period (as proposed in Section \ref{SubSec:init_burn_in}). Rather than fine-tuning the length of this period for every experiment, we just always set it to $n_{\text{burn}} := n_{\text{its}} / 10$, which turned out to work very well and may be a good rule of thumb more generally (unless one is looking to draw very large numbers of samples, in which case a smaller fraction is likely advisable). Of course this burn-in period had to be subtracted from the samplers' resource budgets, so we only ran each PATT sampler for $0.9 \cdot n_{\text{its}}$ (multi-chain) iterations after the burn-in.

\subsection{Runtime Measurements} \label{SubApp:runtimes}

To obtain precise representations of the various samplers' per-iteration runtimes in the different experiments, we relied on granular measurements whenever possible. That is, for the two PATT samplers as well as HRUSS and AdaRWM, we explicitly computed the duration of each iteration as the physical time passed between beginning and end of the iteration (as measured with python's \texttt{time.time()} function), thereby relying on the fact that these samplers ran each of their chains on an otherwise idle processor core. In case of the PATT samplers, we took care to also record the time spent waiting for the parallel chains' synchronisation at each update time. We could then compute the total runtime of each of the aforementioned samplers across the analysis-relevant range of iterations (i.e.~the latter half) by summing over the corresponding per-iteration times for each chain and taking the arithmetic mean over the different chains. The numbers of samples per second of each of these samplers were then computed by dividing their amounts of analysis-relevant samples by these total runtimes.

For the remaining two samplers, GESS and Stan's NUTS, this method of measuring runtimes was deemed infeasible (for GESS because it ran more chains than we had idle processor cores and for NUTS because we did not want to modify its implementation), and so the runtimes of these two samplers had to be determined somewhat less accurately. Specifically, we measured the amount of time passed between beginning and end of each of these methods' overall execution in each experiment and scaled this time down from the total number of iterations to the number of analysis-relevant ones. For GESS, we made this into a fair comparison with the other methods by actively restricting it to using only $n_{\text{thr}} := 10$ parallel threads (corresponding to the number $p$ of parallel chains used by all other samplers).

We now proceed to explain the models and data underlying our experiments, explicitly state the target densities and provide some results omitted in Section \ref{Sec:experiments}.

\subsection{Bayesian Inference with Multivariate Exponential Distributions}

In any fixed dimension $d \in \N$, we define the \textit{multivariate exponential distributions} as a sub-class $(\Exp_d(\tau,\Pi))_{\tau,\Pi}$ of the $d$-variate symmetric generalized multivariate normal distributions by specifying that, for any $\tau \in \R^d$ and any positive definite $\Pi \in \R^{d \times d}$, the distribution $\Exp_d(\tau,\Pi)$ has (unnormalized) density
\begin{equation*}
	\Exp_d(x; \tau,\Pi)
	:= \exp(-\big( (x - \tau)^{\top} \Pi^{-1} (x - \tau) \big)^{1/2}) .
\end{equation*}
Note that for any $\tau \in \R^d$ and $\sigma > 0$ one has
\begin{equation*}
	\Exp_d(x; \tau, \sigma^2 I_d)
	= \exp(-\norm{x - \tau} / \sigma)
\end{equation*}
and that, in one dimension, $\Exp_1(0,\sigma^2)$ is a symmetrized version of $\Exp(\sigma^{-1})$, the usual exponential distribution with rate $\lambda := \sigma^{-1}$.

For our experiment, we constructed a Bayesian inference setting from these distributions as follows. Suppose there is a $d$-variate random vector of interest $X$. We imposed the standard multivariate exponential prior ${X \sim \Exp_d(\fatzero,I_d)}$. Furthermore, we presumed to have access to a spectrum of different ways of measuring $X$ up to additive noise, and that these ways differ in the type of noise incurred, with there being advantages and disadvantages to each end: On one end of the spectrum, the different components of the noise vector are very strongly correlated but their variance is large. On the other end they have small variance but are far closer to independence. Specifically, the different measurement options were modeled by a sequence of random vectors $(Z_m)_{m \in \mathbb{N}}$ given by
\begin{equation*}
	Z_m = X + \varepsilon_m , 
	\quad \varepsilon_m \sim \Exp_d(\mathbf{0},\Sigma^{(m)}) ,
\end{equation*}
where
\begin{equation*}
	\Sigma^{(m)}_{i,j} 
	= \begin{cases} 
		(m+1)^2 / d & i = j , \\ 
		m(m+1) / d & i \neq j . 
	\end{cases}
\end{equation*}
To obtain a concrete likelihood, we ``secretly'' sampled a ``true value'' for $X$ from $\Nc_d(\fatzero, d \cdot I_d) \approx \Exp_d(\fatzero,I_d)$ and used it to generate independent synthetic data $z_1,\dots,z_{100} \in \R^d$, representing one realization each of the vectors $Z_1,\dots,Z_{100}$ (although for simplicity we only generated the data from an approximation of the above model). To summarize, our target density in this experiment was the (unnormalized) posterior density
\begin{equation*}
	\varrho(x) 
	:= \Exp_d(x; \fatzero,I_d) \prod_{m=1}^{100} \Exp_d(z_m; x,\Sigma^{(m)}) ,
	\quad x \in \R^d .
\end{equation*}

We set the sample space dimension to $d = 50$ and the base number of (multi-chain) iterations to $n_{\text{its}} := 10^5$. As the target turned out to be roughly centered around the sample mean $\overline{z}$ of the data points $z_1,\dots,z_{100}$, which is typically quite far from the origin, we drew the samplers' initial states from $\Nc_d(\overline{z}, I_d)$.

Recall that the sampling statistics for this experiment are presented in Table \ref{Tab:multiv_exp_dists}, and that histograms of the steps sizes and some marginal trace plots are provided in Appendix \ref{App:plots}, Figure \ref{Fig:multiv_exp_dists}.

\subsection{Bayesian Logistic Regression} \label{SubApp:BLR}

In our experiments on Bayesian logistic regression (BLR), we aimed to reproduce and expand upon experiments conducted by \citet{GenEllSS} when proposing GESS. Our BLR experiments -- as well as those of \citet{GenEllSS} -- all share the same probabilistic model, under which the target distribution is the posterior resulting from the combination of a logistic regression likelihood and the Gaussian prior $\Nc_d(\fatzero, 10^2 I_d)$. The only differences between the experiments' settings are therefore due to the data plugged into this model.

For the sake of completeness, we give an explicit formula for the target density corresponding to this posterior. Let us denote an abstract data set for binary classification with $n_{\text{data}} \in \N$ data points and $d \in \N$ features per data point by $(a,b) = (a^{(i)},b^{(i)})_{i=1,\dots,n_{\text{data}}} \subset \R^d \times \{-1,1\}$. Then the (unnormalized) posterior density of BLR with prior $\Nc_d(\fatzero, 10^2 I_d)$ for the data $(a,b)$ is given by
\begin{equation*}
	\varrho(x)
	= \exp(-\tfrac{1}{200} \norm{x}^2) \prod_{i=1}^{n_{\text{data}}} \frac{1}{1 + \exp(-b^{(i)} \langle a^{(i)}, x \rangle)}
\end{equation*}
for $x \in \R^d$, where $\langle \cdot, \cdot \rangle$ denotes the standard inner product on $\R^d$.

As is common practice for BLR, we always normalized the given data. For purely numerical data, this simply meant shifting and scaling the feature vectors so that each feature has sample mean zero and sample variance one. However, in one of the data sets we used (the German credit data) many of the features were either binary or one-hot-encoded. For such features it does not actually make sense to normalize them in this fashion, so we took care to leave them out of the normalization procedure and instead ensured that each of them uses the values $0$ and $1$ (as some were using $1$ and $2$ instead). In the experiments with feature engineering (FE), we normalized the data before FE, and refrained from re-normalizing it afterwards. Moreover, in order to enable the regression hyperplane to have a non-zero intercept, we always appended the given data (after FE, if applicable) by a constant feature. As the key quantity for us is the dimension of the sample space, we implicitly include this constant feature in our statements about the dimensionality of the different data sets below.

Let us now talk in some more detail about the data sets we used and provide the experimental results omitted earlier.
In our first BLR experiment, following \citet{GenEllSS}, we used the \textit{German credit data set} \cite{GermanCredit}, which consists of $n_{\text{data}} = 1000$ data points and $d = 25$ features. The data's features represent various attributes of people that applied for credits at a German credit institution, with the binary labels signifying whether or not that institution ultimately deemed them creditworthy.

We set $n_{\text{its}} := 2 \cdot 10^4$. Here we ran GESS for ${n_{\text{its\_gess}} := 10^4}$ iterations per chain. The results are presented in Table \ref{Tab:BLR_German_credit}. Both PATT samplers achieved near-perfect sample quality in terms of their mean IATs ($1.00$ being the smallest possible value), at a computational cost that was either reasonably (PATT-GPSS) or extremely low (PATT-ESS). Although NUTS maintained a similar sample quality, its iterations were so much more costly that both PATT samplers outperformed it handily according to the summary metrics TDE/ES and ES/s. All other methods failed to achieve a comparable sample quality and consequently their performances lagged far behind those of PATT and NUTS.
We note again that the final covariance/scale matrices used by the adaptive samplers for this experiment are shown in Appendix \ref{App:plots}, Figure \ref{Fig:BLR_German_credit_covs}.

\begin{table*}[t]
	\caption{Sampling statistics for the experiment on BLR for German credit data.}
	\label{Tab:BLR_German_credit}
	\vskip 0.1in
	\begin{center}
		\begin{small}
			\begin{sc}
				\begin{tabular}{lrrrrrr}
					\toprule
					Sampler 	& TDE/it	& Samples/s	& Mean IAT	& MSS  & TDE/ES		& ES/s \\
					\midrule
					PATT-ESS	&  1.27		& 38005		&   1.42	& 1.77 &    1.80	& 26792.51 \\
					PATT-GPSS	&  6.11		& 18332		&   1.24	& 1.81 &    7.57	& 14801.03 \\
					HRUSS 		&  5.21		& 36850		& 425.88	& 0.11 & 2220.89	&    86.53 \\
					AdaRWM		&  1.00		& 90467		& 191.07	& 0.10 &  191.07	&   473.48 \\
					GESS		&  4.74		&  8623		&  83.35	& 0.65 &  394.74	&   103.45 \\
					Stan's NUTS	& 39.21		&  4744		&   1.29	& 1.93 &   50.73	&  3667.39 \\
					\bottomrule
				\end{tabular}
			\end{sc}
		\end{small}
	\end{center}
	\vskip -0.1in
\end{table*}

For our second BLR experiment, again following \citet{GenEllSS}, we used the \textit{breast cancer Wisconsin (diagnostic) data set} \cite{BreastCancer}, which consists of $n_{\text{data}} = 569$ data points and $d = 31$ features. The features were extracted from image data created in the context of breast cancer diagnostics, the binary labels denote whether the specimen behind the image was ultimately diagnosed as malignant or benign.

We set ${n_{\text{its}} := 10^5}$. The results are presented in Table \ref{Tab:BLR_breast_cancer}. Clearly, the PATT samplers had more difficulties dealing with this target than with the previous one, and neither of them managed to match the near-perfect sample quality put forth by NUTS. On the other hand, their respective TDE/it remained reasonably low, whereas NUTS required vastly more TDE/it than even in the previous experiment. Altogether, both PATT samplers managed to beat NUTS according to their TDE/ES (albeit by a smaller margin than in the previous experiment), and they even came relatively close to NUTS in the ES/s category\footnote{Let us emphasize here that PATT should not be expected to beat NUTS in terms of ES/s because we implemented the former in the comparatively slow language Python, whereas the latter essentially runs entirely in the much faster language C++ (as a result of how Stan and its Python interface PyStan are set up).}. Notably, AdaRWM managed to beat NUTS (but not PATT) in TDE/ES here, thanks to its consistent use of just one TDE/it in combination with its relatively reasonable mean IAT.
We note again that the final covariance/scale matrices used by the adaptive samplers for this experiment are shown in Appendix \ref{App:plots}, Figure \ref{Fig:BLR_breast_cancer_covs}.

\begin{table*}[t]
	\caption{Sampling statistics for the experiment on BLR for breast cancer data.}
	\label{Tab:BLR_breast_cancer}
	\vskip 0.1in
	\begin{center}
		\begin{small}
			\begin{sc}
				\begin{tabular}{lrrrrrr}
					\toprule
					Sampler 	& TDE/it	& Samples/s	& Mean IAT	& MSS	& TDE/ES	& ES/s \\
					\midrule
					PATT-ESS	&   3.45	& 18665		&   12.56	& 14.93 &    43.38	& 1485.76 \\
					PATT-GPSS	&   8.28	& 15365		&    8.66	& 15.94 &    71.74	& 1773.74 \\
					HRUSS 		&   5.69	& 40322		& 2968.41	&  0.73 & 16893.87	&   13.58 \\
					AdaRWM		&   1.00	& 87359		&  150.95	&  1.86 &   150.95	&  578.72 \\
					GESS		&   5.25	& 12431		&  109.07	&  5.01 &   572.34	&  113.97 \\
					Stan's NUTS	& 131.85	&  2829		&    1.33	& 34.97 &   175.27	& 2128.92 \\
					\bottomrule
				\end{tabular}
			\end{sc}
		\end{small}
	\end{center}
	\vskip -0.1in
\end{table*}

In our third BLR experiment, we used the \textit{Pima diabetes data} \cite{PimaDiabetes}, which has ${n_{\text{data}} = 768}$ data points and ${d_{\text{data}} = 8}$ features. Its features represent diagnostic measurements, with each data point corresponding to a patient, and the binary labels denote whether or not the patient was found to have type II diabetes.

As BLR for this data set is known not to lead to a particularly good classification performance (a posterior mean classifier achieves a classification accuracy of around $78\%$), it is of interest whether this issue can be alleviated through the use of non-linear classifiers. One way to achieve this is to transform the data in non-linear ways before plugging it into the BLR model (this is what we mean by \textit{feature engineering}, FE). Following \citet{NUTS}, where this exact approach was applied to the German credit data, we augment the data by all two-way interactions between features. That is, for any indices $1 \leq i \leq j \leq d_{\text{data}}$, we append the entry-wise product of the vectors representing the $i$-th and $j$-th feature to the data. With the usual constant feature for an intercept added in the end, we thus obtain augmented data with $d = 45$ features. In the following we refer to BLR with FE in short as BLR-FE.

We set $n_{\text{its}} := 5 \cdot 10^4$. The results are presented in Table \ref{Tab:BLR-FE_Pima_diabetes}. Here the PATT samplers achieved a very good mean IAT using fairly few TDE/it. NUTS yet again produced samples with nearly perfect mean IAT, but used comparatively many TDE/it to do so, thereby allowing PATT-ESS to significantly outperform it in the TDE/ES category. Notably, PATT-ESS also beat NUTS according to their ES/s, albeit by a smaller margin. None of the other methods managed to compete on equal footing with PATT and NUTS in these summary metrics.
We note again that the final covariance/scale matrices used by the adaptive samplers for this experiment are shown in Appendix \ref{App:plots}, Figure \ref{Fig:BLR-FE_Pima_diabetes_covs}. In this case, said figure appears to show that both AdaRWM and GESS did not yet succeed at fully grasping the target's covariance structure.

\begin{table*}[t]
	\caption{Sampling statistics for the experiment on BLR-FE for Pima diabetes data.}
	\label{Tab:BLR-FE_Pima_diabetes}
	\vskip 0.1in
	\begin{center}
		\begin{small}
			\begin{sc}
				\begin{tabular}{lrrrrrr}
					\toprule
					Sampler 	& TDE/it	& Samples/s	& Mean IAT	& MSS	& TDE/ES	& ES/s \\
					\midrule
					PATT-ESS	&  2.01		& 24200		&    2.93	& 1.05 &     5.88	& 8259.59 \\
					PATT-GPSS	&  7.70		& 13803		&    2.48	& 1.08 &    19.11	& 5564.14 \\
					HRUSS 		&  5.43		& 28600		&  613.58	& 0.09 &  3331.42	&   46.61 \\
					AdaRWM		&  1.00		& 66163		&  420.16	& 0.05 &   420.16	&  157.47 \\
					GESS		&  8.85		&  5754		& 1374.96	& 0.07 & 12170.44	&    4.19 \\
					Stan's NUTS	& 29.53		&  5552		&    1.01	& 1.57 &    29.80	& 5502.30 \\
					\bottomrule
				\end{tabular}
			\end{sc}
		\end{small}
	\end{center}
	\vskip -0.1in
\end{table*}

In our fourth and final experiment on BLR, we used the \textit{red wine quality data set} \cite{WineQuality}, which has $n_{\text{data}} = 1599$ data points and $d_{\text{data}} = 11$ features. Here the features represent measurements of various physical and chemical properties of red wine samples and the categorical labels, which are integer-valued and range from $0$ to $10$, represent a quality scale on which the wine samples were placed by human testers. In order to obtain binary labels, we transformed the classification problem from the prediction of the precise quality level into merely predicting whether or not a given wine sample is of quality $6$ or higher (which leads to a relatively balanced data set).
By the same rationale as in the third BLR experiment, we use FE for the wine data, again appending the data by the two-way interactions between its features. In this case, the augmented data possessed $d = 78$ features. Due to this relatively high dimension, we deemed it appropriate to exclude GESS from the experiment (as it would require very large amounts of runtime and memory to run here, which would impair the experiment's reproducibility).

We set $n_{\text{its}} := 10^5$. The results are shown in Table \ref{Tab:BLR-FE_wine_quality}. Like in the previous experiment, the PATT samplers' mean IATs were quite good, but worse than that of NUTS. However, NUTS again required an enormous amount of TDE/it to achieve this, which once more allowed both PATT samplers to beat it by a substantial margin in terms of TDE/ES. Interestingly, the advantage of PATT over NUTS according to their ES/s was an entire order of magnitude here, which is a significantly larger margin than in all prior BLR experiments (where, in one case, PATT did not even beat NUTS in this category). As usual, the other competitors' performances in the summary metrics TDE/ES and ES/s did not come close to those of PATT and NUTS.
We note again that the final covariance matrices used by PATT-GPSS and AdaRWM for this experiment are shown in Appendix \ref{App:plots}, Figure \ref{Fig:BLR-FE_wine_quality_covs}. In this case they suggest that AdaRWM was still nowhere near a good approximation of the target's true covariance, even after its $n_{\text{its\_rwm}} = 5 \cdot 10^5$ iterations per chain, which explains its unusually bad performance.

\begin{table*}[t]
	\caption{Sampling statistics for the experiment on BLR-FE for wine quality data.}
	\label{Tab:BLR-FE_wine_quality}
	\vskip 0.1in
	\begin{center}
		\begin{small}
			\begin{sc}
				\begin{tabular}{lrrrrrr}
					\toprule
					Sampler 	& TDE/it	& Samples/s	& Mean IAT	& MSS	& TDE/ES	& ES/s \\
					\midrule
					PATT-ESS	&   2.53	& 10099		&    4.97	& 1.69 &    12.58	& 2031.61 \\
					PATT-GPSS	&   7.35	&  7027		&    3.84	& 1.80 &    28.27	& 1828.44 \\
					HRUSS 		&   6.06	& 11408		& 4278.22	& 0.05 & 25926.77	&    2.67 \\
					AdaRWM		&   1.00	& 24382		& 3207.84	& 0.03 &  3207.84	&    7.60 \\
					Stan's NUTS	& 146.89	&   287		&    1.55	& 2.65 &   227.45	&  185.86 \\
					\bottomrule
				\end{tabular}
			\end{sc}
		\end{small}
	\end{center}
	\vskip -0.1in
\end{table*}

We note that in both BLR-FE experiments, the FE turned out to improve the samplers' classification accuracy, but only very marginally (by about 1-2\%).

Moreover, we note that throughout the four BLR experiments, PATT-ESS always required significantly fewer (often several times less) TDE/ES than PATT-GPSS. We conjecture this to be a result of the BLR model's Gaussian prior, as it gives each specific target distribution tails that are no heavier than Gaussian, and ESS is known to work particularly well in this light-tailed regime. Accordingly, we would not expect PATT-ESS to be able to keep up with PATT-GPSS in any setting where the target's tails are substantially heavier than Gaussian. Some evidence of this is given by the experiment on Bayesian inference with multivariate exponential distributions, which used a target with heavier than Gaussian tails and showed PATT-GPSS outperforming PATT-ESS (recall Table \ref{Tab:multiv_exp_dists}). However, the difference between the two methods' performances in that experiment is not all that large. For targets with heavy polynomial tails (equivalent to those of multivariate $t$-distributions with few degrees of freedom), we would expect PATT-GPSS to have a much more pronounced advantage over PATT-ESS (in fact, in such cases the former should require several orders of magnitude fewer TDE/ES than the latter), based on experiments with GPSS and ESS conducted by \citet{GPSS}.

\subsection{Bayesian Hyperparameter Inference for Gaussian Process Regression of US Census Data} \label{SubApp:hyperparam_inf}

The setting of our final experiment is a bit more involved than those in the previous ones. We again took great inspiration from an experiment of \citet{GenEllSS}. Specifically, we performed Bayesian inference on the hyperparameters of a Gaussian process regression model for some real-world data. Let us elaborate on this for a bit.

Denote an abstract collection of regression data as
\begin{equation*}
	(a,b) = (a^{(i)},b^{(i)})_{i=1,\dots,n_{\text{data}}} ,
\end{equation*}
where $a^{(i)} \in \R^d$ and $b^{(i)} \in \R$. Suppose we want to model this data by \textit{Gaussian process} (GP) \textit{regression}. That is, we model it by a posterior distribution over an infinite-dimensional space of latent functions $f: \R^d \ra \R$. The likelihood is obtained from
\begin{equation*}
	b^{(i)} = f(a^{(i)}) + \epsi_i ,
	\qquad \epsi_i \sim \Nc(0,\sigma^2) ,
\end{equation*}
with the noise variables $\epsi_i$ being independent and $\sigma^2 \geq 0$ a fixed hyperparameter (we used $\sigma^2 = 10^{-2}$).
The prior is a GP prior with mean zero and covariance determined by an \textit{anisotropic radial basis function} (RBF) \textit{kernel}, i.e.~a function $k^{(\gamma)}: \R^d \ra \R^d$ of the form
\begin{equation*}
	k^{(\gamma)}(a, a^{\prime}) 
	= \exp( -\frac{1}{2} \sum_{j=1}^d \frac{(a[j] - a^{\prime}[j])^2}{\gamma[j]^2} )
\end{equation*}
for $a, a^{\prime} \in \R^d$, which has $d$ length scale parameters
\begin{equation*}
	\gamma = (\gamma[1],\dots,\gamma[d]) \in \ooint{0}{\infty}^d
\end{equation*}
that are hyperparameters to the GP regression model.

To obtain a good regression model, we want to infer these length scale hyperparameters through the Bayesian framework, which is commonly done by placing a suitable prior on them and using the marginal likelihood of the data in the above GP model as the likelihood for a given $\gamma$ (see \citet{GP4ML}, Chapter 5). Recall that the \textit{kernel matrix} for the data $(a,b)$ corresponding to the kernel $k^{(\gamma)}$ is the matrix $K^{(\gamma)} \in \R^{n_{\text{data}} \times n_{\text{data}}}$ with entries
\begin{equation*}
	K^{(\gamma)}_{i,i^{\prime}} = k^{(\gamma)}(a^{(i)}, a^{(i^{\prime})}) ,
	\quad i,i^{\prime} \in \{1,\dots,n_{\text{data}}\} ,
\end{equation*}
and note that it can be used to concisely write the marginal likelihood of the data for a given kernel $k^{(\gamma)}$ as
\begin{equation*}
	p(b \mid a, \gamma)
	= \Nc_{n_{\text{data}}}(b; \fatzero, K^{(\gamma)} + \sigma^2 I_{n_{\text{data}}}) ,
	\quad \gamma \in \ooint{0}{\infty}^d
\end{equation*}
(see \citet{GP4ML}, equation 5.8).

Following \citet{GenEllSS}, we imposed the independent exponential prior
\begin{equation*}
	p(\gamma) 
	= \prod_{j=1}^d \text{Exp}(\gamma[j]; r)
	\propto \exp(-r \sum_{j=1}^d \gamma[j])
\end{equation*}
for $\gamma \in \ooint{0}{\infty}^d$, with fixed rate $r = 0.1$. 

Despite its relatively simple algebraic form, the posterior distribution of the parameters $\gamma$ resulting from the above model turns out to be very challenging to sample from (at least for all MCMC-based samplers we are aware of), since its level sets are far from ellipsoidal, instead resembling simplices with rounded corners. We illustrate this issue in Figure \ref{Fig:hyperparam_inf_biv_marginal_gamma}. Fortunately, certain default behaviors of the Stan software provide an easy way to substantially simplify the problem: As the posterior of $\gamma$ is only supported on $\ooint{0}{\infty}^d$, Stan would not target it directly, instead resorting to the corresponding distribution of the variables $\beta := \log(\gamma)$ (the $\log$ being applied entry-wise). Although Stan's main motivation for doing this is the fact that the distribution of $\beta$ is supported on all of $\R^d$, which makes it easier for the software's default sampler NUTS to take large steps on it, at the same time the switch brings with it the advantage that the distribution of $\beta$ has a much nicer geometry than that of $\gamma$ (see Figure \ref{Fig:hyperparam_inf_biv_marginal_beta}). Hence we used the distribution of $\beta$ resulting from the posterior of $\gamma$ as our target distribution. Since we did not require the target density to be normalized, we could simply use the product of prior and likelihood from the posterior of $\gamma$ as the basis of our target, plug in $\gamma = \exp(\beta)$ (also applied entry-wise) and additionally multiply by the Jacobian determinant of the inverse transformation $\beta \mapsto \gamma$. From this we ultimately obtained the target density
\begin{align*}
	\varrho(\beta)
	&= \exp(\sum_{j=1}^d \beta[j]) \cdot \exp(- r \sum_{j=1}^d \exp(\beta[j])) \\
	&\quad \cdot \, \Nc_{n_{\text{data}}}(b; \fatzero, K^{(\exp(\beta))} + \sigma^2 I_{n_{\text{data}}})
\end{align*}
for $\beta = (\beta[1], \dots, \beta[d])^{\top} \in \R^d$.

\begin{figure}[t]
	\begin{center}
		\includegraphics[width=0.4\textwidth]{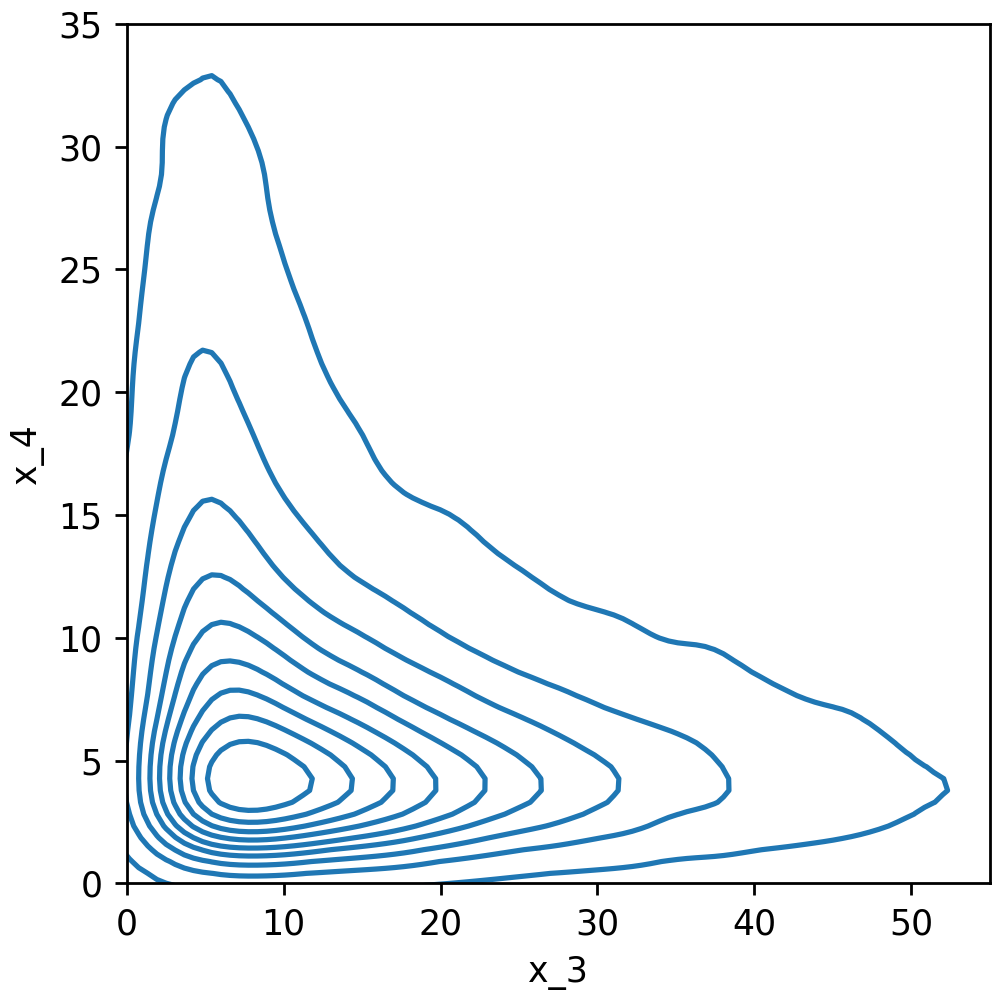}
	\end{center}
	\vskip -0.1in
	\caption{Contour plot corresponding to a kernel density estimate of a bivariate marginal of the posterior density of the length scale parameters $\gamma$ in the setting of the experiment on Bayesian hyperparameter inference for GP regression of US census data. \label{Fig:hyperparam_inf_biv_marginal_gamma}}
\end{figure}

\begin{figure}[t]
	\begin{center}
		\includegraphics[width=0.4\textwidth]{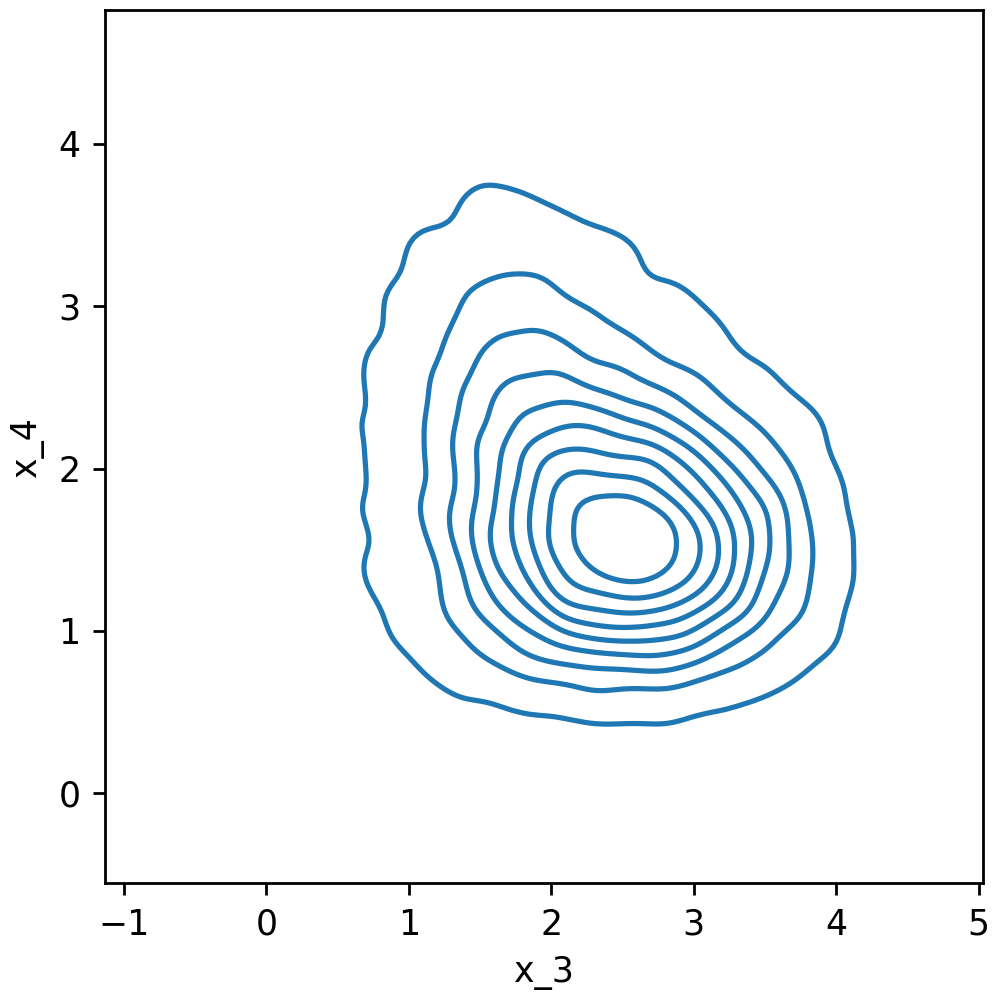}
	\end{center}
	\vskip -0.1in
	\caption{Contour plot corresponding to a kernel density estimate of a bivariate marginal of the posterior density of the log-transformed length scale parameters $\beta$ in the setting of the experiment on Bayesian hyperparameter inference for GP regression of US census data. \label{Fig:hyperparam_inf_biv_marginal_beta}}
\end{figure}

Regrettably, the Bayesian approach employed here to infer the parameters $\gamma$ is computationally infeasible for all but the smallest data sets because every evaluation of the above target density involves computing the kernel matrix $K^{(\exp(\beta))}$, which has complexity $\O(n_{\text{data}}^2 d)$, as well as the inverse and determinant of $K^{(\exp(\beta))} + \sigma^2 I_{n_{\text{data}}}$, each of which has complexity $\O(n_{\text{data}}^3)$.

As data for the model we used a small subset of county-wise accumulations of some recent \textit{US census data}, which we obtained from Kaggle\footnote{\href{https://www.kaggle.com/datasets/muonneutrino/us-census-demographic-data?resource=download&select=acs2015_county_data.csv}{https://www.kaggle.com/datasets/muonneutrino/us-census-demographic-data?resource=download\&select=acs2015\_county \_data.csv}}. The raw data consists of ${n_{\text{raw}} = 3220}$ data points with $d_{\text{raw}} = 37$ features. These features represent statistics of various demographical, social and economical attributes, with each data point summarizing the population of a separate county. We decided to use it as regression data for the median household income of a county, given the other attributes. To this end we removed a few features that we deemed to closely related to this target variable (such as per capita income), as well as several categorical and identifier features, after which $d_{\text{data}} = 30$ features (plus the target variable) remained. Due to the aforementioned complexity issues, we randomly selected a small subset of just $n_{\text{data}} = 100$ data points. Finally, we normalized both the features and the target variable to each have sample mean zero and sample variance one.

We should acknowledge here that (according to several preliminary runs) the experiment's overall results (here meaning the relative sizes of the samplers' TDE/ES) appear to vary substantially (by a factor of at least $2$). Moreover, the results presented in Table \ref{Tab:hyperparam_inf} seemed to lie more towards the favorable end of this variation with regards to the performance of PATT. We chose to present PATT in such a favorable light because in our view the experiment should serve to demonstrate that PATT \textit{can} work very well for such targets, not that this is guaranteed to be the case. Unfortunately, at present we cannot provide any insight into which data-dependent characteristics of the target are to blame for the (comparatively) large variation in the results.

We set $n_{\text{its}} := 2 \cdot 10^4$ and generated each initial state by drawing from $\Nc_d(\fatzero, 10^2 I_d)$ and then applying the coordinate-wise absolute value, followed by the coordinate-wise logarithm. In order to end up with an experiment that terminates reasonably quickly (despite the target being very expensive to evaluate), we refrained from running GESS here and additionally restricted NUTS to performing only $n_{\text{its\_stan}} := 10^4$ iterations per chain.

Since we ended up letting our samplers simulate the variables $\beta$, whereas our main interest remained in the variables $\gamma = \exp(\beta)$, we found it more appropriate to analyze the sampling performance based on the corresponding $\gamma$-samples (i.e.~the entry-wise exponentials of the ``raw'' samples produced by simulating $\beta$) so this is what we did to obtain the mean IAT and MSS values (and by extension the TDE/ES and ES/s) presented in Table \ref{Tab:hyperparam_inf}.

For a brief discussion of the experiment's results, we refer back to Section \ref{Sec:experiments}.

\section{Ablation Studies} \label{App:ablation}

Generally speaking, an ablation study can be performed on any compound procedure possessing methodological components that are not strictly necessary for the procedure's well-definedness. In the simplest case, the ablation study is performed for a single such component and essentially consists of comparing two different procedures, the one with all methodological components left in place and the one with the component under consideration removed and all others left in place. In this section we perform ablation studies of various aspects of PATT.

\subsection{Adjustment Types} \label{SubApp:ablation_adj_types}

\begin{table*}[t] 
	\caption{Sampling statistics for the ablation study on adjustment types. The mean IATs were computed w.r.t.~the component-wise absolute values of the shifted samples generated after the initialization burn-in period.}
	\label{Tab:ablation_adj_types}
	\vskip 0.1in
	\begin{center}
		\begin{small}
			\begin{sc}
				\begin{tabular}{lrrrrrr}
					\toprule
					Sampler	& TDE/it	& Samples/s	& Mean IAT	& MSS	& TDE/ES	& ES/s \\
					\midrule
					Plain   & 16.50  	& 21536		&  679.88	&  9.27	& 11218.52	&   31.68 \\
					Cen 	& 11.70  	& 20433		&  910.44	& 29.36	& 10655.52	&   22.44 \\
					Var 	& 15.95  	& 16305		&  509.41	& 10.32	&  8127.33	&   32.01 \\
					Cov 	& 14.82  	& 12993		&  202.59	& 13.31	&  3002.90	&   64.14 \\
					Cen+Var &  9.98  	& 22912		&  191.18	& 48.66	&  1908.14	&  119.85 \\
					Cen+Cov &  7.06  	& 24866		&    2.91	& 93.13	&    20.53	& 8555.48 \\
					\bottomrule
				\end{tabular}
			\end{sc}
		\end{small}
	\end{center}
	\vskip -0.1in
\end{table*}

We begin by demonstrating that the three different types of adjustments proposed in Section \ref{Sec:ATT}, i.e.~centering, variance adjustments and covariance adjustments, can indeed all be beneficial to a PATT sampler's performance. To this end, we conducted an experiment in which each reasonable subset of the adjustment types is used for a separate sampler and all of these samplers are made to compete in sampling from the same target distribution.

Specifically, we set the target distribution $\nu$ to be the multivariate $t$-distribution with $\gamma = 10$ degrees of freedom, location parameter $\tau = (\sqrt{d},\dots,\sqrt{d})^{\top} \in \R^d$ and scale matrix $\Pi \in \R^{d \times d}$ given by
\begin{align*}
	&\Pi := \Pi^{(1)} \Pi^{(2)} \Pi^{(1)} , \\
	&\Pi^{(1)} = \diag(\sqrt{1},\dots,\sqrt{d}) , \\
	&\Pi^{(2)}_{i,j}
	= \begin{cases}
		1 & i = j , \\
		0.5 & i \neq j ,
	\end{cases}
\end{align*}
in dimension $d = 100$. For the reader's convenience, we recall that this $\nu$ has density
\begin{equation*}
	\varrho(x)
	= \left(1 + \frac{1}{\gamma} (x - \tau)^{\top} \Pi^{-1} (x - \tau)\right)^{-(d + \gamma)/2}.	
\end{equation*}
Note that $\nu$ is fairly heavy-tailed and has mean $\tau$ and covariance $3 \cdot \Pi$. In particular, it is far from being centered and has both highly inconsistent coordinate variances and rather strong correlations between the variables, so that all adjustment types can in principle be expected to improve performance.

As outlined above, we ran six different samplers for this target, ``plain'' GPSS and five variants of PATT-GPSS (one for each individual adjustment type as well as two for the meaningful combinations two adjustment types). Each used $p = 10$ parallel chains and was run for $n_{\text{burn}} = 2 \cdot 10^4$ initialization burn-in iterations and $n_{\text{its}} = 10^5$ regular iterations after that. All samplers were initialized with the same values, which were independently (w.r.t.~the parallel chains) drawn from $\Nc_d(\fatzero, d^2 \cdot I_d)$.

To capture the samplers' long-term performances (rather than their performances in the stage where they are still making major adjustments to their adaptively chosen parameters), we only considered the latter half of each sampler's after-burn-in samples when computing the usual performance metrics. Moreover, in order to properly measure the samples' autocorrelation despite certain pathologies of GPSS in cases where it struggles with the target, we shifted all samples by the true target mean $\tau$ and took entry-wise absolute values of them before computing the mean IATs. The results are shown in Table \ref{Tab:ablation_adj_types}. It can be seen that all the PATT samplers exhibit improved performance compared to the base sampler in at least two of the four non-aggregate performance metrics (TDE/it, samples/s, mean IAT, MSS). Moreover, combining two adjustment types is seen to produce much larger benefits than using just one. Given that the target has a relatively challenging covariance structure, we find it unsurprising that the sampler using a combination of centering and covariance adjustments achieved by far the highest sample quality in this experiment. We provide a glimpse into the sampling behind the table's statistics in Appendix \ref{App:plots}, Figure \ref{Fig:ablation_adj_types_overview}.

\subsection{Parallelization and Update Schedules} \label{SubApp:par_US}

Next we demonstrate that our proposed \textit{entangled parallelization} (EP) approach (i.e.~the one outlined in Section \ref{SubSec:PATT}) outperforms the more obvious \textit{naive parallelization} (NP) of ATT chains (the latter being the scheme in which a number of parallel chains each apply ATT without sharing any information between them, thus maintaining independence from one another). Because update schedules (US) are necessitated in large part by the use of EP, we incorporate them into this ablation study. That is, we consider not only the effect of replacing EP by NP, but also that of removing the US from either sampler (which is taken to mean that the sampler must update its parameters in every iteration). All in all, we thus needed to run four different samplers: NP without US, NP with US, EP without US and EP with US. Since the samplers without US have far larger tuning overhead than those with US, the focus in the analysis of this experiment has to be on measuring the samplers' costs by their physical runtimes.

To eliminate the need for an initialization burn-in phase, we considered a tractable target distribution and initialized all samplers with an exact draw from it. Specifically, we chose the Gaussian distribution $\nu = \Nc_d(\tau, \Pi)$ with ${\tau = (1,\dots,d)^{\top} \in \R^d}$ and $\Pi = \diag(1^2,\dots,d^2)^{\top}$ in dimension $d = 100$. All four samplers were made to use ATT with centering (via sample means) and variance adjustments. As this meant that all of them would achieve very high sample quality if left running for long enough, we examined their performances based on somewhat shorter chains than usual: Each sampler was allowed to use $p = 10$ parallel chains and run for $n_{\text{its}} = 2 \cdot 10^4$ iterations.

\begin{table}[t]
	\caption{Sampling statistics for the ablation study on parallelization and update schedules.}
	\label{Tab:ablation_par_us}
	\vskip 0.1in
	\begin{center}
		\begin{small}
			\begin{sc}
				\begin{tabular}{crrrr}
					\toprule
					Sampler & TDE/it	& Samples/s	& Mean IAT & ES/s \\
					\midrule
					NP		& 9.02		& 30926		& 47.25    &   654 \\
					NP+US	& 8.48		& 50130		& 28.52    &  1757 \\
					EP		& 6.65		&  2078		&  1.13    &  1844 \\
					EP+US	& 7.11		& 41072		&  3.44    & 11948 \\
					\bottomrule
				\end{tabular}
			\end{sc}
		\end{small}
	\end{center}
	\vskip -0.2in
\end{table}

This time we considered all samples when computing the performance statistics. The results are shown in Table \ref{Tab:ablation_par_us}. One can see at a glance that EP+US (which corresponds to PATT as proposed in Sections \ref{Sec:ATT} and \ref{Sec:par_sched}) substantially outperforms the three other samplers in terms of ES/s.

To get a better understanding of how this result came about, we may take a closer look at the other columns: The samplers' TDE/it are all roughly the same, meaning that differences in runtime are mostly due to parameter update overhead. The amounts of samples each sampler generated per second of runtime are also roughly similar, except that EP without US was more than an order of magnitude slower than all other samplers. This is not particularly surprising, as that sampler needs to synchronize all its parallel chains in each iteration in order to perform its parameter update, and so ends up spending an exorbitant amount of time (in total) waiting for whichever is the slowest chain in any given iteration. It is also unsurprising that NP+US and EP+US are both faster than NP without US, given that the latter performs orders of magnitude more tuning parameter updates than the two former ones.

For the most part, the relative sizes of the mean IAT values are also as one might have expected: The two NP samplers produce more strongly correlated samples than the two EP samplers, because they incorporate far less information into the choice of each tuning parameter they use. EP without US generates the least correlated samples, because it constantly updates its tuning parameters with all information it has available (across chains). The only aspect of the IAT values that was perhaps unexpected is that NP without US produced more strongly correlated samples than NP+US, even though the former updates its tuning parameters much more frequently than the latter. Our best guess as to what caused this would be that the rapid updating of NP without US initially (when there was still very little sampling information in total and each new sample had a large impact on sample means etc.) caused its chains to behave a bit erratically, thus impeding their sample quality. Under that hypothesis, NP without US would already substantially benefit from an update schedule that simply postpones the first update to a point at which there is already a sizeable set of samples available.

\begin{table*}[t] 
	\caption{Sampling statistics for the ablation study on initialization burn-in periods. The mean IATs were computed w.r.t.~the component-wise absolute values of the samples generated in the latter half of the iterations.}
	\label{Tab:ablation_IBP}
	\vskip -0.1in
	\begin{center}
		\begin{small}
			\begin{sc}
				\begin{tabular}{lrrrrrr}
					\toprule
					Sampler		& TDE/it	& Samples/s	& Mean IAT	& MSS   & TDE/ES & ES/s \\
					\midrule
					without IBP & 11.66  	& 12631		& 77.20		&  2.85 & 899.82 & 163 \\
					with IBP 	&  6.77		& 19981		&  3.01		& 11.18 &  20.38 & 6640 \\
					\bottomrule
				\end{tabular}
			\end{sc}
		\end{small}
	\end{center}
	\vskip -0.1in
\end{table*}

We conclude the ablation study by emphasizing that various aspects of an ATT sampling setup are relevant in determining the merits of entangled parallelization and (any particular) update schedules in that specific setup: Firstly, one needs to take into account the types of adjustments that are to be used. For example, if one aims to use covariance adjustments, especially in moderate to high dimension, it is very much advisable to use a suitable update schedule to prevent the tuning parameter updates from producing an astronomically large runtime overhead (cf.~Section \ref{App:param_choices}). Secondly, one needs to factor in how computationally expensive it is to evaluate the target density. For instance, if each evaluation is extremely costly, the update overhead is much less relevant overall and it may be advisable too use a somewhat ``denser'' update schedule to improve sample quality (and ideally decrease TDE/it) at the cost of increased tuning overhead. Thirdly, it is also relevant how many parallel chains are to be used. If the number of chains is very large, e.g.~because the sampler is supposed to run on a large-scale cluster, it becomes increasingly harder to synchronize them for the purpose of a tuning parameter update in the EP framework. Of course it should still be feasible to employ EP in such situations, but they necessitate using a fairly ``sparse'' update schedule, i.e.~one that updates relatively rarely. Moreover, if such a sparse update schedule is used, the overall performance advantage of EP over NP should actually be very pronounced, because the former chooses its tuning parameters based on vastly more information than the latter, due to the large number of chains.

\subsection{Initialization Burn-In} \label{SubApp:IBP}

To demonstrate the potential advantage of using an initialization burn-in period (IBP), we fixed a simple (but somewhat challenging) target distribution and ran two versions of PATT-GPSS for it, one with an IBP and one without, both generating the same total amount of samples.

Specifically, we considered the Gaussian target distribution $\nu = \Nc_d(\tau, \Pi)$ with $\tau = (2 d,0,\dots,0)^{\top} \in \R^d$ and
\begin{equation*}
	\Pi_{i,j} 
	= \begin{cases}
		1 & i = j \\
		0.75 & i \neq j
	\end{cases}
\end{equation*}
in dimension $d = 100$.

Both samplers were instructed to use centering (based on sample means) and covariance adjustments. For both samplers we used $p = 10$ parallel chains and let them run for $n_{\text{its}} = 10^5$ iterations each. For the sampler with IBP, these $n_{\text{its}}$ iterations were divided into $n_{\text{burn}} = 10^4$ initialization burn-in iterations plus $n_{\text{att}} = 9 \cdot 10^4$ PATT iterations. Both samplers were initialized with the same values, which were independently (w.r.t.~the parallel chains) drawn from the multivariate standard normal distribution $\Nc_d(\fatzero, I_d)$.

The samplers' performances, in terms of the usual sampling statistics, are summarized in Table \ref{Tab:ablation_IBP}. Needless to say, the sampler with IBP has a clear advantage over the one without, outperforming it by roughly a factor of $40-50$ in the categories TDE/ES and ES/s. This is also evident when looking at trace plots of the samples (omitted here).

However, the reason for this performance difference is not obvious: If one merely examines the distances in Euclidean / Frobenius norm between the learned means / covariances and the respective ground truth values, one will observe little difference between the two samplers. Only closer examination of these estimates reveals the source of the performance difference: The sampler without IBP vastly overestimated the first coordinate variance throughout its run, still thinking it to be around $30$ (when the true value is $\Pi_{1,1} = 1$) after all $n_{\text{its}}$ iterations, see Figure \ref{Fig:ablation_IBP_cov_devs}. The cause of this is clear: It is a lingering effect of the bad initialization. Since the parallel chains are initialized close to the origin, and the target distribution's mean $\tau$ has first entry $\tau_1 = 2 d = 200$ far from zero, the early samples generated by either method (which record the chains' wandering from the initialization region closer to the mean), tend to have large variations in their first vector entry. Thus, considering these early samples in the computation of empirical covariances naturally leads to a significant overestimation of the corresponding coordinate variance.

\begin{figure}[tb]
	\begin{center}
		\includegraphics[width=0.35\textwidth]{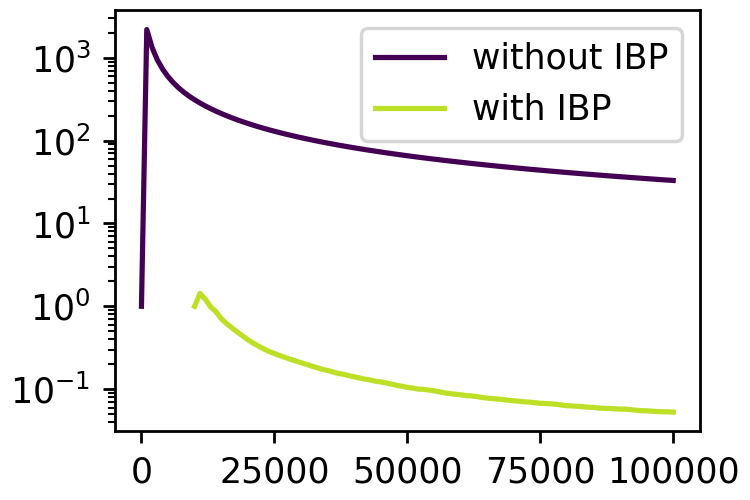}
	\end{center}
	\caption{Progression over iterations of the deviations between estimated and true first coordinate variance in the ablation study on initialization burn-in periods. The values on the $x$-axis mark the iteration indices at the update times, which for the sampler with IBP were made to include the initialization burn-in iterations. \label{Fig:ablation_IBP_cov_devs}}
	\vskip -0.1in
\end{figure}

The experiment should serve as a solid motivation for the use, in principle, of initialization burn-in periods. There are, of course, some caveats to this. Firstly, using an IBP only makes sense if the initial states are ``uninformed'', i.e.~if one is unsure how well they reflect the target. If the initial states are already representative of what samples from the target should look like, using an IBP may actually do more harm than good, as it can slow the rate (in terms of progress per total number of iterations, including IBP) at which the affine transformation's parameters approach their respective optimal values.

Secondly, even in cases of uninformed initialization, one should avoid choosing the length of the IBP much larger than it needs to be. Generally speaking, it is desirable for the IBP to last only as long as it takes the slowest chain to move from the initialization region to a region of high probability mass w.r.t.~the target. Of course it is impossible to know how long this takes beforehand, not least because it is random. It is therefore advisable to try and use an IBP length that, with high probability, is a relatively sharp upper bound for the length of the aforementioned process. If one simply uses a very long IBP, the benefit of excluding unrepresentative early samples is eventually overwhelmed by the detriment of using far fewer samples to learn the parameters of the affine transformation (assuming a fixed total number of iterations).

\onecolumn
\section{Additional Plots} \label{App:plots}

\begin{figure}[H]
	\begin{center}
		\includegraphics[width=\textwidth]{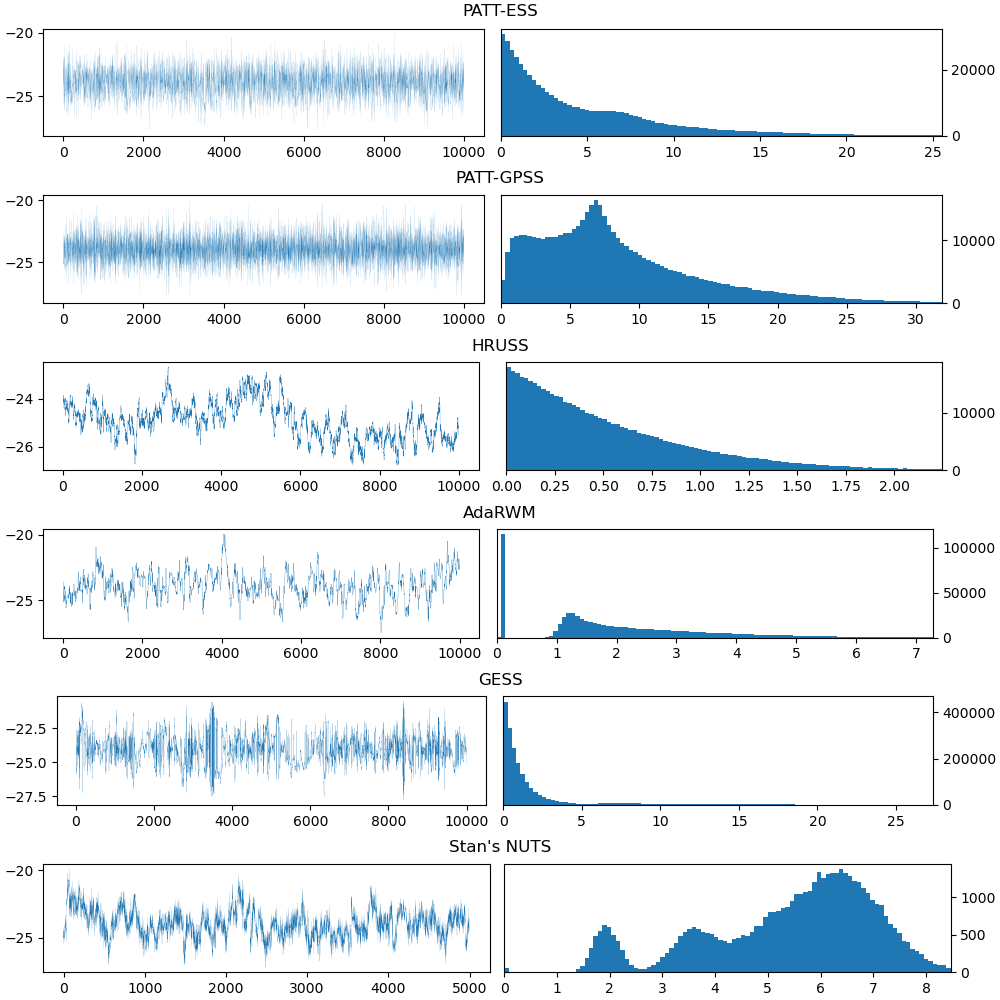}
	\end{center}
	\caption{Trace plots and step size histograms for the experiment on Bayesian inference with multivariate exponential distributions. The left column plots the progression of marginal samples for the 1st marginal (i.e.~the sequence of 1st vector entries) of the samples generated by each sampler's first chain in its final $n_{\text{window}} = 10^4$ iterations. The right column presents histograms of the step sizes (i.e.~Euclidean distances between consecutive samples) of each sampler during the latter half of its iterations. \label{Fig:multiv_exp_dists}}
\end{figure}

\phantom{} \\

\begin{figure}[H]
	\begin{center}
		\includegraphics[width=\textwidth]{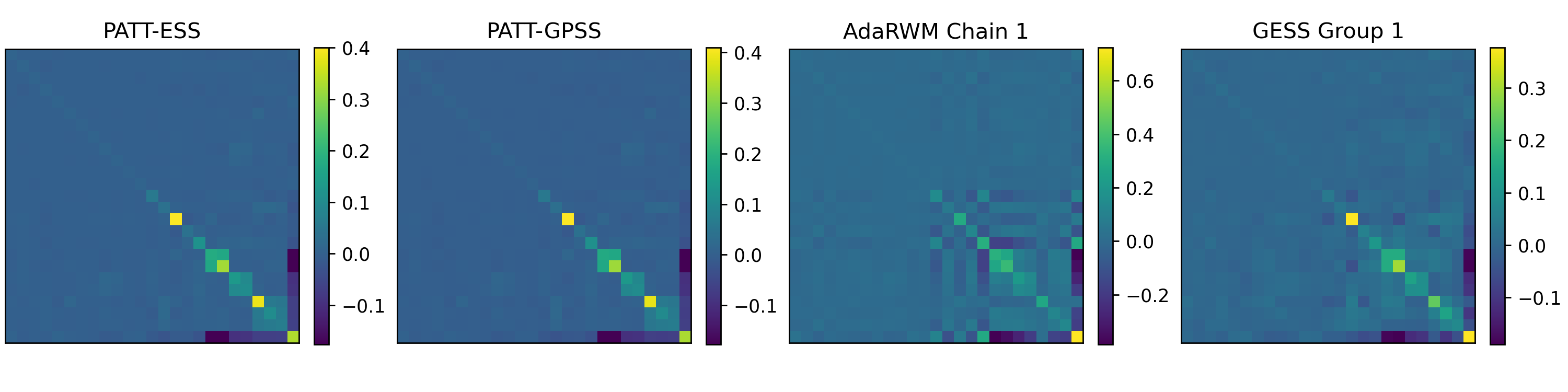}
	\end{center}
	\vspace{-0.75cm}
	\caption{Final covariance/scale matrices used by the adaptive samplers in the experiment on Bayesian logistic regression for German credit data. \label{Fig:BLR_German_credit_covs}}
\end{figure}

\begin{figure}[H]
	\begin{center}
		\includegraphics[width=\textwidth]{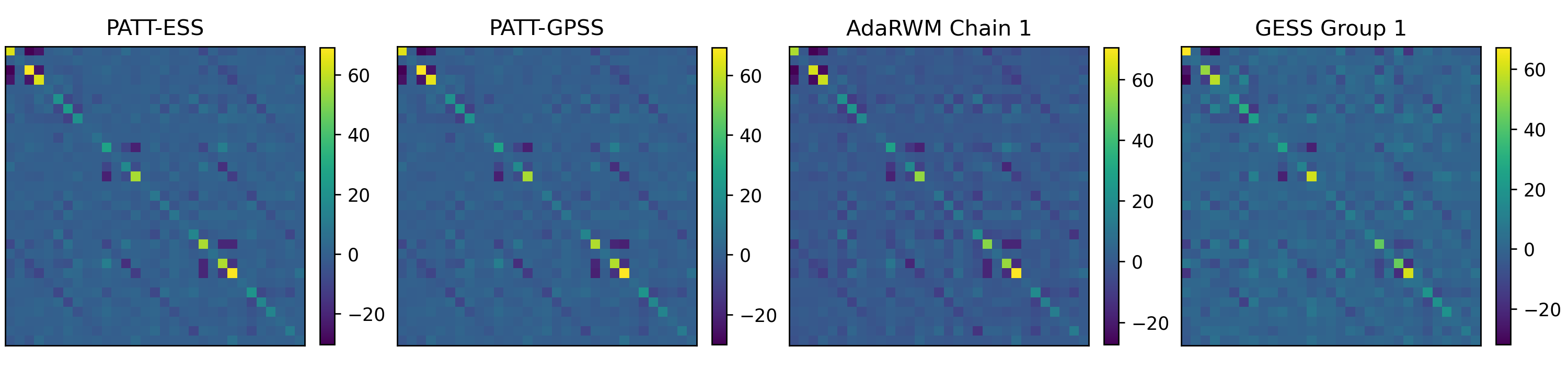}
	\end{center}
	\vspace{-0.75cm}
	\caption{Final covariance/scale matrices used by the adaptive samplers in the experiment on Bayesian logistic regression for breast cancer diagnostic data. \label{Fig:BLR_breast_cancer_covs}}
\end{figure}

\begin{figure}[H]
	\begin{center}
		\includegraphics[width=\textwidth]{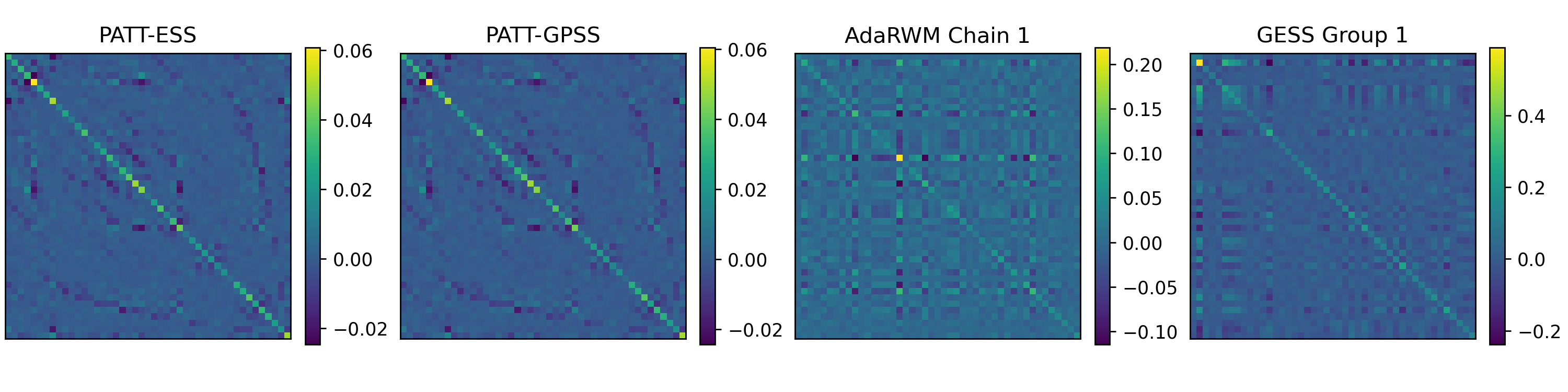}
	\end{center}
	\vspace{-0.75cm}
	\caption{Final covariance/scale matrices used by the adaptive samplers in the experiment on Bayesian logistic regression with feature engineering for Pima diabetes data. \label{Fig:BLR-FE_Pima_diabetes_covs}}
\end{figure}

\begin{figure}[H]
	\begin{center}
		\includegraphics[width=\textwidth]{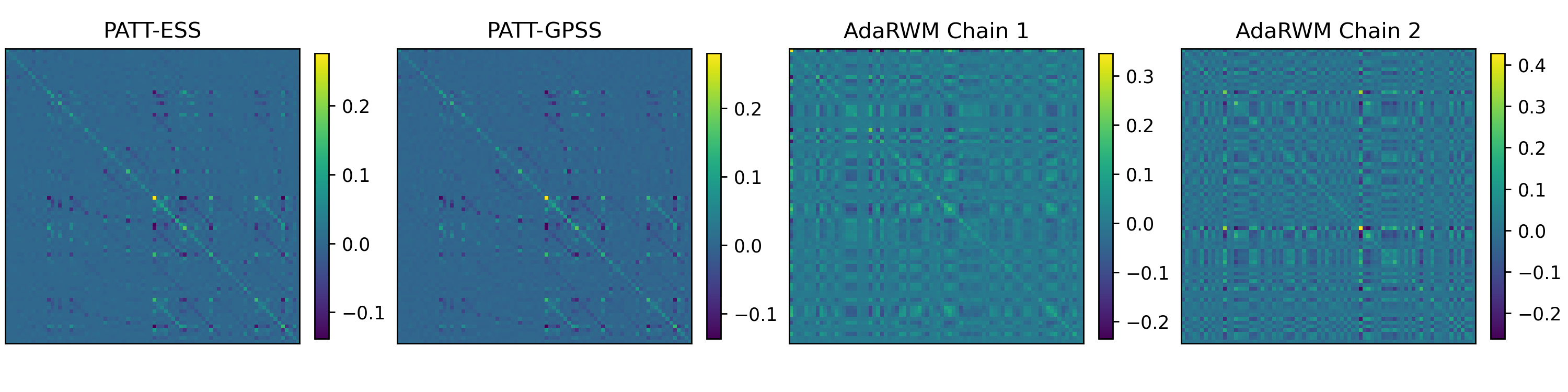}
	\end{center}
	\vspace{-0.75cm}
	\caption{Final covariance/scale matrices used by the adaptive samplers in the experiment on Bayesian logistic regression with feature engineering for wine quality data. \label{Fig:BLR-FE_wine_quality_covs}}
\end{figure}

\phantom{} \\

\begin{figure}[H]
	\begin{center}
		\includegraphics[width=\textwidth]{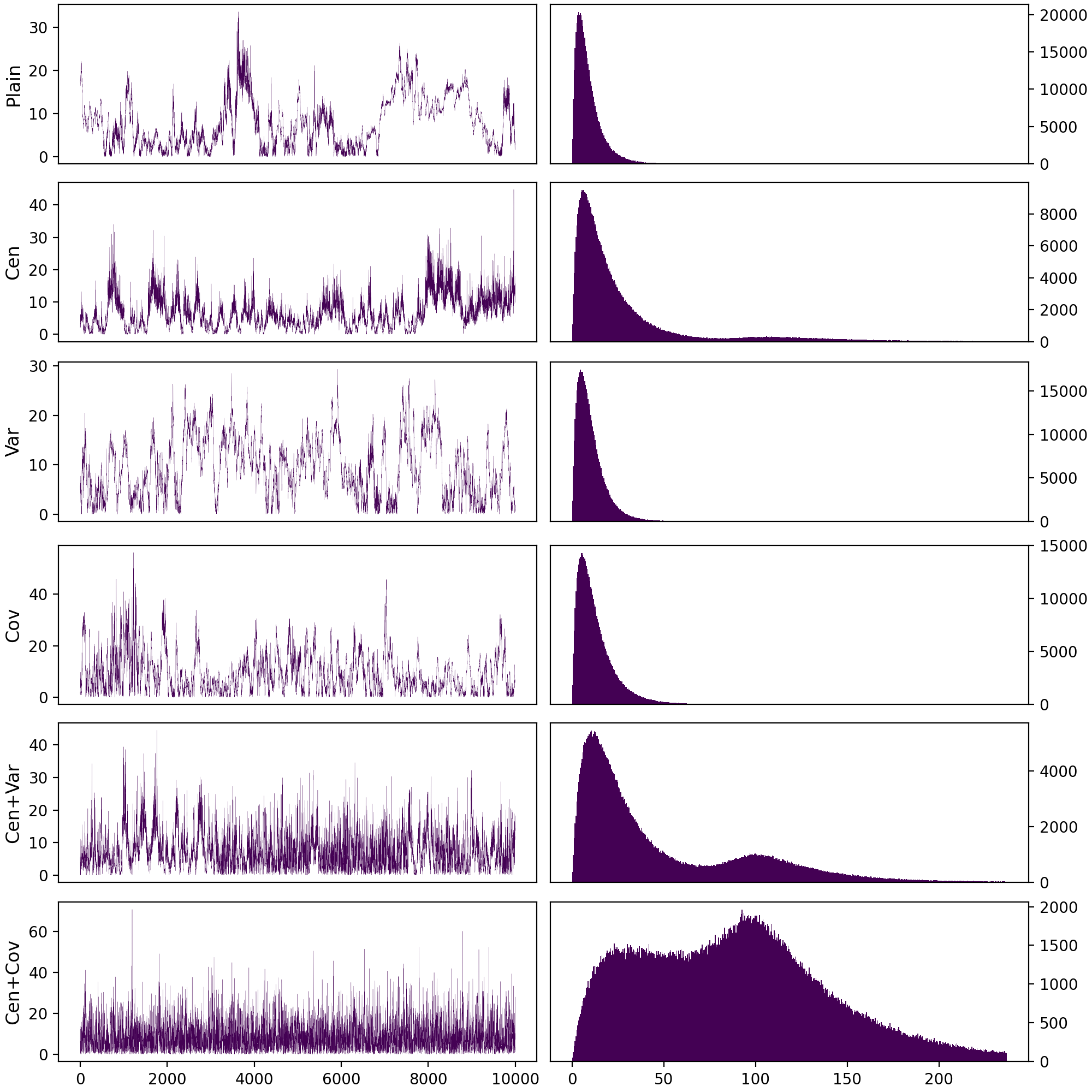}
	\end{center}
	\caption{Trace plots and step size histograms for the ablation study on adjustment types. The left column plots the progression of abs shifted marginal samples for the $d$-th marginal (i.e.~the absolute value of the last vector entry shifted by the last vector entry of the true target mean $\tau$) of the samples generated by each sampler's first chain in its final $n_{\text{window}} = 10^4$ iterations. The right column presents histograms of the step sizes (i.e.~Euclidean distances between consecutive samples) of each sampler during the latter half of sampling (excluding the initialization burn-in period). \label{Fig:ablation_adj_types_overview}}
\end{figure}

\end{document}